\documentclass[11pt]{article}
\pdfoutput=1
\usepackage[margin=1in]{geometry} 

\usepackage{etoolbox}

\usepackage[round]{natbib}
\usepackage{array,amsmath,amsthm,amssymb,amsfonts,titling,bbm, titlesec, bm, physics, enumitem, accents, setspace, fancyhdr, natbib, geometry, pdflscape}
\usepackage{graphicx}
\usepackage{float}
\usepackage[font=footnotesize,labelfont=bf]{caption}
\usepackage{array}
\usepackage[title]{appendix}
\usepackage{afterpage}
\usepackage{listings}
\usepackage{prodint}
\usepackage{algorithm}
\usepackage{comment}
 \usepackage{booktabs}
\usepackage{multirow,booktabs}
\usepackage{subcaption}
\newlength{\continueindent}
\usepackage{etoolbox}
\usepackage{setspace}
\usepackage[dvipsnames]{xcolor}

  \usepackage{tikz}
\usetikzlibrary{positioning}

\newcommand{\new}{\mathrm{new}}

\usepackage{adjustbox}

\definecolor{Bleu}{RGB}{0,0,204}
\usepackage[hypertexnames=false]{hyperref}
\hypersetup{
colorlinks,
  citecolor=Bleu,
  linkcolor=Bleu,
  urlcolor=Bleu}
\newcommand\independent{\protect\mathpalette{\protect\independenT}{\perp}}
\def\independenT#1#2{\mathrel{\rlap{$#1#2$}\mkern2mu{#1#2}}}

\newcolumntype{P}[1]{>{\centering\arraybackslash}p{#1}}

\usepackage{amsmath}
\usepackage{amssymb}
\usepackage{amsthm,thmtools}
\usepackage{enumitem}
\usepackage{amsfonts}
\RequirePackage{algorithm}
\RequirePackage{algorithmic}
\usepackage{bm,upgreek}
\usepackage{mathrsfs}  
\usepackage{calc}
\usepackage{subcaption}
\usepackage{authblk}

\theoremstyle{plain}
\newtheorem{theorem}{Theorem}
\newtheorem{lemma}{Lemma}
\newtheorem{corollary}{Corollary}

\theoremstyle{definition}
\newtheorem{definition}{Definition}

\newtheorem{example}{Example}

\theoremstyle{remark}

\newtheorem*{remark*}{Remark}  
\newtheorem{assumption}{Assumption}

\theoremstyle{definition}
\newtheorem*{definition*}{Definition}

\providecommand{\keywords}[1]
{
  \small	
  \textbf{\textit{Keywords:}} #1
}
\usepackage[utf8]{inputenc}
\usepackage{textcomp}

\DeclareMathOperator*{\argmin}{argmin}

\usepackage[capitalize,noabbrev]{cleveref}

\usepackage{authblk}

\DeclareFontFamily{U}{jkpmia}{}
\DeclareFontShape{U}{jkpmia}{m}{it}{<->s*jkpmia}{}
\DeclareFontShape{U}{jkpmia}{bx}{it}{<->s*jkpbmia}{}
\DeclareMathAlphabet{\mathfrak}{U}{jkpmia}{m}{it}
\SetMathAlphabet{\mathfrak}{bold}{U}{jkpmia}{bx}{it}

\usepackage[hang,flushmargin]{footmisc}

\title{Nonparametric Instrumental Variable Inference\\with Many Weak Instruments}
\author[1,2]{Lars van der Laan}
\author[1,3]{Nathan Kallus}
\author[1]{Aurélien Bibaut\thanks{Corresponding authors: lvdlaan@uw.edu; aurelien.bibaut@gmail.com}}

\affil[1]{Netflix Research, USA}
\affil[2]{Department of Statistics, University of Washington, USA}
\affil[3]{Cornell Tech, Cornell University, USA}

 \usepackage{setspace}
\newcommand{\slightspacing}{\setstretch{1.175}}

\begin{document}

\slightspacing
\maketitle

\begin{abstract}
\singlespacing

We study inference on continuous linear functionals in the nonparametric
instrumental variable (NPIV) problem with a discretely valued instrument under a
many--weak--instruments asymptotic regime, where the number of instrument levels
grows with the sample size. A key motivating example is the estimation of
long-term causal effects in a new experiment with only short-term outcomes,
using past experiments as instruments for the effect of short- on long-term
outcomes. Here, assignment to a past experiment serves as the instrument: there
are many past experiments but only a limited number of units in each. Because
the structural regression function is nonparametric but constrained by finitely
many moment restrictions, point identification typically fails. We address three challenges that arise in NPIV models with many weak
instruments: nonidentification of the linear functional, asymptotic bias when
the number of instruments is comparable to the sample size, and the absence of
an efficiency theory for this asymptotic regime. To address nonidentification,
we study linear functionals of the minimum-norm solution to the moment
restrictions, which are always identified. As the number of instrument levels
increases, these functionals provide a second-order approximation to the target
functional, allowing for valid inference on the true target under suitable
conditions. We propose debiased machine learning estimators based on the
approximating functional that incorporate a jackknife correction to remove
bias arising in many--weak--instrument settings. To mitigate bias in nuisance
estimation, we extend the Jackknife Instrumental Variable Estimator (JIVE)
beyond the classical parametric setting to nonparametric models implemented via
Tikhonov-regularized empirical risk minimization. Finally, we develop a
nonparametric efficiency theory for regular estimators under
many--weak--instrument asymptotics.

\end{abstract}

\slightspacing
 
\keywords{nonparametric instrumental variables; weak instruments; discrete instruments; many-instrument asymptotics; ill-posed inverse problems; jackknife; debiased machine learning; semiparametric efficiency}

\section{Introduction}

Estimating causal effects from observational data is challenging due to unmeasured confounding. Instrumental variable (IV) methods address this issue by leveraging variables that influence treatment assignment but are independent of unobserved confounders and affect the outcome only through the treatment~\citep{wald1940fitting, reiersol1941confluence, angrist1996identification}. Within this framework, causal effects can be represented as functionals of a structural regression function that encodes the treatment--outcome relationship. Under a parametric linear structural model, both estimation and inference can be performed using two-stage least squares (2SLS)~\citep{basmann1957generalized, sargan1958estimation, wooldridge2010econometric}.

Nonparametric instrumental variable (NPIV) methods seek to recover the structural
function without imposing linearity assumptions by solving conditional moment
restrictions (CMRs) defined by the instrument~\citep{ai2003efficient,
newey2003instrumental, hall2005nonparametric, darolles2011nonparametric,
horowitz2011applied, ai2012semiparametric, newey2013nonparametric}. Existing approaches include minimum-norm series and sieve methods, which extend
classical 2SLS to nonparametric settings using basis
expansions of increasing complexity~\citep{newey2003instrumental, chen2007sieve,
chen2013optimal, chen2016methods, chetverikov2017nonparametric}, and spectral methods, which
regularize the ill-posed inverse problem using Tikhonov regularization or
spectral decompositions of the conditional expectation
operator~\citep{hall2005nonparametric, blundell2007semi, carrasco2007linear,
chen2011rate, darolles2011nonparametric, engl2015regularization,
meunier2025demystifying}. NPIV estimation can also be formulated as a min--max
optimization problem over an adversarial function class
\citep{lewis2018adversarial, dikkala2020minimax, daskalakis2021complexity,
bennett2023minimax}. This formulation enables the use of flexible
machine-learning algorithms for estimation and facilitates nonparametric
inference through debiased machine-learning techniques~\citep{bennett2023inference,
bennett2023source}.

However, nonparametric identification of the structural function often fails,  
particularly when the treatment is continuous and the instrument is discrete  
\citep{carrasco2007linear, darolles2011nonparametric, newey2013nonparametric, chen2014local, song2024instrumental}.  
In such settings, the structural function is infinite-dimensional but constrained  
by only finitely many moment conditions \citep{florens2012instrumental, newey2013nonparametric}. Weak instruments---those that exert only limited influence on treatment assignment--- 
further exacerbate this challenge \citep{hahn2003weak}.  
Recent work has shown that imposing monotone shape constraints can yield unique
solutions in NPIV models estimated via sieves~\citep{scaillet2016ill,
chetverikov2017nonparametric}. These results, however, are established only for
settings with a univariate, continuous treatment and a univariate, continuous
instrument satisfying a monotone IV assumption, together with monotonicity of
the structural function.

These difficulties can be mitigated by focusing on low-dimensional functionals of the structural function, such as average causal effects, which may remain identifiable under weaker assumptions~\citep{severini2006some, blundell2007semi, bennett2023inference, bennett2023source}. For continuous linear functionals, point identification holds when the functional is invariant over the solution set of the conditional moment restriction, or equivalently, when the corresponding Riesz representer lies in the range of the adjoint operator defining the restriction \citep{severini2006some, carrasco2007linear, darolles2011nonparametric, chen2014local}. However, this condition is restrictive and often violated in nonparametric settings: with a discrete instrument, it requires the Riesz representer to lie in a finite-dimensional space.

These challenges are of practical relevance, as discrete instruments---whether categorical, ordinal, or high-dimensional binary---are pervasive across biomedical, econometric, and digital experimentation contexts.  In biomedical research, Mendelian randomization leverages genetic variants---typically encoded as large sets of discrete indicators---as instruments for modifiable risk factors~\citep{smith2004mendelian, vanderweele2014methodological, bowden2015mendelian, emdin2017mendelian, burgess2017interpreting}. In digital experimentation, past A/B test assignments can serve as high-dimensional categorical instruments for estimating treatment effects in observational analyses~\citep{peysakhovich2018learning, bibaut2024nonparametric}. In econometrics, classical designs use discrete instruments such as draft lottery numbers~\citep{angrist1990lifetime}, time and place of birth~\citep{angrist1991does}, or random assignment to judges~\citep{kling2006incarceration, dobbie2018effects}, all of which yield instruments with many levels~\citep{angrist1999jackknife, mikusheva2022inference, angrist2022machine}.

In applied settings, researchers increasingly confront identification challenges
by leveraging many weak instruments~\citep{
staiger1994instrumental, stock2005asymptotic, bekker2005instrumental,
chao2005consistent, hansen2008estimation, bowden2015mendelian,
chernozhukov2015post, breunig2016adaptive, burgess2017interpreting,
ye2024genius}. The key intuition is that as the number of instruments grows, the
collective first-stage relationship between instruments and treatment may
strengthen, making identification of the structural function---or of specific
functionals thereof---more plausible even when individual instruments are weak.
Such designs arise naturally in modern applications, including digital
experimentation~\citep{peysakhovich2018learning} and Mendelian
randomization~\citep{burgess2017interpreting}, where the number of instruments
can exceed the sample size. This perspective is particularly relevant for
discrete instruments: with finitely many instrument levels, point identification
typically fails, yet the identification gap may shrink to zero as the number of
levels grows.

An increasing number of instruments introduces additional challenges for estimation.  In such settings, the sampling behavior of estimators is best understood under a \emph{many-weak-instrument asymptotic} regime, in which the number of instruments grows with the sample size~\citep{staiger1994instrumental, stock2005asymptotic, bekker2005instrumental, chao2005consistent}.  
In this regime, classical large-sample approximations may fail, and estimators can exhibit persistent bias~\citep{nelson1988distribution, staiger1994instrumental, angrist1999jackknife, stock2002survey, stock2005asymptotic, hansen2008estimation, angrist2009mostly, kolesar2015identification, peysakhovich2018learning}.  
This phenomenon is well documented for 2SLS in linear IV models, motivating the development of bias-corrected estimators such as Jackknife Instrumental Variable Estimation (JIVE)~\citep{angrist1999jackknife, hansen2014instrumental}.  
However, analogous bias-correction strategies for nonparametric IV estimation remain largely unexplored. The nonparametric setting poses further challenges, as bias can arise from both point estimation and high-dimensional nuisance estimation, for example in Tikhonov-regularized empirical risk minimization.

Addressing the identification and estimation challenges of NPIV with many weak instruments is particularly relevant for long-term causal inference in digital experimentation~\citep{peysakhovich2018learning, athey2019surrogate, duan2021online, zhang2023evaluating, bibaut2024nonparametric}, which serves as a running example throughout the paper.  
In digital settings, rapid innovation produces a large and growing number of short-term randomized experiments. Because long-term outcomes are rarely observed, researchers often rely on short-term surrogates to infer long-term effects.  
When a valid surrogate index preserves the causal effect on long-term outcomes, new interventions can be evaluated before long-term data become available \citep{athey2019surrogate}.  
In practice, however, surrogate indices are typically learned from observational data and may be biased by confounding between the surrogate and the long-term outcome.  
Unmeasured confounding can cause surrogate-based estimates to overstate---or even reverse---true long-term effects, a phenomenon known as the surrogate paradox \citep{elliott2015surrogacy}.  
To address this issue, \citet{peysakhovich2018learning} propose a linear IV approach that uses assignment to past experimental arms as instruments for the causal effect of the surrogate on the long-term outcome.

 \begin{figure}[htb!]
\centering
\begin{minipage}{0.78\linewidth}
\centering

\begin{tikzpicture}[
    exp/.style={draw, minimum width=0.9cm, minimum height=0.6cm},
    bigbox/.style={draw, rounded corners, inner sep=0.20cm},
    node distance=2.0cm,
    >=stealth
]

\node[bigbox] (hist) {
    \begin{tikzpicture}[x=1cm,y=1cm]
        \foreach \i in {0,...,3} {
          \foreach \j in {0,...,2} {
            \node[exp] at (0.95*\i, -0.7*\j) {};
          }
        }
    \end{tikzpicture}
};

\node[above=0.18cm of hist] {\small Historical experiments $Z$};

\node[
  draw,
  rounded corners,
  minimum width=2.6cm,
  minimum height=0.95cm,
  right=1.6cm of hist  
] (S) {\begin{tabular}{c} Surrogate $X$ \end{tabular}};

\node[
  draw,
  rounded corners,
  minimum width=3.0cm,
  minimum height=0.95cm,
  right=1.8cm of S  
] (Y) {Outcome $Y$};

\draw[->, thick] 
  (hist.east) -- 
  node[above=0.05cm, midway] {} 
  (S.west);

\draw[->, thick] 
  (S.east) -- 
  node[above=0.05cm, midway] {} 
  (Y.west);

\end{tikzpicture}

\caption{
Past experiments (cells in the rectangle) form a high-dimensional instrument
vector~$Z$ identifying the causal effect of the surrogate~$X$ on the long-term
outcome~$Y$ using observational data. In a new short-term experiment with
treatment~$A$, we observe $X$ and estimate the mapping $A \rightarrow X$ to
recover the long-term causal effect.
}
\label{fig:npiv_surrogate_schematic}

\end{minipage}
\end{figure}

Extending NPIV to accommodate many weak discrete instruments provides a principled framework for constructing confounding-robust surrogate indices without imposing restrictive modeling assumptions.  
Here, the instrument is encoded as a high-dimensional indicator vector denoting membership in each past experimental arm (Figure~\ref{fig:npiv_surrogate_schematic}).  
Each experiment contributes only a finite number of observations, which may remain fixed as the total number of experiments grows.  
Thus, the relevant asymptotic regime is one in which the number of experiments increases while the sample size per experiment remains fixed or grows slowly. In this regime, the long-term effect in a new short-term experiment is point identified if its surrogate distribution can be expressed as a linear combination of those from past experiments.  
Intuitively, a \emph{synthetic experiment} can be formed from historical experiments that reproduces the surrogate distribution of the new experiment, thereby enabling identification of its long-term effect---an idea reminiscent of synthetic controls~\citep{abadie2015comparative}.  
With finitely many historical studies this condition typically fails, but it may hold approximately as the number and diversity of experiments increase.  
This motivates the need for a theoretical framework that addresses (i) identification that emerges only as the number of experiments increases, (ii) bias arising from limited data within each experiment, and (iii) an efficiency theory for optimal estimation in this many--weak--instrument NPIV setting.

\subsection{Contributions of this work}

We develop a framework for identification and inference on continuous linear
functionals in NPIV models with many weak discrete instruments. The central
challenge we address is how to conduct valid inference when the structural
function is not point identified, the number of instrument levels grows with the
sample size, and classical NPIV conditions---such as negligible many-instrument
bias and fixed-distribution efficiency theory---no longer apply. Our framework
addresses these challenges by constructing debiased estimators suited to
many-instrument designs and introducing a nonparametric JIVE procedure to
estimate the required nuisance components under limited cell sizes. We then
develop an asymptotic and efficiency theory for these estimators in the
many--weak--instrument regime.

\medskip
\noindent\textbf{Approximating functional, identification gap, and debiased machine learning.}
Our framework starts by formalizing the appropriate target of inference when
point identification fails. We work with the continuous linear functional
evaluated at the minimum-norm solution of the conditional moment restriction,
which is always well defined in our setting. Sieve approximation theory implies
that the resulting identification gap for the target functional is second-order
and may vanish as the number of instrument levels increases
\citep{carrasco2007linear, babii2020completeness}.

Building on this formulation, we construct asymptotically linear, debiased
machine-learning estimators based on flexible, data-adaptive nuisance
estimation \citep{vanderLaanRose2011, DoubleML}. When the identification gap
decreases sufficiently quickly, the resulting estimator is consistent and
asymptotically normal for the true functional; otherwise, it yields valid
inference for the approximating functional. The construction is fully automatic,
in that it applies to any continuous linear functional without requiring
functional-specific adjustments \citep{chernozhukov2022automatic, van2025automatic2}.

\medskip
\noindent\textbf{Jackknife debiased estimation with many weak instruments.}
We implement this debiased estimator using a jackknife correction tailored to
the many--weak--instrument setting. Standard approaches can suffer from
persistent bias when empirical averages within instrument levels are based on
few observations, affecting both the plug-in component and nuisance estimation.
Our \emph{DML--JIVE estimator} addresses this issue by incorporating a
jackknife correction that removes the leading-order many-instrument bias,
yielding asymptotic normality under the mild growth condition
\(N/K \gg \log K\), where \(N\) is the total sample size and \(K\) the number of
instrument levels. In contrast, without the jackknife correction one typically
requires the substantially stronger condition \(N/K \gg K\). When \(N/K\)
remains bounded, a single-split variant remains consistent and asymptotically
normal, albeit with reduced efficiency.

\medskip
\noindent\textbf{Nonparametric JIVE for nuisance estimation.}
To mitigate bias in nuisance estimation, we introduce \emph{nonparametric
Jackknife Instrumental Variable Estimation} (npJIVE), a nonparametric extension
of the classical JIVE method to IV models. npJIVE uses cross-fold sample
splitting to construct a jackknife-corrected estimator of the population risk
associated with the conditional moment restriction. We establish consistency and
convergence rates for Tikhonov-regularized empirical risk minimization under
this design, expressed in terms of the metric entropy of the underlying
function class. These results show that, without additional spectral smoothness
assumptions, many weak instruments reduce the effective sample size from \(N\)
to \(N/\sqrt{K}\) for nonparametric nuisance estimation. Under standard source
conditions imposing mild spectral smoothness, npJIVE can recover full-sample
convergence rates for a broad range of growth rates in \(K\). To our knowledge,
such convergence guarantees for nonparametric empirical risk minimization under
many--weak--instrument asymptotics have not previously been established.

\medskip
\noindent\textbf{Efficiency theory with many weak instruments.}
We develop a semiparametric efficiency theory for NPIV models with many weak
discrete instruments, formulated as a triangular array in which both the
sampling distribution and the target parameter vary with \(K\), and the Fisher
information for the approximating functional may diverge as \(K\) grows,
leading to convergence rates slower than \(N^{1/2}\). For linear IV models,
\citet{kaji2021theory} introduced a related framework for weak identification,
but their results are limited to specific parametric functionals and do not
extend to the nonparametric, many--weak--instrument regime considered here. Our
framework introduces a relaxed notion of regularity and establishes local
asymptotic normality of the likelihood ratio process, yielding a convolution
theorem and efficiency bounds for regular estimators in both root-\(N\) and
slower regimes. For the functionals considered here, surjectivity (but not
injectivity) of the conditional mean operator ensures pathwise differentiability
and sharp lower bounds on attainable performance, even under non-identification.
Relaxing surjectivity without additional structure appears difficult and is
closely related to local overidentification in the sense of
\citet{chen2018overidentification}.

  \medskip

The remainder of the paper is organized as follows.
Section~\ref{sec:setup} introduces the discrete NPIV problem and its identification challenges.
Section~\ref{sec::identification} develops the approximating functional, the identification gap, and conditions for asymptotic identification.
Section~\ref{sec::estimatorsall} presents our debiased machine learning estimators, and Section~\ref{sec::theory} establishes their asymptotic properties.
Section~\ref{sec::nuisance} introduces the npJIVE procedure for nuisance estimation.
Section~\ref{sec:efficiency} develops the efficiency theory for many--weak--instrument asymptotics.

\section{Problem set-up}\label{sec:setup}
\subsection{Data structure and causal parameter}

We consider data \( O := (Z, X, Y) \in \mathcal{O} \), where \( Z \in [K] \) is a categorical instrument with \( K < \infty \) levels, \( X \in \mathcal{X} \subset \mathbb{R}^d \) is a treatment variable, and \( Y \in \mathbb{R} \) is a real-valued outcome. We adopt an instrument-fixed design in which \(Z\) is treated as deterministic, and randomness arises solely from the distribution of \(X\) and \(Y\) conditional on \(Z\). We observe \( N := N(K) = \sum_{k=1}^K n_k \) units, denoted by \( O^{(K)} := \{O_{ki} : 1 \leq k \leq K,\; 1 \leq i \leq n_k\} \), across the \(K\) levels of \(Z\), drawn jointly from a distribution \( P^{(K)} \) on \( \mathcal{O}^N \) with \( n_k \) units at level \(k\). Each unit, indexed by \((k,i)\), yields an observation \( O_{ki} = (Z_{ki}, X_{ki}, Y_{ki}) \) where \( Z_{ki} := k \) is fixed. Under this fixed-instrument design, expectations such as \( E_{P^{(K)}}[Y_{ki}] \) and \( E_{P^{(K)}}[h(X_{ki})] \) correspond to conditional expectations \( E[Y \mid Z = k] \) and \( E[h(X) \mid Z = k] \). We assume only that \( P^{(K)} \) belongs to a nonparametric model \( \mathcal{P}^{(K)} \) containing all distributions satisfying \( Z_{ki} = k \) for all \( k,i \), with \( (X_{ki}, Y_{ki}) \) conditionally i.i.d.\ given \( Z_{ki} \) across \(k\) and \(i\). We consider a many-weak-instruments asymptotic regime where the number of instrument levels grows (\( K \to \infty \)) with the sample size (\( N \to \infty \)), while the number of observations per level, \(n_k\), may remain bounded. For notational simplicity, we write \( S_K \) for any population quantity \( S_{P^{(K)}} \), for example \( E_K \) for \( E_{P^{(K)}} \).

We describe the data-generating process using a nonparametric structural equation model (NPSEM) \citep{pearl2009causality}. For each \(k = 1, \ldots, K\) and \(i = 1, \ldots, n_k\), let latent variables \(U_{ki} = (U_{X,ki}, U_{Y,ki}, U_{C,ki})\) be drawn independently from a distribution \(P_U\) with mutually independent components. The observed data are generated by fixed functions \(f_X\) and \(f_Y\):
\[
Z_{ki} = k, \quad 
X_{ki} = f_X(Z_{ki}, U_{X,ki}, U_{C,ki}), \quad 
Y_{ki} = f_Y(X_{ki}, U_{Y,ki}, U_{C,ki}).
\]
Here \(U_{C,ki}\) captures unmeasured confounding between \(X_{ki}\) and \(Y_{ki}\), while \(Z_{ki}\) affects \(Y_{ki}\) only through \(X_{ki}\) and is independent of \(U_{C,ki}\). Under this model, the potential outcome under intervention \(X = x\) is \(f_Y(x, U_{Y,ki}, U_{C,ki})\), and the \emph{structural regression function}
\begin{equation}
h^\star(x) := E_{P_U}[f_Y(x, U_Y, U_C)]
\end{equation}
gives the counterfactual mean of \(Y\) under \(X = x\).

We aim to infer a continuous linear functional \( \psi(h^\star) \in \mathbb{R} \) of the structural function \(h^\star\), with continuity defined formally later. Because of unmeasured confounding through \(U_{C,ki}\), \(h^\star\) generally differs from the observed regression \(x \mapsto E_{P^{(K)}}[Y_{ki} \mid X_{ki} = x]\). Under these assumptions, the outcome satisfies
\(Y_{ki} = h^\star(X_{ki}) + \varepsilon_{ki}\)
with \(E_K[\varepsilon_{ki} \mid Z_{ki}] = 0\),
implying the conditional moment restrictions (CMR) of \citet{newey2003instrumental, ai2003efficient}:
\begin{equation}
E_K[Y_k - h^\star(X_k) \mid Z_k = k] = 0, \quad k \in [K].
\end{equation}
Thus, \(h^\star\) is set-identified as a solution to the ill-posed linear inverse problem \(E_K[h^\star(X_k)] = E_K[Y_k]\) for each instrument level \(k\). A notable example of a linear functional is the counterfactual mean of \(Y\) under a dynamic or stochastic intervention on \(X\), identified when confounding enters additively in the outcome model.

\setcounter{example}{0}
\begin{example}[Means of counterfactual treatments]
Given a counterfactual treatment distribution \(P_X^\star\), generate counterfactual data \(O^\star = (X^\star, Y^\star) \sim P^{\star}\) by intervening on the treatment equation in the NPSEM:
\[
Y^\star = f_Y(X^\star, U_Y, U_C), \quad X^\star \sim P_X^\star, \quad U \sim P_U.
\]
Suppose the outcome is additively generated as 
\(Y_{ki} = f_{XY}(X_{ki}, U_{Y,ki}) + f_{CY}(U_{Y,ki}, U_{C,ki})\)
for structural functions \(f_{XY}\) and \(f_{CY}\). Then $h^\star(x) = \int \!\left(f_{XY}(x, u_Y) + f_{CY}(u_Y, u_C)\right)\, dP_U(u),$
and the counterfactual mean outcome is identified as 
\(\psi(h^\star) := \int h^\star(x)\, dP_X^\star(x)\).
\end{example}

The discrete NPIV problem is particularly relevant in digital experimentation, where it enables confounding-robust inference on long-term causal effects \citep{peysakhovich2018learning, bibaut2024nonparametric}. Here, one typically observes short-term experimental data on a surrogate outcome together with historical data linking the surrogate to the long-term outcome \citep{athey2019surrogate}. Additional background is provided in Appendix~\ref{sec: surrogates}, and an earlier version of this application appeared in our technical report~\citep{bibaut2024nonparametric}.

\setcounter{example}{1}
\begin{example}[Confounding-robust surrogate indices]
Let \(Y\) denote the long-term outcome and \(X\) a vector of short-term surrogates. The goal is to estimate the intervention effect on \(Y\) in a new experiment where only \(X\) is observed. A large number of historical randomized experiments (\(K\)) serve as instruments for \(X\), where the discrete instrument \(Z\) is a binary vector indicating membership in an experimental arm \citep{peysakhovich2018learning}. These instruments permit confounding-robust estimation of the surrogate effect \(h^\star(X)\) from observational data. Under the additive-error structure of Example 1, \(h^\star(X)\) acts as a surrogate index for \(Y\) in the new experiment: $E_{P^\star}[Y^\star] = E_{P^\star}[h^\star(X^\star)],$
where \((X^\star, Y^\star) \sim P^\star\) denotes the joint distribution under the new intervention.
\end{example}


\subsection{Notation and identification challenges}

 We introduce notation and review key challenges in identifying \(\psi(h^\star)\). Define the Hilbert spaces
\( L^2_K(X) := \{ h : \mathcal{X} \to \mathbb{R} : \sum_{k=1}^K \tfrac{n_k}{N} E_K[h^2(X_k)] < \infty \} \)
and
\( L^2_K(Z) := \{ q : [K] \to \mathbb{R} : \sum_{k=1}^K \tfrac{n_k}{N} q^2(k) < \infty \} \),
with inner products
\(\langle h_1, h_2 \rangle_{L^2_K(X)} := \sum_{k=1}^K \tfrac{n_k}{N} E_K[h_1(X_k)h_2(X_k)]\)
and
\(\langle q_1, q_2 \rangle_{L^2_K(Z)} := \sum_{k=1}^K \tfrac{n_k}{N} q_1(k)q_2(k)\). Let \(\mathcal{T}_K : L^2_K(X) \to L^2_K(Z)\) and its adjoint 
\(\mathcal{T}_K^* : L^2_K(Z) \to L^2_K(X)\) be defined by
\[
(\mathcal{T}_K h)(k) = E_K[h(X_k)], \qquad
(\mathcal{T}_K^* q)(x)
  = \frac{\sum_{k=1}^K \frac{n_k}{N} q(k)\, p^{(K)}(X_k = x)}
         {\sum_{k=1}^K \frac{n_k}{N} p^{(K)}(X_k = x)}.
\]
Under an instrument-randomized design, we would have $(\mathcal{T}_K h)(k) = \mathbb{E}[h(X) \mid Z = k]$ and \(\mathcal{T}_K^* q(x) = \mathbb{E}[q(Z) \mid X = x]\). Let \(\mathcal{H}_K := \mathcal{R}(\mathcal{T}_K^*) \subseteq L^2_K(X)\) denote the range of \(\mathcal{T}_K^*\),
and \(\Pi_K : L^2_K(X) \to \mathcal{H}_K\) the orthogonal projection onto this subspace. Because both \(\mathcal{R}(\mathcal{T}_K)\) and \(\mathcal{R}(\mathcal{T}_K^*)\) are finite-dimensional, they are closed.
By the rank--nullity theorem,
\(\mathcal{H}_K = \mathcal{N}(\mathcal{T}_K)^\perp = \mathcal{R}(\mathcal{T}_K^*)\),
where \(\mathcal{N}(\mathcal{T}_K) := \{ h \in L^2_K(X) : \mathcal{T}_K(h) = 0 \}\) is the null space of \(\mathcal{T}_K\).

The structural function \(h^\star\) is set-identified as a solution to the
inverse problem \(\mathcal{T}_K(h^\star) = \mu_K\), where
\(\mu_K(k) := E_K[Y_k]\). Thus \(h^\star\) is identified only up to the null
space of \(\mathcal{T}_K\):
\[
h^\star \in \mathcal{T}_K^{-1}(\mu_K)
= \{\, h^\star + h : h \in \mathcal{N}(\mathcal{T}_K) \,\}.
\]
A standard condition ensuring uniqueness is \emph{completeness}, which requires
\(\mathcal{T}_K\) to be injective: if \(E_K[g(X_k)] = 0\) for all \(k\), then
\(g = 0\) a.e.\ \citep{newey2003instrumental, blundell2007semi,
carrasco2007linear, newey2013nonparametric}. This condition is restrictive and
typically fails when \(Z\) is discrete and \(X\) is continuous, since $h^\star$ is only constrained by
finitely many moment conditions (Example~2.2 of
\citealp{babii2020completeness}).

Identification of the functional \(\psi(h^\star)\) does not require unique
identification of \(h^\star\). It suffices that the estimand be invariant over
the solution set, i.e.\ \(\psi(h^\diamond) = \psi(h^\star)\) for all
\(h^\diamond \in \mathcal{T}_K^{-1}(\mu_K)\)
\citep{severini2006some, babii2020completeness, bennett2023inference,
bennett2023source}. When \(\psi\) is continuous from \(L^2_K(X)\) to
\(\mathbb{R}\), the Riesz representation theorem implies
\[
\psi(h^\star)
  = \sum_{k=1}^K \frac{n_k}{N} E_K[\alpha_K(X_k) h^\star(X_k)]
  = \langle \alpha_K, h^\star \rangle_{L^2_K(X)},
\]
where \(\alpha_K \in L^2_K(X)\) is the Riesz representer of \(\psi\).
Hence, \(\psi(h^\star)\) is point-identified if and only if
\(\alpha_K \in \mathcal{H}_K\) \citep{severini2006some}, that is, if and only
if \(\psi(h) = 0\) for all \(h \in \mathcal{N}(\mathcal{T}_K)\).
This requirement is restrictive because \(\mathcal{H}_K\), the image of
\(L^2_K(Z)\) under \(\mathcal{T}_K^*\), is at most \(K\)-dimensional,
which sharply limits the allowable complexity of the Riesz representer
\(\alpha_K\). The following examples illustrate this limitation.

\setcounter{example}{0}
\begin{example}[continued]
For the counterfactual mean 
\(\psi(h^\star)=\int h^\star(x)\,p_X^\star(x)\,dx\), 
the Riesz representer is 
\(\alpha_K(x)=p_X^\star(x)/p_X^{(K)}(x)\), where 
\(p_X^{(K)}(x)=\sum_{k=1}^K (n_k/N)\,p^{(K)}(X_k=x)\).  
Thus point identification requires that, for some \(q^\star\in L^2_K(Z)\),
\[
p_X^\star(x)=\sum_{k=1}^K (n_k/N)\, q^\star(k)\, p^{(K)}(X_k=x),
\]
i.e., the counterfactual density \(p_X^\star\) lies in the linear span of the factual densities \(\{p^{(K)}(X_k):k\in[K]\}\).
\end{example}

\setcounter{example}{1}
\begin{example}[continued]
For confounding-robust surrogate construction, the previous example implies
that point identification of the counterfactual mean 
\(\psi(h^\star) = \int h^\star(x)\, p_X^\star(x)\,dx\)
requires that the surrogate distribution \(P_X^\star\) in the new experiment
lies in the linear span of the surrogate distributions from past experiments. 
\end{example}

\subsection{Connections to prior work}
\label{sec::relatedworkident}

When completeness holds, \(h^\star\) is uniquely identified,
\(\mathcal{T}_K\) has full finite rank, and \(\mathcal{T}_K^{-1}\) is
continuous (its domain being finite dimensional). In this case, any continuous
linear functional of \(h^\star\) can be written as one of \(\mu_K\)
\citep{newey2003instrumental, blundell2007semi, carrasco2007linear,
newey2013nonparametric}, and standard semiparametric efficiency methods apply
\citep{bickel1993efficient, laan2003unified, vanderLaanRose2011, DoubleML}. When completeness fails, the NPIV problem becomes ill-posed
\citep{newey2003instrumental, hall2005nonparametric, severini2006some,
blundell2007semi, chen2013optimal}. Point identification---corresponding to the
Riesz representer lying in the range of the adjoint operator---is typically insufficient
for root-\(N\) estimation of continuous linear functionals
\citep{severini2006some, carrasco2007linear, babii2017identification,
smucler2025asymptotic}. Stronger conditions, often termed \emph{strong
identification}, require the Riesz representer to satisfy a source-type
smoothness condition
\citep{carrasco2007linear, babii2020completeness, bennett2023inference,
bennett2023source}. Debiased machine learning estimators under such conditions
have been developed by \citet{bennett2023inference, bennett2023source}; see also
\citet{chen2011rate, chen2012estimation, darolles2011nonparametric,
horowitz2011applied} for related discussions. A distinctive feature of the
discrete NPIV setting considered here is that the range of the adjoint operator
is finite dimensional and therefore closed, so point identification and strong
identification of continuous linear functionals coincide.

\section{Approximating functional and bias expansions}

\label{sec::identification}

\subsection{Approximating functional and identification gap}

When point identification fails, inference on the target functional
\(\psi(h^\star)\) is ill posed. We therefore work with an always-identified
approximating functional \(\psi_K(h^\star)\), defined by evaluating \(\psi\) at
the minimum-norm solution of the moment restriction. This shifts attention from
the structural function to its minimum-norm representative and induces an
identification gap \(\psi_K(h^\star) - \psi(h^\star)\), which may vanish as the
number of instrument levels grows.

The approximating functional \(\psi_K(h^\star)\) is defined by the linear map
\(\psi_K : L^2_K(X) \to \mathbb{R}\),
\[
    \psi_K(h) := \psi(\Pi_K h),
    \qquad
    \Pi_K h := \arg\min_{f \in \mathcal{H}_K} \| h - f \|_{L^2_K(X)},
\]
where \(\Pi_K\) projects \(h\) onto \(\mathcal{H}_K\), the range of the adjoint
operator \(\mathcal{T}_K^*\). The Riesz representer of \(\psi_K\) is
\(\Pi_K \alpha_K\), which by construction lies in \(\mathcal{H}_K\). Hence,
\(\psi_K(h^\star)\) is point-identified through \(\psi(h_K)\), where
\(h_K := \Pi_K h^\star\) is the minimum-norm solution to
\[
    h_K := \arg\min_{h \in L^2_K(X)} \|h\|_{L^2_K(X)}
    \quad \text{subject to } \mathcal{T}_K h = \mu_K.
\]
The minimum-norm solution \(h_K\) represents the identifiable component of the
structural function \citep{severini2006some}. In dual terms, the approximating
functional \(\psi_K(h^\star) = \langle \Pi_K \alpha_K, h^\star\rangle\)
captures the identified component of \(\psi(h^\star)\), obtained by projecting
the Riesz representer \(\alpha_K\) onto \(\mathcal{H}_K\).

The minimum-norm solution is also the object typically estimated by standard
NPIV procedures---such as sieve, spectral, or Tikhonov regularization---when
point identification fails. In this sense, $\psi_K(h^\star)$ corresponds to the
identification-assumption-free target implicitly recovered by many NPIV
estimators \citep{carrasco2007linear, darolles2011nonparametric,
babii2020completeness}. Variants are possible by minimizing alternative norms,
including weighted \(L^2\) norms.

When \(\alpha_K \in \mathcal{H}_K\), the approximation \(\psi_K(h^\star)\)
coincides with the target \(\psi(h^\star)\). Otherwise, the next theorem shows
that \(\psi_K(h^\star)\) provides a second-order approximation to
\(\psi(h^\star)\). This result follows directly from the orthogonality of
projections, and closely related bounds appear---explicitly or implicitly---in
analyses of minimum-norm sieve and regularization estimators
\citep{shen1997methods, newey1997convergence, newey2003instrumental,
chen2012estimation, babii2020completeness}.

\begin{enumerate}[label=\bf{C\arabic*)}, ref=C\arabic*, series=cond]
\item \textit{Functional continuity:} \label{cond::continousfun}
For each \( K \in \mathbb{N} \), \( \psi : L^2_K(X) \to \mathbb{R} \) is bounded
and linear, satisfying \( |\psi(h)| \le C_K \|h\|_{L^2_K(X)} \) for all
\( h \in L^2_K(X) \) and some \( C_K \in (0, \infty) \).
\end{enumerate}

\begin{theorem}[Identification gap]
\label{theorem::approx}\label{theorem::asymident}
Under~\ref{cond::continousfun}, the identification gap satisfies
\[
\psi_K(h^\star) - \psi(h^\star)
  = \langle \alpha_K - \Pi_K \alpha_K,\,
           \Pi_K h^\star - h^\star \rangle_{L^2_K(X)}.
\]
Consequently, \( \psi_K(h^\star) \to \psi(h^\star) \) as \( K \to \infty \)
whenever either \( \|\Pi_K h^\star - h^\star\|_{L^2_K(X)} \to 0 \) or
\( \|\Pi_K \alpha_K - \alpha_K\|_{L^2_K(X)} \to 0 \).
\end{theorem}

In words, the identification gap vanishes whenever either the target functional
or the structural function becomes increasingly aligned with what the
instruments can identify as $K$ grows. By the Cauchy--Schwarz inequality, the gap satisfies the bound
\[
|\psi_K(h^\star) - \psi(h^\star)|
   \le \|\alpha_K - \Pi_K \alpha_K\|_{L^2_K(X)}\,
        \|h^\star - \Pi_K h^\star\|_{L^2_K(X)}.
\]
The term \(\|\alpha_K - \Pi_K \alpha_K\|_{L^2_K(X)}\) quantifies the extent to
which the identification condition \(\alpha_K \in \mathcal{H}_K\) is violated,
while \(\|h^\star - \Pi_K h^\star\|_{L^2_K(X)}\) measures how well the
minimum-norm solution approximates \(h^\star\). We refer to the convergence
\(\psi_K(h^\star) \to \psi(h^\star)\) as \emph{asymptotic identification}, in
the sense that the identification gap vanishes as \(K \to \infty\).

Motivated by these results, we aim to construct estimators that remain valid
whether or not point identification holds. In the next section, we develop an
efficient DML–JIVE estimator \( \widehat{\psi}_K^* \) of the approximating
functional \( \psi_K(h^\star) \). When point identification happens to hold---or
when the identification gap vanishes sufficiently quickly---\(\widehat{\psi}_K^*\)
gives inference for the target \( \psi(h^\star) \); otherwise, it provides valid
inference for the projection target \( \psi_K(h^\star) \). We view the
approximating functional as a working parameter that is well defined regardless
of identification, with any remaining bias governed by the gap characterized in
Theorem~\ref{theorem::approx}.

Before turning to estimation, we comment on when the identification gap may
vanish. According to Theorem~\ref{theorem::approx}, the gap disappears whenever
either \(\alpha_K\) or \(h^\star\) is well approximated by its projection onto
\(\mathcal{H}_K\) with vanishing error. This condition is \emph{doubly robust}:
asymptotic identification requires only one of the projection errors to vanish.
The first sufficient condition,
\(\|\alpha_K - \Pi_K \alpha_K\|_{L^2_K(X)} \to 0\), can admit a causal
interpretation in some settings and may be plausible in practice. The second, \(\|h^\star - \Pi_K h^\star\|_{L^2_K(X)} \to 0\), may hold when
\(\mathcal{H}_K\) becomes increasingly rich and \(h^\star\) is spectrally smooth
enough to be well approximated within these finite-dimensional spaces. Although the causal mechanism inducing such smoothness may be
unclear, Theorem~\ref{theorem::approx} shows that any amount of smoothness
shrinks the identification gap.



The following example illustrates how asymptotic identification may plausibly
arise in the problem of confounding-robust surrogate construction using multiple
historical randomized experiments as a categorical instrument \(Z\).

\setcounter{example}{1}
\begin{example}[continued]
\textit{Asymptotic identification with many experiments.}
A sufficient condition for asymptotic identification of the counterfactual mean
\(\psi(h^\star)=\int h^\star(x)\,p_X^\star(x)\,dx\)
is that the surrogate distribution \(P_X^\star\) of the new experiment can be
approximated, with vanishing error, by linear combinations of surrogate
distributions from past experiments as the number of experiments increases. We have
\[
\Pi_K \alpha_K(x)
  = \frac{p_X^{\diamond}(x)}{p_X^{(K)}(x)}, \qquad  
p_X^{\diamond}(x)
  := \sum_{k=1}^K \tfrac{n_k}{N} q^{\diamond}(k)\,p^{(K)}(X_k=x),
\]
where \(\alpha_K(x)=p_X^\star(x)/p_X^{(K)}(x)\) and \(q^{\diamond}\in L^2_K(Z)\)
is chosen so that \(\Pi_K\alpha_K\) is the \(L^2_K(X)\)-projection of
\(\alpha_K\) onto the linear span of the historical density ratios
\(\{p^{(K)}(X_k=x)/p_X^{(K)}(x):k\in[K]\}\). In digital experimentation, where the number and diversity of historical
experiments can grow rapidly \citep{duan2021online, athey2019surrogate}, this
span may become sufficiently rich to ensure asymptotic identification.
\end{example}

The next example illustrates how the working parameter  
\(\psi_K(h^\star)=E_K[\Pi_K \alpha_K(X)\,h^\star(X)]\) can be viewed as a
(potentially signed) weighted average of \(h^\star\).

\setcounter{example}{2}
\begin{example}[continued]
\textit{Interpretation of the working parameter.}
The parameter \(\psi_K(h^\star)\) corresponds to the counterfactual mean under a
synthetic experiment whose surrogate distribution \(p_X^{\diamond}\) is obtained
by projecting the density ratio \(\alpha_K = p_X^\star/p_X^{(K)}\) onto the
linear span of historical density ratios and rescaling by \(p_X^{(K)}\). 
Orthogonality of the projection implies
\(\sum_{k=1}^K \tfrac{n_k}{N} E_K[\Pi_K \alpha_K(X_k)]
  = \sum_{k=1}^K \tfrac{n_k}{N} q^{\diamond}(k) = 1\), so \(p_X^{\diamond}\)
integrates to one. When \(p_X^{\diamond}\) is nonnegative it defines a valid
density; otherwise, it should be interpreted as a (possibly signed) linear
combination of the historical surrogate distributions.
\end{example}

\subsection{Functional bias expansions}

\label{sec::vonmises}
The estimand \(\psi_K(h^\star)\) corresponds to the statistical parameter
\(\Psi^{(K)}(P^{(K)})\), where \(\Psi^{(K)}: P^{(K)} \mapsto
\psi(h_{P^{(K)}})\) is defined on the nonparametric model \(\mathcal{P}^{(K)}\)
and \(h_{P^{(K)}} := \Pi_{P^{(K)}} \mu_{P^{(K)}}\) denotes the minimum-norm
solution under \(P^{(K)}\). We begin by deriving a von Mises expansion for this
parameter, which forms the basis of our DML–JIVE estimators.

We begin with the bias expansion for the approximating functional $\psi_K$. Let \(q_K\) denote the minimum-norm solution to the dual inverse problem \(\mathcal{T}_K^* q_K = \Pi_K \alpha_K\). This dual function provides an equivalent representation of \(\psi_K(h^\star)\):
\begin{equation}
    \psi_K(h^\star) = \langle \mathcal{T}_K^* q_K, h^\star \rangle_{L^2_K(X)} = \langle q_K, \mu_K \rangle_{L^2_K(Z)}.
\end{equation}
Thus, \(\psi_K(h^\star)\) admits a dual representation as a weighted expectation of the outcome, given by \(\sum_{k=1}^K \frac{n_k}{N}\, q_K(k)\, E_K[Y_k]\). Next, let \(r_K \in \mathcal{R}(\mathcal{T}_K)\) denote the minimum-norm solution to the inverse problem \(\mathcal{T}_K^* r_K = h_K\), which necessarily exists since \(h_K \in \mathcal{H}_K = \mathcal{R}(\mathcal{T}_K^*)\). Finally, let \(\Pi_K^\perp = I_{L^2_K(X)} - \Pi_K\) denote the projection operator onto the orthogonal complement \(\mathcal{H}_K^\perp = \mathcal{N}(\mathcal{T}_K)\), where \(I_{L^2_K(X)}\) is the identity map. We define the score function as
\[
\varphi_K^*(o)
    := q_K(z)\bigl(y - h_K(x)\bigr)
       + \Pi_K^\perp \alpha_K(x)\,\bigl(r_K(z) - h_K(x)\bigr).
\]

\begin{theorem}[Functional bias expansion]
\label{theorem::vonmises}
Assume Condition~\ref{cond::continousfun} holds. Then, for any $\bar{P}^{(K)} \in \mathcal{P}^{(K)}$,
\begin{align*}
    \Psi^{(K)}(\bar{P}^{(K)}) - \Psi^{(K)}(P^{(K)}) 
    + \sum_{k=1}^K \frac{n_k}{N} E_{K}\!\left[\varphi_{\bar{P}^{(K)}}^*(O_{ki})\right]  
    ={}& - \left\langle \bar{q}_K - q_K,\, \mathcal{T}_K (\bar{h}_K - h_K) \right\rangle_{L_K^2(Z)} \\
    & + \left\langle \bar{\Pi}_K^\perp \bar{\alpha}_K - \Pi_K^\perp \alpha_K,\, \mathcal{T}_K^* (\bar{r}_K - r_K) - (\bar{h}_K - h_K) \right\rangle_{L_K^2(X)},
\end{align*}
where overbars indicate dependence on $\bar{P}^{(K)}$ (e.g., $\bar{h}_K = h_{\bar{P}^{(K)}}$).
\end{theorem}
Theorem~\ref{theorem::vonmises} shows that the score function~\(\varphi_K^*\)
characterizes the first-order behavior of \(\psi(\bar{h}_K) - \psi(h_K)\).
In this sense, \(\varphi_K^*\) serves as an influence function for the
parameter \(\Psi^{(K)}\), a central object in semiparametric theory and debiased
estimation~\citep{bickel1993efficient, vanderLaanRose2011, DoubleML}.
In Section~\ref{sec:efficiency}, we establish conditions for pathwise 
differentiability under which $\varphi_K^*$ is the efficient influence 
function. We therefore henceforth refer to $\varphi_K^*$ as the efficient 
influence function (EIF).

Theorem~\ref{theorem::vonmises} also implies a bias expansion for DML
estimators derived under strong identification, which is useful for analyzing
their behavior under asymptotic identification and under violations of
identification. We define the uncorrected score used in prior DML estimators as
\(\varphi_K : o \mapsto q_K(z)\{\, y - h_K(x) \,\}\), which coincides with
\(\varphi_K^*\) when \(\alpha_K \in \mathcal{H}_K\).

\begin{corollary}[Expansion under uncorrected score]
\label{cor::EIF}
Under Condition~\ref{cond::continousfun}, for any \(\bar{P}^{(K)} \in \mathcal{P}^{(K)}\), we have
\begin{align*}
    \psi(\bar{h}_K) - \psi_K(h^\star)
    + \sum_{k=1}^K \frac{n_k}{N} E_{K}[\varphi_{\bar{P}^{(K)}}(O_{ki})]
    ={}& - \left\langle \bar{q}_K - q_K,\, \mathcal{T}_K(\bar{h}_K - h_K) \right\rangle_{L_K^2(Z)} \\
    & + \left\langle \alpha_K - \Pi_K \alpha_K,\, h_K- \bar{h}_K \right\rangle_{L_K^2(X)}.
\end{align*}
\end{corollary}
Under the identification condition $\alpha_K \in \mathcal{H}_K$, 
Corollary~\ref{cor::EIF} recovers the expansions for $\psi(h^\star)$ 
established in previous work~\citep{bennett2023inference, bennett2023source}. 
When point identification fails, the expansion applies to the approximation 
$\psi_K(h^\star)$ and includes an additional bias term
$\langle \alpha_K - \Pi_K \alpha_K ,\, h_K - \bar{h}_K \rangle_{L^2_K(X)}$.  
If this term does not vanish sufficiently quickly, standard estimators based on 
$\varphi_K$ are generally biased for both $\psi_K(h^\star)$ and 
$\psi(h^\star)$, and corrections based on $\varphi_K^*$ become necessary.

\subsection{Connections to minimum-norm NPIV and regularization}

We relate our contribution to prior work on NPIV estimation in settings without
point identification, where attention likewise centers on an approximating
functional under fixed--instrument asymptotics.  

\citet{babii2020completeness} study a
class of spectrally regularized plug-in estimators of linear functionals under
fixed-instrument asymptotics, allowing completeness---and hence point
identification---to fail. They show that when completeness is violated, such
estimators remain consistent for the functional evaluated at the minimum-norm
solution but may converge to irregular, non-Gaussian limits. In contrast, we
develop DML estimators that incorporate one-step corrections based on the
efficient influence function of the approximating functional
\citep{bickel1993efficient}. These estimators retain asymptotic linearity and
normality even when completeness fails---or, more broadly, when the
identification gap does not vanish---and they attain semiparametric efficiency
under operator surjectivity in both $\sqrt{N}$ and slower-than-$\sqrt{N}$
regimes. A key technical ingredient is that the von~Mises expansion for the
approximating functional yields an influence function with an additional score
component that removes the first-order bias arising under failure of point
identification. This bias is not corrected by minimum-norm DML estimators
derived under point identification, such as those in
\citet{bennett2023inference, bennett2023source}. As shown in
Section~\ref{sec::vonmises}, under mild violations of strong identification this
bias decays in a controlled manner, whereas more substantial violations require
an additional bias-correction step to restore asymptotic normality. A detailed comparison between spectrally regularized plug-in estimators and the
DML-based approaches studied here lies beyond the scope of this paper and
represents an interesting direction for future work.

\section{Debiased Machine Learning JIVE Estimation}
\label{sec::estimatorsall}

\subsection{Estimator under asymptotic identification}
\label{sec::estimators}

We now turn to estimation. Our goal is to construct procedures that (i) remain
valid when point identification fails, (ii) accommodate many weak instruments,
and (iii) allow flexible, data-adaptive nuisance estimation. To this end, we
introduce an estimator of \(\psi_K(h^\star)\) and \(\psi(h^\star)\) under
asymptotic identification that combines debiased machine learning (DML)
\citep{vanderLaanRose2011, DoubleML} with a cross-fold splitting scheme playing
the same role as the jackknife correction in JIVE \citep{angrist1999jackknife}.
We refer to this estimator as \emph{DML--JIVE}. In the next subsection, we show
how DML--JIVE can be modified to provide valid inference for the projection
\(\psi_K(h^\star)\) when the identification gap does not vanish, and in
Section~\ref{sec::nuisance} we introduce npJIVE for bias-corrected nuisance
estimation with many instruments.

Given an estimator \( \widehat{h}_K \) of the minimum-norm solution \( h_K \), the plug-in estimator \( \psi(\widehat{h}_K) \) can be used to estimate \( \psi_K(h^\star) \). However, when \( \widehat{h}_K \) is obtained via flexible learning methods, the plug-in estimator \( \psi(\widehat{h}_K) \) may be biased and fail to achieve \( \sqrt{N} \)-convergence or asymptotic normality due to its first-order sensitivity to the estimation error \( \widehat{h}_K - h_K \) \citep{bickel1993efficient, vanderLaanRose2011, DoubleML}. To address this, debiasing methods are often required to eliminate the leading bias term and enable valid inference.

To improve upon the plug-in estimator, we propose a doubly robust, cross-fold
DML--JIVE estimator. Recall that $q_K$ denotes the minimum-norm solution to the
dual inverse problem $\mathcal{T}_K^{*} q_K = \Pi_K \alpha_K$. Because
$q_K \in \mathcal{R}(\mathcal{T}_K)$, there exists $\beta_K \in \mathcal{H}_K$
such that $q_K = \mathcal{T}_K \beta_K$, where $\beta_K$ is the minimum-norm
solution to $\mathcal{T}_K^{*}\mathcal{T}_K \beta_K = \Pi_K \alpha_K$.
Given estimators $\widehat{h}_K$ and $\widehat{\beta}_K$ of $h_K$ and $\beta_K$,
respectively, our estimator takes the form of a one-step debiased estimator
combined with a cross-fold jackknife correction:
\begin{align*}
  \widehat{\psi}_K 
  = \psi\!\bigl(\widehat{h}_K\bigr) 
    + \frac{2}{N} \sum_{v=0,1} \sum_{k=1}^K \sum_{i=1}^{n_k} 
      \mathbbm{1}\{V_{ki}=v\}
      \left[
        \frac{2}{n_k} \sum_{j=1}^{n_k} 
        \mathbbm{1}\{V_{kj} \neq v\}\,\widehat{\beta}_K(X_{kj})
      \right]
      \bigl(Y_{ki} - \widehat{h}_K(X_{ki})\bigr),
\end{align*}
where $V_{ki} \in \{0,1\}$ indicates fold membership in a two-fold split of
approximately equal size.  
This estimator resembles the fixed-\(K\) adversarial DML estimator of
\citet{bennett2023inference}, which is derived and analyzed only under strong
identification, adapted here to the discrete-instrument setting with the
adversarial class taken to be the full space \(L^2_K(Z)\). Departing from
standard DML estimators, our cross-fold construction ensures that the estimate
\(\frac{2}{n_k}\sum_{j=1}^{n_k}\mathbbm{1}\{V_{kj}\neq v\}\,\widehat{\beta}_K(X_{kj})\)
of \(q_K(Z_{ki})\) is computed using data independent of \((Y_{ki}, X_{ki})\),
thereby mitigating the many–weak–instrument bias that arises when some \(n_k\)
are small. Although our theoretical analysis uses two-fold sample splitting,
leave-one-out (jackknife) splitting---akin to JIVE \citep{angrist1999jackknife}---can
also be used and is more stable in our experiments.

The cross-fold JIVE construction substantially relaxes the required growth rate
of the per-cell sample sizes \(n_k\) as \(K \to \infty\), requiring only that
\(\min_{k \in [K]} n_k / \log K \to \infty\); when \(n_k \asymp N/K\), this
reduces to the mild condition \(N/(K \log K) \to \infty\). In contrast, without
this jackknife-style adjustment, much stronger growth conditions are needed.
For instance, ensuring that
\(\widehat{q}_{K,v}(k) := (2/n_k)\sum_{i:V_{ki}=v}\widehat{\beta}_K(X_{ki})\)
is \(N^{-1/4}\)-consistent for \(q_K(k) = \mathcal{T}_K(\beta_K)(k)\) would
require \(n_k \asymp \sqrt{N}\) (equivalently, \(N/K^{2} \to \infty\) when
\(n_k \asymp N/K\)). These stringent requirements arise solely from the sampling
variability of the empirical averages within each instrument cell. Fixing \(\widehat{\beta}_K\), \(\widehat{h}_K\), and the observation
\((X_{ki},Y_{ki},V_{ki}=v)\), cross-fold splitting guarantees $
\mathbb{E}\!\left[\widehat{q}_{K,v}(k)\{Y_{ki}-\widehat{h}_K(X_{ki})\}
\mid X_{ki},Y_{ki},V_{ki}=v\right]
= \mathcal{T}_K(\widehat{\beta}_K)(X_{ki})\{Y_{ki}-\widehat{h}_K(X_{ki})\}.$
This conditional unbiasedness removes the leading bias of
\(\psi(\widehat{h}_K)\), even if \(\widehat{q}_{K,v}\) converges slowly.
Informally, this independence ensures that the bias within each level \(k\)
averages out across \(k \in [K]\) in the definition of \(\widehat{\psi}_K\) as
\(K \to \infty\).

In Section~\ref{sec::theory1}, we first show that the estimator
$\widehat{\psi}_K$ is $\sqrt{N}$–consistent, doubly robust, and asymptotically
normal for the approximating functional $\psi_K(h^\star)$ under
many--weak--instrument asymptotics. It admits the asymptotically linear
expansion
\[
\widehat{\psi}_K - \psi_K(h^\star)
= \frac{1}{N}\sum_{k=1}^K \sum_{i=1}^{n_k} \varphi_K(O_{ki})
+ o_p\!\left((N/\sigma_K^2)^{-1/2}\right),
\]
where $\varphi_K(o) := q_K(z)\{y - h_K(x)\}$ and
$\sigma_K^2 := \sum_{k=1}^K \tfrac{n_k}{N} q_K^2(k)\,
\mathrm{Var}_K[Y_k - h_K(X_k)]$.  
Hence, by Lindeberg’s central limit theorem,
$\sqrt{N/\sigma_K^2}\big(\widehat{\psi}_K - \psi_K(h^\star)\big)
\overset{d}\to \mathcal{N}(0,1)$ as $K\to\infty$, and inference may be carried
out using Wald-type confidence intervals with influence-function variance
estimates. When $\sigma_K^2$ diverges with $N$, the estimator converges at a
slower rate (Section~\ref{sec::weakident}). Section~\ref{sec:efficiency} further shows that
$\widehat{\psi}_K$ is nonparametrically regular and efficient under
many--weak--instrument asymptotics, even when $\sigma_K^{2}\to\infty$, provided
that $\mathcal{T}_K$ is surjective.

To obtain asymptotic normality for the true target $\psi(h^\star)$, the
identification gap must vanish sufficiently quickly:
\begin{equation}
\big\langle \alpha_K - \Pi_K \alpha_K,\,
      \Pi_K h^\star - h^\star \big\rangle_{L^2_K(X)}
= o_p\!\left((N/\sigma_K^2)^{-1/2}\right).
\label{eqn::firstDR}
\end{equation}
If this condition fails, in view of Corollary \ref{cor::EIF}, the estimator remains asymptotically normal for the
projected target $\psi_K(h^\star)$ provided that
\begin{equation}
    \big\langle \alpha_K - \Pi_K \alpha_K,\,
      \widehat{h}_K - h_K \big\rangle_{L^2_K(X)}
= o_p\!\left((N/\sigma_K^2)^{-1/2}\right),
\label{eqn::secondDR}
\end{equation}
together with standard rate conditions for the nuisance estimators. It is
consistent for the true target $\psi(h^\star)$ whenever
$\|\alpha_K - \Pi_K \alpha_K\|_{L^2_K(X)} = o(1)$. As discussed in
Section~\ref{sec::identification}, decay of
$\|\alpha_K - \Pi_K \alpha_K\|_{L^2_K(X)}$ is plausible and may correspond to a
substantive causal assumption (e.g., in Example~2, the sequence of historical
experiments becomes sufficiently rich and diverse). Consequently,
\eqref{eqn::secondDR} is, in our view, more attainable than
\eqref{eqn::firstDR}: the term $\|h^\star - \Pi_K h^\star\|_{L^2_K(X)}$ need not
vanish without strong spectral smoothness, whereas $\widehat{h}_K$ is consistent
for $h_K$ under comparatively mild conditions.

\medskip

\noindent \textbf{Single-split DML-JIVE estimator for bounded cell size.} 
The estimator $\widehat{\psi}_{K}$ need not be asymptotically normal when some $n_k$ remain bounded (for example, when $N/K$ remains bounded while $K$ grows), although it remains $\sqrt{N}$-consistent. Intuitively, this loss of normality arises from correlations across folds in the jackknife procedure that do not vanish asymptotically, causing $\widehat{\psi}_{K}$ to converge to a mixture distribution. To enable valid inference in such cases, we propose a single-split variant of the one-step estimator that uses each data split only once:
\begin{align*}
  \widehat{\psi}_{K}^{\diamond}
  = \psi\bigl(\widehat{h}_K\bigr) 
    + \frac{2}{N} \sum_{k=1}^K \sum_{i=1}^{n_k} 
      \mathbbm{1}\{V_{ki} = 0\} 
      \left\{
        \frac{2}{n_k} \sum_{j=1}^{n_k} 
        \mathbbm{1}\{V_{kj} = 1\} \widehat{\beta}_K(X_{kj})
      \right\}
      \bigl(Y_{ki} - \widehat{h}_K(X_{ki})\bigr).
\end{align*}
The variance of \( \widehat{\psi}_{K}^{\diamond} \) is generally at least twice that of \( \widehat{\psi}_K \), since only half of the data are used to compute the empirical mean. Note \( \widehat{\psi}_K \) is the average of two versions of \( \widehat{\psi}_{K}^{\diamond} \), each computed on a different data split.

\subsection{Modified estimator for settings without asymptotic identification}
\label{sec::estimators2}

Asymptotic normality of the DML--JIVE estimator $\widehat{\psi}_K$ for $\psi(h^\star)$ generally requires the identification gap to vanish. When this condition fails, $\widehat{\psi}_K$ need not be asymptotically normal for either the projected estimand $\psi_K(h^\star)$ or the true target $\psi(h^\star)$. We therefore propose a modified estimator that delivers valid inference for $\psi_K(h^\star)$ without requiring identification, and that, under asymptotic identification, also provides valid inference for $\psi(h^\star)$.

Our modified estimators require the estimation of additional nuisance functions. Specifically, we need estimates of \( \alpha_K \), \( \Pi_K \alpha_K \), and \( r_K \) in Theorem~\ref{theorem::vonmises}. Because \( r_K \in \mathcal{R}(\mathcal{T}_K) \), there exists \( \rho_K \in L^2_K(X) \) such that \( r_K = \mathcal{T}_K \rho_K \). The function \( \rho_K \) can be estimated using techniques analogous to those employed for estimating \( h_K \) and \( \beta_K \) (see Appendix~\ref{appendix::nuisanceestimation3}).

Given estimators \( \widehat{h}_K \), \( \widehat{\alpha}_K - \widehat{\Pi}_K \widehat{\alpha}_K \), and \( \widehat{\rho}_K \) of \( h_K \), \( \alpha_K - \Pi_K \alpha_K \), and \( \rho_K \), respectively, our modified DML-JIVE estimator equals the DML-JIVE estimator \( \widehat{\psi}_K \), augmented with an additional correction term:
\begin{align*}
  \widehat{\psi}_K^*
  = \widehat{\psi}_K
  + \frac{1}{N} \sum_{v=0,1} \sum_{k=1}^K \sum_{i=1}^{n_k}
      \mathbbm{1}\{V_{ki} = v\}
      \bigl(\widehat{\alpha}_K(X_{ki}) - \widehat{\Pi}_K \widehat{\alpha}_K(X_{ki})\bigr)
      \bigl(\widehat{r}_{K,-v}(k) - \widehat{h}_K(X_{ki})\bigr),
\end{align*}
where \( \widehat{r}_{K,-v}(k) := \frac{2}{n_k} \sum_{j=1}^{n_k} \mathbbm{1}\{V_{kj} \neq v\} \widehat{\rho}_K(X_{kj}) \). The large-sample properties of this estimator will follow from Theorem \ref{theorem::vonmises}. The DML-JIVE estimator \( \widehat{\psi}_K \) from the previous section is a special case of \( \widehat{\psi}_K^* \) obtained by setting \( \widehat{\alpha}_K - \widehat{\Pi}_K \widehat{\alpha}_K := 0 \), in which case the correction term vanishes. One advantage of \( \widehat{\psi}_K \) is that it avoids the estimation of \( \alpha_K - \Pi_K \alpha_K \) and \( r_K \) by relying on asymptotic identification.

\section{Nonparametric jackknife instrumental variable estimation}

\label{sec::nuisance}

\subsection{Estimation of the minimum-norm solution}
\label{sec::nuisance1}

The performance of the proposed DML--JIVE estimators depends critically on the
quality of nuisance estimation under many weak instruments. In particular, our
estimators rely on two nuisance functions: the minimum-norm solution $h_K$ and
the debiasing nuisance $\beta_K$. Standard approaches, such as 2SLS, for
estimating these functions can be biased or inconsistent when the number of
instruments is comparable to the total sample size
\citep{angrist1999jackknife, chao2012asymptotic, peysakhovich2018learning}. In
linear IV models, the jackknife instrumental variables estimator (JIVE)
addresses this issue by regressing $Y$ on leave-one-out predictions of $X$ given
$Z$, obtained from OLS fits that exclude the observation being predicted
\citep{angrist1999jackknife}.

We propose the \emph{nonparametric Jackknife Instrumental Variables Estimator}
(npJIVE) for estimating minimum-norm solutions of conditional moment
restrictions via Tikhonov-regularized empirical risk minimization. When
$n_k \asymp N/K$, a key finding is that, under many weak instruments, the
effective sample size for learning a nonparametric function is $N/\sqrt{K}$.
Under a $\nu$-source condition on $h_K$, this effective sample size improves to
$\min\!\{ N,\; (N/\sqrt{K})^{(2 + 4\nu)/(2 + 2\nu)} \}$ for $\nu \in [0,1]$, and
full-sample efficiency is recovered when
$K = o\!\left(N^{2\nu/(1 + 2\nu)}\right)$. Stronger source conditions therefore
permit more rapid growth in the number of instruments; for example, when
$\nu = 1$, full efficiency holds for $K = o(N^{2/3})$. Such source conditions
are standard in analyses of ill-posed inverse problems
\citep{carrasco2007linear}.

To develop a nonparametric extension of JIVE, we adopt the perspective of empirical risk minimization.  The minimum-norm solution \(h_K\) minimizes the risk  $R_K(h) := \sum_{k=1}^K \frac{n_k}{N} \left\{[\mathcal{T}_K(\mathrm{id}_Y - h)](k)\right\}^2,$
where \(\mathrm{id}_Y:(x,y)\mapsto y\) projects onto the \(y\)-coordinate. 
Accordingly, \(h_K\) can be estimated by minimizing an empirical estimate of \(R_K(h)\) over \(h\) in a flexible function class. Under fixed-\(K\) asymptotics, a natural choice is the plug-in estimator
\[
\widehat{R}_K^{\text{plug-in}}(h)
  = \sum_{k=1}^{K} \frac{n_k}{N} 
    \left([\widehat{\mathcal{T}}_K(\mathrm{id}_Y - h)](k)\right)^2,
  \quad \text{where} \quad 
  [\widehat{\mathcal{T}}_K f](k)
  = \frac{1}{n_k} \sum_{i=1}^{n_k} f(k, X_{ki}, Y_{ki}).
\]
Minimizing this risk over a linear function class recovers the classical 2SLS estimator, while minimization over a nonparametric (possibly regularized) class yields the adversarial NPIV estimators of \citet{dikkala2020minimax, bennett2019deep, bennett2023inference, bennett2023variational}, where the adversarial class is \(L^2_K(Z)\). With many instruments, these estimators may be biased or inefficient because the empirical risk \(\widehat{R}_K^{\text{plug-in}}(h)\) can be inconsistent for \(R_K(h)\) when \(\min_{k \in [K]} n_k \not\to \infty\), and converge slowly when \(\min_{k \in [K]} n_k \to \infty\). 
This bias arises because \([\widehat{\mathcal{T}}_K(\mathrm{id}_Y - h)](k)^2\) is biased for \([\mathcal{T}_K(\mathrm{id}_Y - h)](k)^2\) in expectation for each \(k \in [K]\).

The key idea of npJIVE is to estimate the population risk using cross-fold splitting. 
Specifically, for each observation \((k, i)\), we define a binary indicator \(V_{ki} \in \{0,1\}\) that assigns data to one of two folds. 
npJIVE estimates \(\mathcal{T}_K(f)\mathcal{T}_K(g)\) via the cross-fold product \(\widehat{\mathcal{T}}_{K,0}(f)\widehat{\mathcal{T}}_{K,1}(g)\), where for each fold \(v \in \{0,1\}\), the empirical operator \(\widehat{\mathcal{T}}_{K,v} : L^2_K(X, Y) \to L^2_K(Z)\) is defined by
\[
[\widehat{\mathcal{T}}_{K,v}(f)](k)
  = \frac{1}{\sum_{i=1}^{n_k} \mathbf{1}\{V_{ki} = v\}}
    \sum_{i=1}^{n_k} \mathbf{1}\{V_{ki} = v\} f(X_{ki}, Y_{ki}).
\]
Although \( \widehat{\mathcal{T}}_{K,0}(f)(k)\widehat{\mathcal{T}}_{K,1}(g)(k) \) is generally inconsistent for \( \mathcal{T}_K(f)(k)\mathcal{T}_K(g)(k) \) when \( n_k \not\to \infty \), it is unbiased in expectation because \(\widehat{\mathcal{T}}_{K,0}\) and \(\widehat{\mathcal{T}}_{K,1}\) are independent. 
The npJIVE estimator of the risk \(R_K(h)\) is then defined as
\[
\widehat{R}_{K}(h)
  = \sum_{k=1}^K \frac{n_k}{N}
    [\widehat{\mathcal{T}}_{K,0}(\mathrm{id}_Y - h)](k)
    [\widehat{\mathcal{T}}_{K,1}(\mathrm{id}_Y - h)](k).
\]
By standard concentration results for means of independent random variables, as \(K \to \infty\),
\begin{align*}
\widehat{R}_{K}(h)
  &= \sum_{k=1}^K \frac{n_k}{N}
     E_K\!\left\{[\widehat{\mathcal{T}}_{K,0}(\mathrm{id}_Y - h)](k)\right\}
     E_K\!\left\{[\widehat{\mathcal{T}}_{K,1}(\mathrm{id}_Y - h)](k)\right\}
     + O_p(N^{-1/2}) \\
  &= \sum_{k=1}^K \frac{n_k}{N}
     \left\{[\mathcal{T}_K(\mathrm{id}_Y - h)](k)\right\}^2
     + O_p(N^{-1/2}).
\end{align*}
Hence, the risk estimator is consistent, with \(\widehat{R}_{K}(h) = R_K(h) + O_p(N^{-1/2})\). Although we use two-fold sample splitting for our theoretical analysis, we recommend using leave-one-out or jackknife splitting, as in the original JIVE procedure \citep{angrist1999jackknife}, which is often more stable in practice.

Our proposed npJIVE estimator of \(h_K\) minimizes the estimated risk \(\widehat{R}_{K}\) with a Tikhonov regularization penalty:
\[
\widehat{h}_K := \argmin_{h \in \mathcal{F}_{\mathrm{primal}}} 
  \widehat{R}_{K}(h) + \lambda_K \|h\|^2_{2,K},
\]
where \(\mathcal{F}_{\mathrm{primal}} \subset L^2_K(X)\) models \(h_K\), and 
\(\|h\|^2_{2,K} := \frac{1}{N} \sum_{k=1}^K \sum_{i=1}^{n_k} \{h(X_{ki})\}^2\). 
The regularization parameter \(\lambda_K > 0\) is chosen to vanish as \(K \to \infty\), ensuring that \(\widehat{h}_K\) converges to the minimum-norm solution \citep{engl2015regularization}. 
We next establish estimation rates for \(\widehat{h}_K\) under the many-instrument asymptotic regime. 
For clarity, we first present rates assuming sup-norm covering numbers are available for \(\mathcal{F}_{\mathrm{primal}}\); 
more general results in terms of critical radii and Rademacher complexities are provided in the appendix.

For a function class \(\mathcal{F}\), let \(N_{\infty}(\varepsilon, \mathcal{F})\) denote the \(\varepsilon\)-covering number of \(\mathcal{F}\) under the \(P\)-essential supremum metric \(d_{\infty}(\theta_1,\theta_2) := \|\theta_1 - \theta_2\|_{\infty}\) \citep[Chapter~2]{vanderVaartWellner}. 
Define the corresponding metric entropy integral as $\mathcal{J}_{\infty}(\delta, \mathcal{F}) := \int_0^{\delta} \{\log N_{\infty}(\varepsilon, \mathcal{F})\}^{1/2} d\varepsilon.$
To establish our results, we impose the following regularity conditions. 
Throughout, let \(h_{K}(\lambda_K)\) denote the minimum-norm solution to the Tikhonov-regularized inverse problem $(\mathcal{T}_{K} \mathcal{T}_{K}^* + \lambda_K I) h_{K}(\lambda_K) = \mu_K.$

\begin{enumerate}[label=\bf{A\arabic*)}, ref={A\arabic*}, series=condA]
\item \textit{(nonparametric, convex function class that is not too large):} \label{cond::regularityOnActionSpace}
   $\mathcal{F}$ is uniformly bounded, convex, and $\mathcal{J}_{\infty}(\delta, \mathcal{F}) \leq C \delta^{1-1/(2\gamma)}$ for some $\gamma > 1/2$ and $C>0$, and for every $\delta > 0$. 
   \item (\textit{Realizability}) \label{cond::funclasscontainsprimary} $h_K(\lambda)$ and $h_K$ lie in $\mathcal{F}_{\mathrm{primal}}$ for each $K < \infty$ large enough and $\lambda \geq 0$ small enough.  
   \item \textit{(Source condition):}  \(h_K = (\mathcal{T}_{K}^*\mathcal{T}_{K} )^{ \nu} u_K \) for some \(\nu \in [0,1]\), \(u_K \in L^2_K(X) \) with $\sup_K \|u_K\|_{L^2_K(X)} < \infty$. \label{cond::sourceprimary} 
\end{enumerate}

In the statement below, let \(n_{\min} := \min_{k \in [K]} n_k\) and define
\[
\Delta_N 
:= (n_{\min}N)^{-\gamma/(2\gamma+1)} 
   + \sqrt{\tfrac{\log K}{n_{\min}}}\,N^{-\gamma/(2\gamma+1)}.
\]

\begin{theorem}
\label{theorem::weakstrongratessupprimary}
Assume \ref{cond::regularityOnActionSpace} with 
\(\mathcal{F} := \mathcal{F}_{\mathrm{primal}}\), 
\ref{cond::funclasscontainsprimary}, 
and \ref{cond::sourceprimary}.  
Choose \(\lambda_K\) so that 
\(\lambda_K^{\nu+1} \asymp \Delta_N\).  
Then
\begin{align*}
\|\mathcal{T}_K(\widehat{h}_K(\lambda_K) - h_K)\|_{L^2_K(Z)}
&= O_p\!\bigl( N^{-\gamma/(2\gamma+1)} + \Delta_N^{(1+2\nu)/(2+2\nu)} \bigr),\\[0.4em]
\|\widehat{h}_K(\lambda_K) - h_K\|_{L^2_K(X)}
&= O_p\!\bigl( \lambda_K^{-1/2}\,[\, N^{-\gamma/(2\gamma+1)} + \Delta_N^{(1+2\nu)/(2+2\nu)} ] \bigr) = O_p\!\bigl(\Delta_N^{\nu/(\nu+1)}\bigr).
\end{align*}
\end{theorem}

Condition~\ref{cond::regularityOnActionSpace} holds with exponent
\(\gamma := s/d\) when \(\mathcal{F}\) is a \(d\)-variate Hölder or Sobolev class
with smoothness \(s > d/2\) (Theorem 2.7.2 of \citealp{vanderVaartWellner};
Corollary 4 of \citealp{nickl2007bracketing}).  
Condition~\ref{cond::sourceprimary} is a source condition on the minimum-norm
solution \(h_K\) and quantifies the degree to which \(h_K\) is identified. While
source conditions are standard in the inverse problems literature
\citep{carrasco2007linear, bennett2023source}, they are typically imposed on the
target function \(h^\star\). Here, we require a source condition only for
\(h_K\), not for \(h^\star\), which makes the assumption milder---particularly
when \(h_K\) does not converge to \(h^\star\) as \(K \to \infty\)---though it
remains nontrivial. The condition holds automatically when \(\nu = 0\) and yields
faster rates for \(\nu > 0\), in which case \(h_K\) is smooth relative to the
eigenfunctions of \(\mathcal{T}_K^* \mathcal{T}_K\). Because \(\mathcal{T}_K^*\)
has finite rank and \(\mathcal{R}(\mathcal{T}_K^*) = \mathcal{R}(\mathcal{T}_K^*
\mathcal{T}_K)\), the representation \(h_K = (\mathcal{T}_K^* \mathcal{T}_K) u_K\)
always exists when \(\nu = 1\). The substantive content of the source condition
lies in the uniform boundedness of \(u_K\), which need not hold and determines
the attainable rates. The assumption \(\sup_K \|u_K\|_{L^2_K(X)} < \infty\) can
be relaxed by tracking the divergence of \(\|u_K\|_{L^2_K(X)}\) explicitly,
which may grow on the order of the inverse singular values of \(\mathcal{T}_K\).

Theorem~\ref{theorem::weakstrongratessupprimary} shows that the estimation error
depends on the oracle rate \(N^{-\gamma/(2\gamma+1)}\) and on the term
\(\Delta_N\), which captures the effect of the growing number of instrument
levels. When \(\nu = 0\) and 
\(\log K = o\!\bigl(n_{\min}^{1/(2\gamma+1)}\bigr)\), the leading component of
\(\|\mathcal{T}_K\{\widehat{h}_K(\lambda_K) - h_K\}\|_{L^2_K(Z)}\) is
\((\sqrt{n_{\min} N})^{-\gamma/(2\gamma+1)}\), which can be written as
\(m_N^{-\gamma/(2\gamma+1)}\) for \(m_N := \sqrt{n_{\min}N}\). Thus \(m_N\) acts as an effective sample size for estimating a
\(\gamma\)-smooth function, as \(N^{-\gamma/(2\gamma+1)}\) is the rate that
would be obtained in fixed-\(K\) settings.
In a balanced design with \(n_{\min} \asymp N/K\), we obtain
\(m_N \asymp N/\sqrt{K}\), indicating that the effective sample size is reduced
by a factor of \(\sqrt{K}\). Under a \(\nu\)-source condition on \(h_K\), the
effective sample size improves to 
\(\min\{N,\; (N/\sqrt{K})^{(2 + 4\nu)/(2 + 2\nu)}\}\) for \(\nu \in [0,1]\).
The full effective sample size \(N\) is recovered when 
\(K = o\!\left(N^{2\nu/(1 + 2\nu)}\right)\); when \(\nu = 1\), this requires
\(K = o(N^{2/3})\).

Theorem~\ref{theorem::weakstrongratessupprimary} shows that the strong-norm
convergence rate $\|\widehat{h}_K(\lambda_K) - h_K\|_{L^2_K(X)}$ is
potentially much slower than the weak-norm rate
$\|\mathcal{T}_K(\widehat{h}_K(\lambda_K) - h_K)\|_{L^2_K(Z)}$. In particular,
the strong-norm rate is $\Delta_N^{\nu/(\nu+1)}$, which vanishes only when the
source condition holds with some $\nu>0$. In the most favorable case,
$\nu=1$, the strong-norm rate simplifies to $\Delta_N^{1/2} = m_N^{-\gamma/(2\gamma+1)}$. We note that $\lambda_K$ was chosen to optimize the weak-norm rate; faster
strong-norm rates are attainable by selecting a larger $\lambda_K$, at the cost
of slower weak-norm convergence \citep{bennett2023source}. However, in practice,
the weak-norm is typically the primary criterion used for tuning.


\subsection{Estimation of dual solution under asymptotic identification}

\label{sec::nuisance2}

In this section, we use npJIVE to construct an estimator \(\widehat{\beta}_K\) of 
the solution \(\beta_K\) to the inverse problem 
\(\mathcal{T}_K^* \mathcal{T}_K \beta_K = \Pi_K \alpha_K\).  
A natural approach is to estimate \(\Pi_K \alpha_K\) and then solve a plug-in 
approximation of this equation. We outline general procedures for estimating 
\(\Pi_K \alpha_K\) and \(\beta_K\) in this manner in 
Appendix~\ref{appendix::nuisanceestimation3}. However, this two-step approach 
introduces additional bias from the first-stage estimation and requires 
estimating further nuisance functions.

Our main contribution in this section is to show that the inverse problem
\(\mathcal{T}_K^* \mathcal{T}_K \beta_K = \Pi_K \alpha_K\) can be replaced by its
Tikhonov-regularized counterpart
\((\mathcal{T}_K^* \mathcal{T}_K + \lambda I)\beta_K = \alpha_K\), in which the
projection \(\Pi_K \alpha_K\) is replaced by the unprojected representer
\(\alpha_K\), while still preserving consistency under asymptotic
identification. This modification is nontrivial: when
\(\alpha_K \notin \mathcal{H}_K\), the representer \(\alpha_K\) lies outside the
range of \(\mathcal{T}_K^* \mathcal{T}_K\), and the unregularized equation admits
no solution. Replacing the inverse problem with its regularized form yields a
substantially simpler estimator that avoids both first-stage projection and
additional nuisance estimation. The validity of this estimator under strong
identification and fixed-instrument asymptotics is established in
\citet{bennett2023inference, bennett2023source, bennett2023minimax}. In contrast,
our results bound the bias when strong identification fails and establish
convergence rates in terms of
\(\|\Pi_K \alpha_K - \alpha_K\|_{L^2_K(X)}\), which quantifies the degree of
identification failure.

Our estimation strategy relies on the fact that $\beta_K$ is the minimizer of the population risk:
\[
R_K^*(\beta) := \|\mathcal{T}_K(\beta)\|^2_{L^2_K(Z)} - 2 \langle \beta, \Pi_K \alpha_K \rangle_{L^2_K(X)}
= \|\mathcal{T}_K(\beta)\|^2_{L^2_K(Z)} - 2 \psi(\beta) - 2 \langle \beta, \Pi_K \alpha_K - \alpha_K \rangle_{L^2_K(X)},
\]
where the final equality follows from the Riesz representation \(\psi(\beta) = \langle \alpha_K, \beta \rangle_{L^2_K(X)}\). The final term is negligible as \(K \to \infty\) under the asymptotic identification condition \(\|\Pi_K \alpha_K - \alpha_K\|_{L^2_K(X)} \to 0\). Thus, to avoid estimating \(\Pi_K \alpha_K\), we drop this term and work with the simplified risk \(\|\mathcal{T}_K(\beta)\|^2_{L^2_K(Z)} - 2\psi(\beta)\). This risk was used under the strong identification assumption in \cite{bennett2023inference} and \cite{bennett2023source}. We now show that estimation based on this risk remains valid under the weaker assumption of asymptotic identification, provided cross-fold splitting and regularization is employed.

We propose to estimate \( \beta_K \) by minimizing the Tikhonov-regularized npJIVE risk:
\[
\widehat{\beta}_K(\lambda_K) := \argmin_{\beta \in \mathcal{F}_{\mathrm{dual}}} \left[ \sum_{k=1}^K \frac{n_k}{N} \widehat{\mathcal{T}}_{K,0}(\beta) \widehat{\mathcal{T}}_{K,1}(\beta) - 2 \psi(\beta) + \lambda_K \| \beta \|_{2,K}^2 \right],
\]
for a function class \( \mathcal{F}_{\mathrm{dual}} \) and regularization parameter $\lambda_K \geq 0$. When strong identification fails, the unregularized population risk \( \|\mathcal{T}_K(\beta)\|^2_{L^2_K(Z)} - 2 \psi(\beta) \) may not admit a minimizer for finite \( K \). This is not problematic, as Tikhonov regularization yields an approximate minimizer.
 
We now establish convergence rates for \(\widehat{\beta}_K(\lambda_K)\). The proof proceeds by first bounding the estimation error \(\widehat{\beta}_K(\lambda_K) - \beta_K(\lambda_K)\) for the regularized solution \(\beta_K(\lambda_K)\) to the inverse problem \((\mathcal{T}_K^* \mathcal{T}_K + \lambda_K I)\beta_K(\lambda_K) = \alpha_K\). We then bound the oracle bias \(\beta_K(\lambda_K) - \beta_K\), which reflects both regularization and approximation errors due to replacing \(\Pi_K \alpha_K\) with \(\alpha_K\). The following theorem controls this bias, with tighter bounds obtainable under a source condition.

\begin{enumerate}[label=\bf{A\arabic*)}, ref={A\arabic*}, resume=condA]
\item \textit{(Source condition):}  \(\beta_K = (\mathcal{T}_{K}^* \mathcal{T}_{K})^{ \nu} w_{K} \) for \(\nu \in [0,1]\), \(w_{K} \in L^2_K(X) \) with $\sup_K \|w_K\|_{L^2_K(X)} < \infty$. \label{cond::source}  
\end{enumerate}

 \begin{theorem}[Oracle bias decomposition]
     \label{lemma::approxstrongident} Under \ref{cond::continousfun} and \ref{cond::source}, we have that
    \[
    \|\mathcal{T}_{K}(\beta_{K}(\lambda_K) - \beta_{K}) \|_{L^2_K(Z)} = O(\lambda_K^{\nu + \frac{1}{2}}) \|w_{K}\|_{L^2_K(X)} + O(\lambda_K^{-1/2}) \|\alpha_K - \Pi_{K} \alpha_K\|_{L^2_K(X)}.
    \] 
\end{theorem}

The first term in the error bound of Lemma~\ref{lemma::approxstrongident} corresponds to the regularization bias and is standard in the analysis of Tikhonov-regularized inverse problems \citep{bennett2023source}. The second term is novel, arising from the approximation of \(\Pi_K \alpha_K\) by \(\alpha_K\). The oracle bias is minimized when \(\lambda_K \asymp \|\alpha_K - \Pi_K \alpha_K\|^{1/(1+\nu)}\), yielding a best-case bias of \(\|\alpha_K - \Pi_K \alpha_K\|^{(2\nu + 1)/(2\nu + 2)}_{L^2_K(X)}\).

We now present the main theorem of this section. In the statement below, we recall that \(n_{\min} := \min_{k \in [K]} n_k\) and $\Delta_N 
:= (n_{\min}N)^{-\gamma/(2\gamma+1)} 
   + \sqrt{\tfrac{\log K}{n_{\min}}}\,N^{-\gamma/(2\gamma+1)}.$

\begin{enumerate}[label=\bf{A\arabic*)}, ref={A\arabic*}, resume=condA]
   \item (\textit{Realizability}) \label{cond::funclasscontainsdual} 
   For each sufficiently large \(K < \infty\) and all \(\lambda \ge 0\) small enough, 
   the functions \(\beta_K(\lambda)\) and \(\beta_K\) lie in 
   \(\mathcal{F}_{\mathrm{dual}}\).
\end{enumerate}

\begin{theorem}[Convergence rates]
\label{theorem::weakstrongratessup}
Assume \ref{cond::continousfun},  
\ref{cond::regularityOnActionSpace} with \(\mathcal{F} := \mathcal{F}_{\mathrm{dual}}\),  
\ref{cond::source}, and \ref{cond::funclasscontainsdual}.  
Choose \(\lambda_K\) so that $\lambda_K^{\nu+1} \asymp \Delta_N + \|\alpha_K - \Pi_K \alpha_K\|_{L^2_K(X)}.$
Then
\begin{align*}
\|\mathcal{T}_K(\widehat{\beta}_K(\lambda_K) - \beta_K)\|_{L^2_K(Z)}
&= O_p\!\Bigl(
N^{-\gamma/(2\gamma+1)}
+
\bigl[\Delta_N + \|\alpha_K - \Pi_K \alpha_K\|_{L^2_K(X)}\bigr]^{(1+2\nu)/(2+2\nu)}
\Bigr),\\[0.6em]
\|\widehat{\beta}_K(\lambda_K) - \beta_K\|_{L^2_K(X)}
&= O_p\!\Bigl(
\lambda_K^{-1/2}\Bigl[
N^{-\gamma/(2\gamma+1)}
+
\bigl[\Delta_N + \|\alpha_K - \Pi_K \alpha_K\|_{L^2_K(X)}\bigr]^{(1+2\nu)/(2+2\nu)}
\Bigr]\Bigr)\\[0.4em]
&= O_p\!\Bigl(
\bigl[\Delta_N + \|\alpha_K - \Pi_K \alpha_K\|_{L^2_K(X)}\bigr]^{\nu/(\nu+1)}
\Bigr).
\end{align*}
\end{theorem}

The weak and strong norm convergence rates in Theorem~\ref{theorem::weakstrongratessup} mirror those in Theorem~\ref{theorem::weakstrongratessupprimary}, up to dependence on the approximation error \(\|\alpha_K - \Pi_K \alpha_K\|_{L^2_K(X)}\). When \(\|\alpha_K - \Pi_K \alpha_K\|_{L^2_K(X)} = O_p\big(\Delta_N)\), we recover the rate in Theorem~\ref{theorem::weakstrongratessupprimary}.

\section{Large-sample theory and inference}

\label{sec::theory}

\subsection{Asymptotic normality under asymptotic identification}

\label{sec::theory1}

We now establish the asymptotic normality of the DML--JIVE estimator
\( \widehat{\psi}_K \) proposed in Section~\ref{sec::estimators} for both the
approximating functional \( \psi_K(h^\star) \) and, under suitable conditions,
the true functional \( \psi(h^\star) \). The efficiency of our estimators is
presented in Section~\ref{sec:efficiency}. The main result relies on the
following conditions.

Let \(\sigma_K^2 := \sum_{k=1}^K \frac{n_k}{N} \, E_K\left[\varphi_K(O_k)^2\right] = \sum_{k=1}^K \frac{n_k}{N} \, q_K^2(k) \, \operatorname{Var}_K[Y_k - h_K(X_k)]\) denote the variance of the influence function \(\varphi_K\) defined in Corollary~\ref{cor::EIF}. We also define the minimum group sample size as \(n_{\min} := \min_{k \in [K]} n_k\).

\begin{enumerate}[label=\bf{D\arabic*)}, ref={D\arabic*}, series=cond2]
      \item \label{cond::grow}  \textit{(Cell size grows slowly)} $\lim n_{\min} / \log K = \infty$. 
    \item \textit{(Sample-splitting)} \label{cond::split} $\widehat{h}_K$ and $\widehat{\beta}_K$ are estimated from data independent of $\{O_{ki}: (k,i)\}$.
    \item \textit{(Boundedness)} 
          As $N \rightarrow \infty$,  $ \max_{k,i} |Y_{ki} - \widehat{h}_K(X_{ki})|= O_p(1)$ and $\|\widehat{\beta}_K\|_{\infty} = O_p(\sigma_K)$. \label{cond::bound}
    \item  \textit{(Lyapunov condition)} $\sum_{k=1}^K \frac{n_k}{N} \, E_K\left[\varphi_K(O_k)^{2 + \delta} \right]   = O(\sigma_K^{2 + \delta})$  with $\liminf_{K \geq 1} \sigma_K^2 > 0$ \label{cond::lyaponov}
    \item \label{cond::nuisrate} \textit{(Nuisance estimation rates)} All of the following hold:
    \begin{enumerate}[label={(\roman*)}, ref={\ref{cond::nuisrate}(\roman*)}]
        \item $\| \mathcal{T}_{K}( \widehat{\beta}_K - \beta_{K}) \|_{L^2_K(Z)} = o_p(1)$ and $\|  \widehat{h}_K - h_K\|_{L^2_K(X)} = o_p(1)$. \label{cond::consistency}
        \item   \(\big \langle \mathcal{T}_{K}( \widehat{\beta}_K - \beta_{K}) , \mathcal{T}_{K} \big( \widehat{h}_K - h_K \big) \big \rangle_{L^2_K(Z)}  = o_p(\sigma_K N^{-1/2})\). \label{cond::nuisrates}
        \item \label{cond::nuisrates2} \(\big\langle \alpha_K -  \Pi_{K} \alpha_K, \widehat{h}_K  -  h_K  \big\rangle_{L^2_K(X)} = o_p(\sigma_K N^{-1/2})\).
    \end{enumerate}
    \item \textit{(Identification gap vanishes)}  $\langle \alpha_K - \Pi_{K} \alpha_K, \Pi_{K} h^\star - h^\star \rangle_{L^2_K(X)} = o_p(\sigma_K N^{-1/2})$. \label{cond::approxrate}
\end{enumerate}
 
 \begin{theorem}[Asymptotic linearity for cross-fold DML-JIVE]
 Under \ref{cond::continousfun},  \ref{cond::grow}-\ref{cond::nuisrate}, we have  \(\widehat{\psi}_K - \psi_K(h^\star) = \frac{1}{N} \sum_{k=1}^K \sum_{i=1}^{n_k}  \varphi_K(O_{ki}) + o_p(\sigma_K N^{-1/2})\),  and  \(\sqrt{N/\sigma_K^{2}} \big( \widehat{\psi}_K - \psi_K(h^\star) \big) \overset{d}{\to} \text{N}(0,1) \text{ as } N \to \infty\).  If \ref{cond::approxrate} also holds, then we have \(\sqrt{N/\sigma_K^{2}} \big( \widehat{\psi}_K - \psi(h^\star) \big) \overset{d}{\to} \text{N}(0,1) \text{ as } N \to \infty\). 
 \label{theorem::asymnormal}
\end{theorem}

Theorem~\ref{theorem::asymnormal} establishes that the DML--JIVE estimator
converges in distribution at rate \(\sqrt{N/\sigma_K^2}\). In well-identified
settings, where \(\sup_K \sigma_K^2 < \infty\), the estimator is therefore
\(\sqrt{N}\)-consistent. Under weak-instrument asymptotics
\citep{bound1995problems}, however, the limiting variance \(\sigma_K^2\) may
diverge as \(K \to \infty\), leading to convergence rates slower than
\(\sqrt{N}\). We note that the rate requirements are easier to satisfy in such
regimes, since \(\sigma_K N^{-1/2}\) converges more slowly than \(N^{-1/2}\).

Condition~\ref{cond::grow} requires each cell size \(n_k\) to grow slightly
faster than \(\log K\), which under \(n_k \asymp N/K\) is equivalent to
\(N/(K \log K) \to \infty\). Condition~\ref{cond::nuisrates} imposes the
standard doubly robust requirement that the nuisance estimators converge to
their minimum-norm limits in the weak norm, holding whenever
\(\|\mathcal{T}_{K}(\widehat{\beta}_K - \beta_{K})\|_{L^2_K(Z)}\,
\|\mathcal{T}_{K}(\widehat{h}_K - h_{K})\|_{L^2_K(Z)}
= o_p(N^{-1/2}\sigma_K)\). By contrast, asymptotic normality of the approximating functional
\(\psi_K(h^\star)\) requires Condition~\ref{cond::nuisrates2}, which demands
\(\|\alpha_K - \Pi_K \alpha_K\|_{L^2_K(X)}\,
\|\widehat{h}_K - h_{K}\|_{L^2_K(X)} = o_p(\sigma_K N^{-1/2})\). This dependence
on \(\|\alpha_K - \Pi_K \alpha_K\|_{L^2_K(X)}\) can be relaxed using the
modified estimator \(\widehat{\psi}_K^*\), which directly corrects the bias
arising from \(\alpha_K - \Pi_K \alpha_K\). As
strong-norm convergence of \(\widehat{h}_K\) is typically slower than
weak-norm convergence unless a spectral smoothness (source) condition holds,
verifying \ref{cond::nuisrates2} may require relatively fast decay of
\(\|\Pi_K \alpha_K - \alpha_K\|_{L^2(X)}\). Consequently, inference for the
approximating functional \(\psi_K(h^\star)\) is more challenging under
violations of point identification, as nontrivial strong-norm rates may be needed.

Condition~\ref{cond::approxrate} guarantees asymptotic identification and, by
Theorem~\ref{theorem::approx}, ensures that the identification gap
\(\psi_K(h^\star) - \psi(h^\star)\) decreases sufficiently quickly to obtain
asymptotic normality for the true target \(\psi(h^\star)\). As discussed in
Section~\ref{sec::identification}, this requirement is nontrivial: although
\(\|\Pi_K \alpha_K - \alpha_K\|_{L^2(X)}\) may plausibly decay rapidly, the
vanishing of \(\|\Pi_K h^\star - h^\star\|_{L^2(X)}\) demands sufficient spectral
smoothness of \(h^\star\). When \(\|\Pi_K h^\star - h^\star\|_{L^2(X)}\) does not
converge to zero, satisfying Condition~\ref{cond::approxrate} would require
substantial smoothness of \(\alpha_K\). In the context of Example~2, this would
mean that the surrogate distribution in a new experiment must be well
approximated by the linear span of only a relatively small subset of historical
experiments. Nonetheless, even when Condition~\ref{cond::approxrate} fails, the
estimator may still be asymptotically normal for the projected target
\(\psi_K(h^\star)\) and remain consistent for the true target under
comparatively milder conditions.

Condition~\ref{cond::bound} imposes boundedness on the nuisance functions and
allows the bound for \(\widehat{\beta}_K\) to diverge at rate \(O(\sigma_K)\).
Condition~\ref{cond::lyaponov} ensures applicability of Lyapunov’s central limit
theorem. Condition~\ref{cond::split} requires sample splitting so that the
nuisance estimators are trained on data independent of the evaluation sample;
full efficiency is recovered by cross-fitting, which averages estimates across
multiple splits \citep{vanderLaanRose2011, DoubleML}. These techniques are
standard for relaxing Donsker-type restrictions on nuisance complexity.
Condition~\ref{cond::consistency} requires weak-norm consistency of
\(\widehat{\beta}_K\) and strong-norm consistency of \(\widehat{h}_K\).
Condition~\ref{cond::nuisrates} imposes a doubly robust rate condition, which
holds whenever each nuisance converges faster than \(N^{-1/4}\). When the
nuisances are estimated as in Section~\ref{sec::nuisance} and the classes
\(\mathcal{F}_{\mathrm{primal}}\) and \(\mathcal{F}_{\mathrm{dual}}\) satisfy the
entropy condition~\ref{cond::regularityOnActionSpace} with exponent
\(\gamma\), it is sufficient that \((n_{\min}N)^{-\gamma/(2\gamma+1)} =
o_p(N^{-1/2})\). This condition holds whenever \(n_{\min} = N^c\) for some
\(c > 1/(2\gamma)\). Smoother nuisance classes (\(\gamma \to \infty\)) and the
presence of a \(\nu\)-source condition both permit slower growth of
\(n_{\min}\). Under a \(\nu\)-source condition with
\(K = o\!\left(N^{2\nu/(1+2\nu)}\right)\), the same rate requirement reduces to
the much milder condition \(N^{-\gamma/(2\gamma+1)} = o_p(N^{-1/4})\), i.e.,
\(\gamma > 1/2\).

 \medskip

\noindent\textbf{Asymptotic theory for single-split estimation when \(N/K\) is bounded.}
Relaxing \ref{cond::split}, the next theorem shows that the single-split estimator \(\widehat{\psi}_K^{\diamond}\) from Section~\ref{sec::estimators2} attains asymptotic normality even when \(N/K\) remains bounded while $N \rightarrow \infty$.

 \begin{theorem}[Asymptotic linearity for single-split DML-JIVE]
 Suppose $N \rightarrow \infty$. Under \ref{cond::continousfun}, \ref{cond::split}-\ref{cond::approxrate}, we have  \(\widehat{\psi}_K^{\diamond} - \psi(h^\star) = \frac{2}{N} \sum_{k=1}^K  \sum_{i=1}^{n_k} \mathbbm{1}\{V_{ki} = 1\} \widetilde{q}_K(k)\{Y_{ki} - h_K(X_{ki})\} + o_p(N^{-1/2} \sigma_K^{\diamond 2})\), where $\widetilde{q}_K(k) := \left\{\frac{2}{n_k}\sum_{j=1}^{n_k}\mathbbm{1}\{V_{kj} = 0\}  \beta_K(X_{kj})\right\}$. Hence, \(\sqrt{N/(2\sigma_K^{\diamond 2})} \big( \widehat{\psi}_K^{\diamond} - \psi(h^\star) \big)\) converges in distribution to an \(\text{N}(0,1)\)-distributed random variable as \(K \to \infty\), where \(
\sigma_K^{\diamond 2} = \sigma_K^2  +    \frac{2K}{N} \sum_{k=1}^K \frac{N}{Kn_k} \frac{n_k}{N} \,\sigma_{\beta}^{2}(k) \,\sigma_{\epsilon}^{2}(k),
\)
with \(\sigma_{\beta}^{2}(k) = \mathrm{Var}_{K}[\beta_{K}(X_k)]\) and \(\sigma_{\epsilon}^{2}(k) = \mathrm{Var}_{K}[Y_k - h_K(X_k)]\).
\label{theorem::asymbounded}
\end{theorem}

The single-split DML-JIVE estimator \(\widehat{\psi}_{K}^{\diamond}\) is not efficient under \(\mathcal{P}^{(K)}\), as applying sample-splitting increase asymptotic variance. This variance is typically about twice that of \(\widehat{\psi}_K\), since only half of the data contribute to the empirical mean. 
The cross-fold DML-JIVE estimator \(\widehat{\psi}_K\) averages two single-split estimators obtained by swapping the folds. When some $n_k$ are bounded, the limiting distribution of $\widehat{\psi}_K$ may be a mixture of normals due to cross-fold correlations, complicating inference. Nevertheless, Theorem~\ref{theorem::asymbounded} shows that $\widehat{\psi}_K$ remains $\sqrt{N}$-consistent, as it averages two single-split one-step estimators.

\subsection{Asymptotics for the modified estimator of the approximating functional}
\label{sec::theory3}

We establish that the modified estimator \( \widehat{\psi}_K^* \) is asymptotically normal for the approximating functional \( \psi_K(h^\star) \) without requiring asymptotic identification, with the resulting bias for $\psi(h^\star)$ characterized in Theorem~\ref{theorem::approx}. When the identification gap vanishes at a suitable rate, \( \widehat{\psi}_K^* \) also delivers valid inference for \( \psi(h^\star) \) and is asymptotically equivalent to \( \widehat{\psi}_K \).

 Our main result makes use of the following conditions. We recall the EIF $\varphi_K^*$ defined in Theorem~\ref{theorem::EIF}, and denote its variance by  $\sigma_K^{2*} := \sum_{k=1}^K \frac{n_k}{N} \, E_K\left[\varphi_K^*(O_k)^2\right].$ 

\begin{enumerate}[label=\bf{E\arabic*)}, ref={E\arabic*}, series=cond2]
        \item \textit{(Boundedness)} 
          As $K \rightarrow \infty$,  $ \max_{k,i} \{Y_{ki}, \widehat{h}_K(X_{ki}), \widehat{\alpha}_K(X_{ki}), \widehat{\Pi}_K  \widehat{\alpha}_{K}(X_{ki}\}|= O_p(1)$, and $\|\widehat{\beta}_K\|_{\infty} + \|\widehat{\rho}_K\|_{\infty} = O_p(\sigma_K^*)$.   \label{cond::bound2}
    \item \textit{(Sample-splitting)} \label{cond::split2} $\widehat{h}_K$, $\widehat{\beta}_K$, $\widehat{\alpha}_K$, $\widehat{\Pi}_K \widehat{\alpha}_K$, and $\widehat{\rho}_K$ are estimated from data independent of $\{O_{ki}: (k,i)\}$.
        \item  \textit{(Lyapunov condition)} $\sum_{k=1}^K \frac{n_k}{N} \, E_K\left[\varphi_K^*(O_k)^{2 + \delta} \right]   = O(\sigma_K^{*2 + \delta})$   with $\liminf_{K \geq 1} \sigma_K^{*2} > 0$ \label{cond::lyaponov2}
    \item \label{cond::nuisrate2} \textit{(Nuisance estimation rates)} All of the following hold:
    \begin{enumerate}[label={(\roman*)}, ref={\ref{cond::nuisrate2}(\roman*)}]
        \item $\|  \widehat{h}_K - h_K\|_{L^2_K(X)}$, $\| \mathcal{T}_{K}( \widehat{\beta}_K - \beta_{K}) \|_{L^2_K(Z)}  $, $\| \mathcal{T}_{K}( \widehat{\rho}_K - \rho_{K}) \|_{L^2_K(Z)} $, $\|  \widehat{\Pi}_K  \widehat{\alpha}_{K}-  \Pi_{K}  \alpha_K  \|_{L^2_K(Z)}$, and $\|    \widehat{\alpha}_{K}   -  \alpha_K  \|_{L^2_K(Z)}$ are all $o_p(1)$.  
        \item   \(\big \langle \mathcal{T}_{K}( \widehat{\beta}_K - \beta_{K}) , \mathcal{T}_{K} \big( \widehat{h}_K - h_K \big) \big \rangle_{L^2_K(Z)}  = o_p(N^{-1/2}\sigma_K^{*})\).  
        \item \label{cond::nuisrates_new} $\big \langle   \big( \text{I} - \widehat{\Pi}_K \big) \widehat{\alpha}_{K} - \big( \text{I} - \Pi_{K} \big) \alpha_K \big), 
   \mathcal{T}_K^*\mathcal{T}_{K} \big(\widehat{\rho}_K -  \rho_K \big) -    \widehat{h}_K - h_K    \big  \rangle_{L^2_K(X)}  = o_p(N^{-1/2}\sigma_K^{*}).$ 
    \end{enumerate}
\end{enumerate}

 \begin{theorem}[Asymptotic linearity for approximate functional]
 \label{theorem::asymnormalgeneral}
Under \ref{cond::continousfun}, \ref{cond::grow}, \ref{cond::bound2}-\ref{cond::nuisrate2}, we have  \(\widehat{\psi}_K^* - \psi_K(h^\star) = \frac{1}{N} \sum_{k=1}^K   \sum_{i=1}^{n_k} \varphi_K^*(O_{ki}) + o_p(N^{-1/2}\sigma_K^{*})\),  and  \(\sqrt{N/\sigma_K^{*2}} \big( \widehat{\psi}_K^* - \psi_K(h^\star) \big) \overset{d}{\to} \text{N}(0,1) \text{ as } K \to \infty\). If $\big\langle \alpha_K - \Pi_{K} \alpha_K, \Pi_{K} h^\star - h^\star \big\rangle_{L^2_K(X)} = o_p(N^{-1/2}\sigma_K^{*2}),$ then this theorem remains valid when \(   \psi_K(h^\star) \) is replaced by \(   \psi(h^{\star}) \).  

\end{theorem}
The conditions of Theorem~\ref{theorem::asymnormalgeneral} differ from those of
Theorem~\ref{theorem::asymnormal} primarily through
Condition~\ref{cond::nuisrates_new}, which ensures that the second-order errors
from estimating additional nuisance components are asymptotically negligible. In
contrast to \(\widehat{\psi}_K\), which requires
Condition~\ref{cond::nuisrates2}, the modified estimator
\(\widehat{\psi}_K^*\) attains asymptotic normality for \(\psi_K(h^\star)\)
under Condition~\ref{cond::nuisrates_new}, a requirement that may hold even when
\(\|\alpha_K - \Pi_K \alpha_K\|_{L^2_K(X)}\) does not converge to zero.


\label{cond::nuisrates1} 

\subsection{Irregular convergence rates under weak identification}
\label{sec::weakident}

The DML--JIVE estimator $\widehat{\psi}_K^*$ converges at rate $\sqrt{N/\sigma_K^{2*}}$, where
$\sigma_K^{2*}$ is the asymptotic variance determined by the influence
function $\varphi_K^*$. When $\sigma_K^{2*}$ diverges, the effective sample size
$N/\sigma_K^{2*}$ is $o(N)$, and the estimator converges more slowly than
$\sqrt{N}$. Classical weak-instrument asymptotics \citep{bound1995problems} constitute one such
setting. More generally, $\sigma_K^{2*}$ may diverge when the $L^2$ norms of the
dual solutions $q_K$ and $r_K$, defined by $\mathcal{T}_K^*(q_K)=\Pi_K\alpha_K$ and
$\mathcal{T}_K^*(r_K)=h_K$, become unbounded. We illustrate this phenomenon with two examples; further details are provided in
Appendix~\ref{appendix::exampleweak}.

\begin{example}[Least-identified linear functional]
Let $\sigma_{\min}$ denote the smallest nonzero singular value of
$\mathcal{T}_K$, with corresponding left and right singular functions
$\psi_{\min}$ and $\varphi_{\min}$. Consider the linear functional
$\psi_K(h^\star)=\langle \varphi_{\min}, h^\star \rangle$. The associated dual
solution is $q_K=\sigma_{\min}^{-1}\psi_{\min}$, and hence
$\|q_K\|_{L^2_K(Z)}=\sigma_{\min}^{-1}$. If $\sigma_{\min}\to 0$ as
$K \to \infty$, then $\|q_K\|_{L^2_K(Z)}$ diverges.
\end{example}

\begin{example}[Weak identification in Gaussian linear IV]
Let $Z\sim\mathcal{N}(0,1)$ and $X=\pi_K Z+U$, where $U$ is independent noise. The 2SLS estimand corresponds to the dual solution $q_K(Z)=Z/\pi_K$.
Thus $\|q_K\|_{L^2_K(Z)}=1/|\pi_K|$, which diverges as $|\pi_K|\to 0$, yielding
weak identification in the classical sense of \citet{bound1995problems}.
\end{example}

\section{Efficiency theory under many--weak--instrument asymptotics}
 
\label{sec:efficiency}

\subsection{Pathwise differentiability of the approximating functional}
\label{sec::theorypathwise}

We conclude the methodological development by establishing an efficiency theory
for NPIV models under many–weak–instrument asymptotics. Unlike classical
semiparametric theory, our setting involves a triangular array with a
$K$-dependent target and potentially diverging information. This section shows
that, despite these complications, meaningful efficiency bounds can still be
derived and attained.

We begin by establishing pathwise differentiability of the approximating
functional~\(\Psi^{(K)}\), a prerequisite for the efficiency framework developed
in Section~\ref{sec::effgen}. Although the von Mises expansion in
Theorem~\ref{theorem::vonmises} suggests that \(\Psi^{(K)}(P^{(K)})\) has
efficient influence function~\(\varphi_K^*\), pathwise differentiability
requires that the mappings \(P^{(K)} \mapsto h_{P^{(K)}}\) and
\(P^{(K)} \mapsto q_{P^{(K)}}\) vary smoothly in Hellinger distance. This may
fail when small perturbations of \(P^{(K)}\) alter the singular-value structure
of the conditional expectation operator $\mathcal{T}_{P^{(K)}}$.

We next show that the score function~\(\varphi_K^*\) is indeed the efficient influence function (or gradient) 
of the pathwise derivative of~\(P^{(K)} \mapsto \Psi^{(K)}(P^{(K)})\), 
provided that the operator~\(\mathcal{T}_{P^{(K)}}\) is surjective.

\begin{enumerate}[label=\bf{C\arabic*)}, ref = C\arabic*, resume=cond]
\item \textit{Surjectivity of $\mathcal{T}_K$:} \label{cond::rightinverse}
For each \( K < \infty \),  \( \mathcal{R}(\mathcal{T}_K) = L^2_K(Z) \).
\end{enumerate}

\begin{theorem}[Pathwise differentiability]
\label{theorem::EIF}
Assume \ref{cond::continousfun} and \ref{cond::rightinverse}. Then, \( \Psi^{(K)} \) is pathwise differentiable  at \( P^{(K)} \)  with efficient influence function given, pointwise, by $ D_K^*(o^{(K)}) := \frac{1}{N} \sum_{k=1}^K \sum_{i=1}^{n_k} \varphi_K^*(o_{ki})$, where $\varphi_K^*$ represents the single-unit EIF.
\end{theorem}
Condition~\ref{cond::rightinverse} requires that the image of \(L^2_K(X)\) under
\(\mathcal{T}_K\) span the finite-dimensional space \(L^2_K(Z)\). Because
\(\mathcal{T}_K\) has a finite-dimensional codomain, surjectivity is equivalent
to injectivity of its adjoint, \(\mathcal{N}(\mathcal{T}_K^{*}) = \{0\}\). Under
this surjectivity condition, \(\mathcal{T}_K\) has full row rank, so its
Moore--Penrose pseudoinverse satisfies
\(\mathcal{T}_K^{+} = \mathcal{T}_K^{*}(\mathcal{T}_K \mathcal{T}_K^{*})^{-1}\)
and depends smoothly on \(P^{(K)}\) \citep{stewart1977perturbation}. Consequently,
\(h_K = \mathcal{T}_K^{+}\mu_K\) varies smoothly under Hellinger perturbations of
\(P^{(K)}\). When
surjectivity fails, small perturbations of \(P^{(K)}\) can change the rank of
\(\mathcal{T}_{P^{(K)}}\), causing the pseudoinverse to vary nonsmoothly
\citep{ben2003generalized} and precluding pathwise differentiability.

While surjectivity is often viewed as a strong assumption in NPIV, it can
plausibly hold in settings with a discrete instrument. When the instrument is
discrete, the treatment is typically richer than the instrument---for example, a
continuous treatment paired with a discrete instrument---making surjectivity
possible, though still nontrivial. The next lemma shows that, for a discrete
treatment \(X\), surjectivity holds only if \(X\) has at least as many levels as
the instrument and the conditional distributions of \(X\) given each instrument
level are linearly independent.

\begin{lemma}[Linear independence condition for surjectivity]
\label{lem:linear-independence}
Let $X \in \{x_1,\ldots,x_m\}$ and $Z \in \{z_1,\ldots,z_K\}$, and define 
$p^{(K)}_k(x) := P^{(K)}(X = x \mid Z = k)$.
Then Condition~\ref{cond::rightinverse} holds if and only if the conditional distributions $\{ p^{(K)}_1(\cdot),\, p^{(K)}_2(\cdot),\,\ldots,\, p^{(K)}_K(\cdot) \}$
are linearly independent in $\mathbb{R}^m$. 
Equivalently, the $K\times m$ matrix $T^{(K)}$ with entries $T_{k i}^{(K)} = p^{(K)}_k(x_i)$ has full row rank $K$.
\end{lemma}
Informally, surjectivity of \(\mathcal{T}_K\) holds if each instrument level
induces a distinct change in the distribution of \(X\), so that the resulting
conditional laws are linearly independent. Conversely, if an instrument level
produces no change in the distribution of \(X\) relative to another level,
surjectivity fails. Surjectivity is equivalent to all singular values of the
matrix \(T^{(K)}\) being nonzero; operators failing this condition lie on a
measure-zero subset of the space of linear maps \citep{horn2012matrix}. Thus,
in discrete-instrument settings, surjectivity may be a plausible
assumption. In practice, however, even if surjectivity holds, the more relevant
challenge is near dependence, where the conditional distributions are almost
linearly dependent---for example, when distinct instrument levels exert only weak
or very similar effects on \(X\). Theorem~\ref{theorem::EIF} and the
accompanying efficiency theory accommodate such settings, but the variance of
the EIF (and hence the efficiency bound) can grow substantially and may
diverge with \(K\).


\subsection{General framework}
\label{sec::effgen}

In this section, we extend semiparametric efficiency theory to the
many--weak--instrument regime, where \(K \to \infty\) and \(N = N(K) \to \infty\)
(e.g., \citealp{bickel1993efficient, van2000asymptotic}). We formalize
regularity for estimator sequences in this setting and derive lower bounds on
the asymptotic variance of any regular estimator, showing that it cannot be
smaller than the variance implied by the efficient influence function for the
parameter sequence \((\Psi^{(K)})_{K \ge 1}\). The framework applies broadly to
pathwise differentiable parameters and is not restricted to the specific
functional studied here. In the absence of pathwise differentiability, regular
root-\(N\) estimation is typically impossible, even for fixed \(K\)
\citep{van1991differentiable, van2000asymptotic, hirano2012impossibility}. For
\(\Psi^{(K)}\), this property holds under Condition~\ref{cond::rightinverse}.

Parametric efficiency theory rests on two pillars: the asymptotic theory of
statistical experiments \citep{le1972limits} and the regularity of estimators
\citep{hajek1970characterization, van2000asymptotic}. Le Cam’s third lemma
establishes local asymptotic normality (LAN), under which likelihood ratios from
a sequence of experiments can be asymptotically matched to those from a limiting
Gaussian experiment \citep{cam1960locally, le2012asymptotic}; regularity then
ensures that an estimator’s limit distribution is equivariant in this Gaussian
experiment, and Anderson’s lemma yields the convolution theorem and the
Cramér--Rao lower bound. Semiparametric efficiency theory extends this logic to
least favorable submodels whose scores coincide with the efficient influence
function \citep{stein1956efficient, bickel1993efficient}. Building on this
foundation, we develop an analogous framework for the many--weak--instrument
regime: we construct a sequence of least favorable submodels, establish LAN
along the corresponding experiments, and define a notion of regularity from
which a convolution theorem follows.

Let \(\Psi^{(K)}: \mathcal{P}^{(K)} \to \mathbb{R}\) denote a sequence of parameters indexed by \(K \in \mathbb{N}\).  
Each distribution \(P^{(K)} \in \mathcal{P}^{(K)}\) is assumed to be dominated by a measure \(\nu^{(K)}\) on \(\mathcal{O}^{(K)}\) and to admit the factorization  
\(\tfrac{dP^{(K)}}{d\nu^{(K)}}(o^{(K)}) = \prod_{k=1}^K \prod_{i=1}^{n_k} p_k^{(K)}(o_{ki})\),  
where \(p_k^{(K)}\) denotes the marginal density of \(O_{ki}\) with respect to a dominating measure \(\nu_k^{(K)}\).  We assume that each parameter \(\Psi^{(K)}\) is pathwise differentiable in the sense of \citet{bickel1993efficient}.

\begin{enumerate}[label=\textbf{F\arabic*)}, ref=F\arabic*, series=eff]
    \item \label{cond::pathwiseff} {(Pathwise differentiability)}  
    For each \(K \in \mathbb{N}\), the parameter \(\Psi^{(K)}: \mathcal{P}^{(K)} \to \mathbb{R}\) is pathwise differentiable at \(P^{(K)} \in \mathcal{P}^{(K)}\) with efficient influence function  $D_{P^{(K)}}^*(o^{(K)}) = \tfrac{1}{N} \sum_{k=1}^K \sum_{i=1}^{n_k} \varphi_{P^{(K)}}^*(o_{ki})$.
\end{enumerate}
 \noindent For each \(K \in \mathbb{N}\), we define the total Fisher information as  
\(I_K := \sum_{k=1}^K n_k\,\operatorname{Var}_{P_k^{(K)}}(\varphi_{P^{(K)}}^*(O_{ki}))\),  
and the corresponding average per-observation information as \(\sigma_K^{*2} := I_K / N\).

To begin, we introduce a version of the least favorable submodel that is
normalized by the Fisher information. This normalization ensures that the
remainder term in the quadratic mean differentiability (QMD) expansion vanishes
as \(K\) grows, allowing the asymptotic efficiency analysis to remain valid in
many--weak--instrument regimes where the Fisher information diverges.

\begin{definition}[Fisher-normalizable sequence of least favorable submodels]
\label{def:leastfavorablemain}
A sequence of one-dimensional submodels
\(\{P_t^{(K)} : |t| \le \delta\}_{K \ge 1}\), passing through
\(P^{(K)} = P_t^{(K)}|_{t=0}\) and with joint density
\(\tfrac{dP_t^{(K)}}{d\nu^{(K)}}(o^{(K)}) = \prod_{k=1}^K \prod_{i=1}^{n_k}
p_{k,t}^{(K)}(o_{ki})\), is said to be \emph{Fisher-normalizable least
favorable} for \((\Psi^{(K)}, \mathcal{P}^{(K)}, P^{(K)})_{K \ge 1}\) if:
\begin{enumerate}
\item[(i)] For each fixed \(K\) and each \(1 \le k \le K\), the submodel
\(\{p_{k,t}^{(K)} : |t| \le \delta\}\) is QMD at \(t=0\) with score function equal
to the unit-level EIF \((x,y) \mapsto \varphi^*_{P^{(K)}}(k,x,y)\).
\item[(ii)] The QMD remainder vanishes when scaled by the total information
\(I_K^{1/2}\); that is, for each fixed \(u \in \mathbb{R}\), as \(K \to \infty\)
and \(I_K \to \infty\),
\[
\sum_{k=1}^K \frac{n_k}{N} \int\!
\Bigl(
\sqrt{p_{k,\,u I_K^{-1/2}}^{(K)}(o_k)} - \sqrt{p_k^{(K)}(o_k)}
 - \tfrac{u}{2} I_K^{-1/2}\, \varphi_{P^{(K)}}^*(o_k)\sqrt{p_k^{(K)}(o_k)}
\Bigr)^2 d\nu_k^{(K)}(o_k)
= o(I_K^{-1}).
\]
\end{enumerate}
\end{definition}

Condition~(i) recovers the usual least favorable construction for fixed
\(K\). Condition~(ii), which we term \emph{Fisher normalizability}, ensures that
the localized experiment \(\{P_{u I_K^{-1/2}}^{(K)} : u \in \mathbb{R}\}\) has a
unit-norm score \(D_{P^{(K)}}^*/\|D_{P^{(K)}}^*\|_{L^2(P^{(K)})}\) and that the
QMD remainder vanishes as \(K \to \infty\). A
Fisher-normalizable least favorable sequence exists whenever $\max_{1 \le k \le K}
\operatorname{Var}_{P_k^{(K)}}\{\varphi_{P^{(K)}}^*(O_{k1})\}/N = o(1),$
and a constructive proof is given in
Appendix~\ref{appendix:existenceleastfavorable}.

Our first main result establishes that any Fisher-normalizable sequence of least
favorable submodels \( \{P_t^{(K)} : |t| \le \delta\}_{K \ge 1} \) is locally
asymptotically normal (LAN) as \(K \to \infty\). This implies that the localized
experiments \( \{P_{u I_K^{-1/2}}^{(K)} : u \in \mathbb{R}\} \) can be
asymptotically approximated by the canonical Gaussian shift experiment. Unlike
standard formulations of LAN, which typically assume i.i.d.\ observations, our
setting involves a single realization
\(O^{(K)} = \{O_{ki} : 1 \le k \le K,\ 1 \le i \le n_k\}\) drawn from
\(P^{(K)}\), forming a triangular array of independent but non-identically
distributed observations
\citep{lecam1960locally, le1972limits, lecam2012asymptotic}.

\begin{enumerate}[label=\textbf{F\arabic*)}, ref=F\arabic*, resume=eff]
    \item \label{cond::efflindeburg} {(Lindeberg-type condition)}  
    As \( N \to \infty \):
    \begin{enumerate}
        \item \(\liminf_{K \ge 0}\min_{1 \le k \le K} \operatorname{Var}_{P^{(K)}}\!\left(\varphi_{P^{(K)}}^*(O_{k1})\right) > 0\) and \(\max_{1 \le k \le K} \operatorname{Var}_{P^{(K)}}\!\left(\varphi_{P^{(K)}}^*(O_{k1})\right) = o(N)\);
        \item \(\tfrac{1}{I_K} \sum_{k=1}^K \sum_{i=1}^{n_k} 
        E_{P^{(K)}}\!\left[ \varphi_{P^{(K)}}^{*2}(O_{ki}) 
        \mathbf{1}\!\left\{ \varphi_{P^{(K)}}^{*2}(O_{ki}) > \tfrac{I_K}{t^2}\varepsilon^2 \right\} \right] \to 0\) for all \(\varepsilon > 0\).
    \end{enumerate} 
\end{enumerate}

\begin{theorem}[Local asymptotic normality]
\label{theorem::LANmain}
Suppose \ref{cond::pathwiseff} and \ref{cond::efflindeburg} hold.  
Then there exists a Fisher-normalizable sequence of least favorable submodels for \( (\Psi^{(K)}, \mathcal{P}^{(K)}, P^{(K)})_{K \ge 1} \) such that the corresponding experiments are locally asymptotically normal:
\[
\log \frac{dP_{t I_K^{-1/2}}^{(K)}(O^{(K)})}{dP^{(K)}(O^{(K)})} 
= t\,I_K^{-1/2} \sum_{k=1}^K \sum_{i=1}^{n_k} \varphi_{P^{(K)}}^*(O_{ki}) 
- \frac{t^2}{2} + o_{P^{(K)}}(1), \quad \text{as } K \to \infty.
\]
Moreover, under \( P^{(K)} \),
\[
\log \frac{dP_{t I_K^{-1/2}}^{(K)}(O^{(K)})}{dP^{(K)}(O^{(K)})}
\overset{d}{\to} \mathcal{N}\!\left(-\tfrac{t^2}{2},\, t^2\right).
\]
\end{theorem}
Theorem~\ref{theorem::LANmain} enables the direct use of Le Cam’s third lemma and local power calculations under contiguous alternatives \citep{van2000asymptotic}, forming the basis for our subsequent derivation of asymptotic efficiency bounds for regular estimators.

We next extend the notion of estimator regularity from fixed to evolving parameters.  In classical semiparametric theory, regularity requires the limiting distribution of an estimator to remain invariant under contiguous local perturbations of the data-generating law along any score in the tangent space.  
Here, we generalize this to sequences of parameters \(\{\Psi^{(K)}\}_{K \ge 1}\), requiring invariance under local perturbations of \(P^{(K)}\) along a sequence of scores \(\{s_K\}_{K \ge 1}\).  
Rather than imposing this for all possible scores, we require regularity only along the least favorable scores \(\{D_{P^{(K)}}^*\}_{K \ge 1}\).  
Finally, while classical localizations occur within Hellinger balls of radius \(N^{-1/2}\), many--weak--instrument asymptotics requires localization at radius \(\sigma_K^* N^{-1/2}\).

\begin{definition}[Regularity along least favorable submodels]
\label{def:regularity}
Let \( \{P_t^{(K)} : |t| \le \delta\}_{K \ge 1} \) be a Fisher-normalizable sequence of least favorable submodels for \((\Psi^{(K)}, \mathcal{P}^{(K)}, P^{(K)})_{K \ge 1}\) (as in Definition~\ref{def:leastfavorablemain}).  
A sequence of estimators \( \{\widehat{\psi}_K : K \in \mathbb{N} \} \) is said to be \emph{regular} along this sequence of submodels if, for each fixed \( t \in \mathbb{R} \), as \( K \to \infty \),
\[
\sqrt{N / \sigma_K^{*2}} 
\left( \widehat{\psi}_K - \Psi^{(K)}(P_{t I_K^{-1/2}}^{(K)}) \right)
\stackrel{d}{\longrightarrow} L
\quad \text{under } P_{t I_K^{-1/2}}^{(K)},
\]
for some tight limiting distribution \( L \) that does not depend on \( t \).
\end{definition}

We now present the main result of this section.
 
 \begin{theorem}[Convolution theorem and efficiency]
\label{theorem::convolution}
Assume  \ref{cond::pathwiseff}, \ref{cond::efflindeburg}, and suppose that the estimator sequence \(\{\widehat{\psi}_K : K \in \mathbb{N}\}\) is regular, in the sense of Definition~\ref{def:regularity}, along a LAN sequence of least favorable submodels for \((\Psi^{(K)}, \mathcal{P}^{(K)}, P^{(K)})_{K \ge 1}\).  
Then, under \(P^{(K)}\),
\[
\sqrt{N / \sigma_K^{*2}} 
\left(\widehat{\psi}_K - \Psi^{(K)}(P^{(K)})\right)
\xrightarrow{d} L := Z + U,
\]
where \( Z \sim \mathcal{N}(0,1) \) and \( U \) is a tight random variable independent of \( Z \) (i.e., \( Z \perp U \)).  

In particular, the asymptotic variance satisfies \(\operatorname{Var}(L) \geq 1\),  
with equality if and only if \(U = 0\) almost surely.  
Thus, \(\widehat{\psi}_K\) is asymptotically efficient over the class of regular estimators if and only if $\sqrt{N / \sigma_K^{*2}}
\bigl[\widehat{\psi}_K - \Psi^{(K)}(P^{(K)})\bigr] 
\xrightarrow{d} \mathcal{N}(0,1).$
\end{theorem}

The next result extends the classical fact that any asymptotically linear estimator whose influence function equals the EIF is both regular and efficient.

\begin{enumerate}[label=\textbf{F\arabic*)}, ref=F\arabic*, resume=eff]
    \item \label{cond::effpathwise} (Pathwise differentiability along least-favorable sequence)  
    For each fixed \(t \in \mathbb{R}\), as \(K \to \infty\),
\[
\left| \Psi^{(K)}\big(P_{t I_K^{-1/2}}^{(K)}\big) - \Psi^{(K)}(P^{(K)}) 
- t I_K^{-1/2} \left.\frac{d}{dt} \Psi^{(K)}(P_t^{(K)})\right|_{t=0} \right| 
= o\!\left(\sigma_K^{*} N^{-1/2}\right).
\]
\end{enumerate}

\begin{theorem}[Asymptotic linearity and regularity]
\label{theorem::regularitymain}
Let \( \{P_t^{(K)} : |t| \le \delta\}_{K \ge 1} \) be a LAN sequence of least favorable submodels for \( (\Psi^{(K)}, \mathcal{P}^{(K)}, P^{(K)})_{K \ge 1} \) as in Definition~\ref{def:leastfavorablemain}.  
Assume  \ref{cond::pathwiseff}-\ref{cond::effpathwise}, and suppose the estimator \( \widehat{\psi}_K \) admits the expansion
\[
\widehat{\psi}_K - \Psi^{(K)}(P^{(K)}) 
= \frac{1}{N} \sum_{k=1}^K \sum_{i=1}^{n_k} \varphi^*_{P^{(K)}}(O_{ki}) 
+ o_{P^{(K)}}\!\left(\sigma_K^* N^{-1/2}\right).
\]
Then \( \widehat{\psi}_K \) is regular along the sequence of submodels in the sense of Definition~\ref{def:regularity} and, by Theorem~\ref{theorem::convolution}, is asymptotically efficient among all such regular estimators.
\end{theorem}

Condition~\ref{cond::effpathwise} parallels the usual pathwise differentiability requirement in the fixed-\(K\) setting.  
It is equivalent to requiring that the remainder in the pathwise Taylor expansion be \(o(\|D_{P^{(K)}}^*\|_{L^2(P^{(K)})})\),  
where \(\|D_{P^{(K)}}^*\|_{L^2(P^{(K)})} = \bigl|\tfrac{d}{dt}\Psi^{(K)}(P_t^{(K)})\big|_{t=0}\bigr|\)  
is the magnitude of the pathwise derivative of \(\Psi^{(K)}\) at \(P^{(K)}\).

\subsection{Discussion}
\label{sec:efficiencyapp}

Section~\ref{sec::effgen} develops an efficiency framework for pathwise
differentiable parameters under many--weak--instrument asymptotics. For the
parameter considered here, we establish pathwise differentiability under
surjectivity of \(\mathcal{T}_K\) (Condition~\ref{cond::bound}), without
requiring injectivity. Analyzing regularity and efficiency in this idealized
case provides a benchmark for optimal performance, even as the fixed-\(K\)
efficiency bound diverges with \(K\) (i.e., \(\sigma_K^{*} \to \infty\)).

We now turn to the regularity and efficiency of our estimators, assuming
Condition~\ref{cond::rightinverse}. Combining
Theorems~\ref{theorem::EIF}, \ref{theorem::asymnormalgeneral},
\ref{theorem::convolution}, and~\ref{theorem::regularitymain}, we conclude that
the DML--JIVE estimator \(\widehat{\psi}_K^*\), introduced in
Section~\ref{sec::estimators2}, is regular and attains the semiparametric
efficiency bound for the approximating functional
\(\Psi^{(K)}(P^{(K)}) = \psi_K(h^\star)\). Under the asymptotic identification
condition
\[
\|\alpha_K - \Pi_K \alpha_K\|_{L^2_K(X)}\,
\|h^\star - \Pi_K h^\star\|_{L^2_K(X)} = o_p(\sigma_K^* N^{-1/2}),
\]
Theorems~\ref{theorem::convolution} and~\ref{theorem::asymident} imply that any
estimator that is regular for the approximating functional, in the sense of
Definition~\ref{def:regularity}, must satisfy a convolution representation,
\[
\sqrt{N / \sigma_K^{*2}}\,
\bigl(\widehat{\psi}_K^* - \psi(h^\star)\bigr)
\xrightarrow{d} L := Z + U,
\]
with \(Z\) and \(U\) as in Theorem~\ref{theorem::convolution}. Consequently,
under asymptotic identification, \(\widehat{\psi}_K^*\) is efficient for
\(\psi(h^\star)\). Because many estimators in practice target the minimum-norm
solution \citep{babii2017identification, babii2020completeness}, it is natural
to assess efficiency under regularity with respect to this approximating
solution.

The uncorrected estimator~\(\widehat{\psi}_K\), proposed in
Section~\ref{sec::estimators}, is asymptotically efficient if it is
asymptotically equivalent to~\(\widehat{\psi}_K^*\). By
Theorem~\ref{theorem::asymnormal}, this equivalence holds when the influence
function~\(\varphi_K\) converges to the efficient influence
function~\(\varphi_K^*\); that is, when
\(\|\varphi_K^* - \varphi_K\|_{L^2_K(O)} \to 0\) as \(K \to \infty\). This
condition is satisfied whenever
\(\|\alpha_K - \Pi_K \alpha_K\|_{L^2_K(X)} \to 0\) and
Condition~\ref{cond::bound} holds, in which case the limiting influence
function coincides with the efficient one. The single-split variant
\(\widehat{\psi}_K^\diamond\) remains asymptotically normal even when per-cell
sample sizes are bounded, but it is inefficient because only half of the data
are used to estimate the target functional. Whether full-sample efficiency can
be achieved in this setting remains an open question. The main difficulty is
that, given \(\beta_K\), estimating \(q_K(k) = (\mathcal{T}_K \beta_K)(k)\)
requires empirical averages of \(\beta_K(X)\) among units with \(Z_i = k\), for
which consistent estimation demands \(n_k \to \infty\). If \(q_K\) were smooth
across instrument levels, pooling adjacent or similar cells---for example, by
combining related experiments in our application---could, in principle, yield
consistent estimation, but such structure need not hold in general.

Finally, it remains an open question to what extent the surjectivity condition
can be relaxed within a fully nonparametric model, or whether additional
structure is required to ensure regularity and efficiency. A key difficulty is
that the Moore--Penrose pseudoinverse \(\mathcal{T}_{P^{(K)}}^{+}\) need not
vary smoothly with \(P^{(K)}\) at points where surjectivity fails
\citep{stewart1977perturbation}, since perturbations of \(P^{(K)}\) need not
preserve the singular-value structure of \(\mathcal{T}_{P^{(K)}}\). In NPIV
conditional-moment models, failure of surjectivity is closely related to local
overidentification in the sense of \citet[Theorem~4.1]{chen2018overidentification}.
In such settings, restricting the model can, in principle, yield efficiency
gains via optimal weighting \citep[Theorem~3.1]{chen2018overidentification}; see
also \citet{ai2003efficient, ai2012semiparametric, florens2024optimal,
chen2025local}. However, as \citet[p.~3]{chen2018overidentification} note, local
overidentification alone does not ensure that such gains are attainable by a
feasible estimator without additional regularity conditions.


\section{Numerical Experiments}

\subsection{Simulation design}

We evaluate the DML-JIVE estimators proposed in Section \ref{sec::estimators}. We consider the problem setup of Example 1, where multiple historical experiments are used as weak instruments to construct a surrogate index for predicting long-term outcomes in a new experiment.

The data is generated as follows. We consider \( K \) historical experiments (instrument levels), each containing \( n \) samples, resulting in a total sample size of \( N = K \times n \). For each experiment \( k \in \{1, \ldots, K\} \), a latent confounder \( U_{ki} \) is generated from a discretized Gaussian distribution with \( M=10 \) levels and standard deviation \( \sigma_U = 0.2 \). For each sample \( i \) in experiment \( k \), we generate \( X_{ki} \mid Z_{ki} = k \sim \text{Beta}(a_k, b_k) + U_{ki} \), where the shape parameters \( (a_k, b_k) \) are generated via a reflected random walk on \( [1,8] \times [1,8] \) to induce smooth variation across experiments. The process starts by sampling \( (a_1, b_1) \) uniformly from \( [1, 8] \). For each group \( k = 2, \ldots, K \), the parameters are updated as \( a_k = a_{k-1} + \eta_k \) and \( b_k = b_{k-1} + \zeta_k \), where \( \eta_k, \zeta_k \sim \mathcal{N}(0, 0.5^2) \). The outcome \( Y_{ki} \) is generated according to the structural equation \( Y_{ki} = h^{\star}(X_{ki}) + 3 U_{ki} + \epsilon_{ki} \), where \( \epsilon_{ki} \) is independent Gaussian noise with standard deviation \( 0.3 \) and \( h^{\star}(x) = x \). A new experiment \( K+1 \) is generated with treatment \( X_{\text{new}, i} \sim \text{Beta}(a_{\text{new}}, b_{\text{new}}) + U_{\text{new}, i} \) for \( 1 \leq i \leq n_{\text{new}} := 50000 \). The functional of interest is the mean outcome under the new experiment, defined as 
\(
\psi(h^{\star}) = \frac{1}{n_{\text{new}}} \sum_{i=1}^{n_{\text{new}}} h^{\star}(X_{\text{new}, i}).
\)

We evaluate three estimators of \( \psi(h^{\star}) \): the cross-fold and single-fold DML-JIVE estimators proposed in Section \ref{sec::estimators}, and the adversarial NPIV DML estimator of \cite{bennett2023variational} that does not perform a JIVE correction. To assess the impact of npJIVE on nuisance estimation, we also evaluate the plug-in and inverse probability weighted (IPW) outcome estimators based on the primal and dual inverse problems, both with and without a JIVE correction. For the cross-fold and single-fold DML-JIVE estimators, the structural function and debiasing nuisance are estimated using the npJIVE procedures described in Sections \ref{sec::nuisance1} and \ref{sec::nuisance2}. For the standard DML estimator, we use non-JIVE variants of these procedures, specifically the adversarial estimators of \cite{bennett2023variational} and \cite{bennett2023inference}, which generalize 2SLS to the NPIV setting.

The function classes \( \mathcal{F}_{\mathrm{primal}} \) and \( \mathcal{F}_{\mathrm{dual}} \) used for optimization consist of functions of bounded total variation \citep{vogel1996iterative}, a flexible class that imposes only a global regularity constraint, allowing for discontinuous functions. Regularized empirical risk minimization is implemented via the Highly Adaptive Ridge \citep{schuler2024highly}, a reproducing kernel Hilbert space variant of the Highly Adaptive Lasso \citep{benkeser2016highly, bibaut2019fast}, which applies ridge penalization to indicator spline functions. For all estimation procedures, regularization parameters for the ridge penalty and Tikhonov regularization are selected using cross-validation based on the unregularized cross-fold risk estimators proposed in Section \ref{sec::estimators}.

\begin{figure}[htb!]
    \centering

    \begin{subfigure}[b]{0.5\textwidth}
        \centering
        \adjustbox{valign=t}{\includegraphics[width=\textwidth]{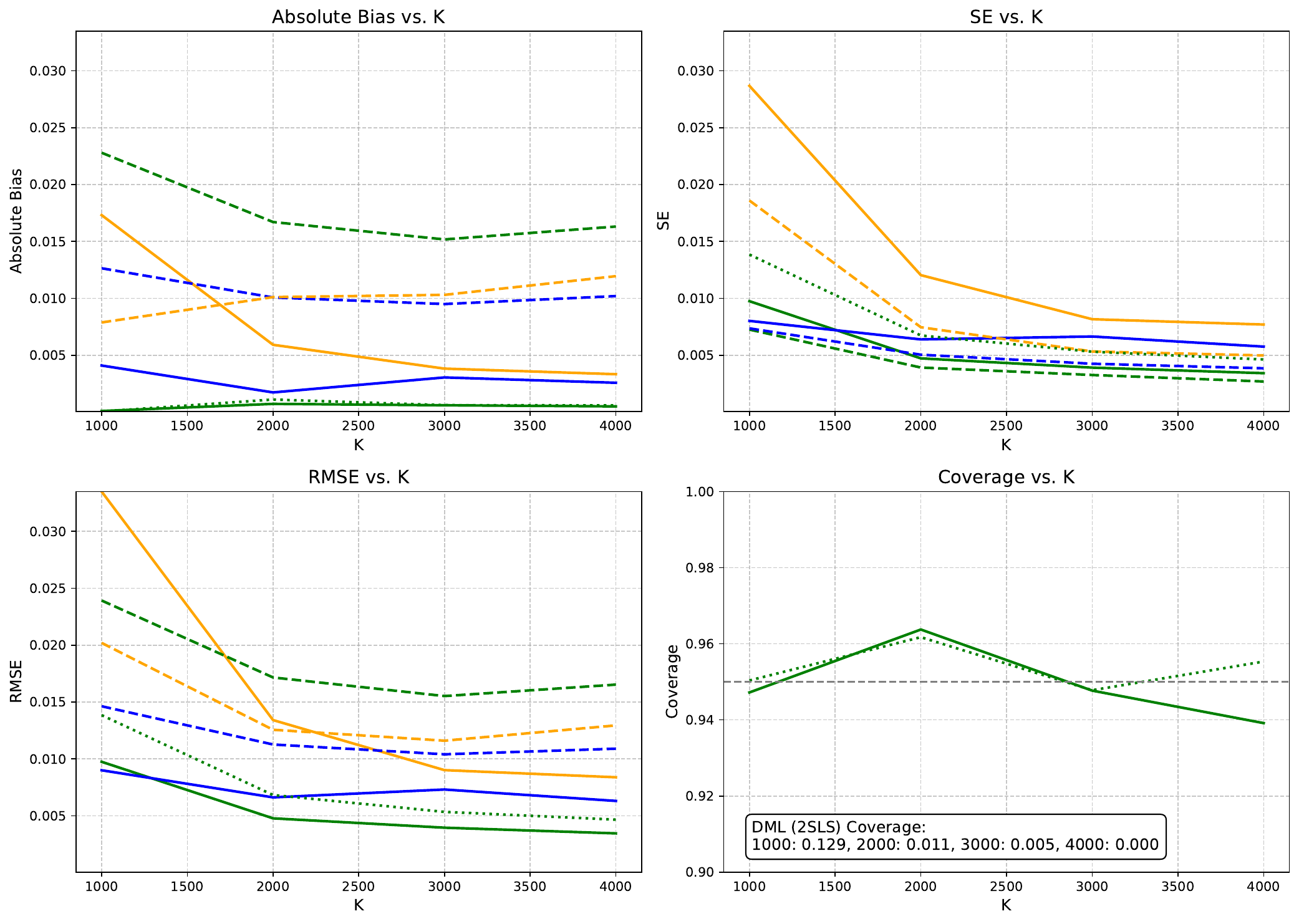}}
        \caption{$n = 30, a_{\text{new}} = 2$, $b_{\text{new}} = 3.5$}
    \end{subfigure}\begin{subfigure}[b]{0.5\textwidth}
        \centering
        \adjustbox{valign=t}{\includegraphics[width=\textwidth]{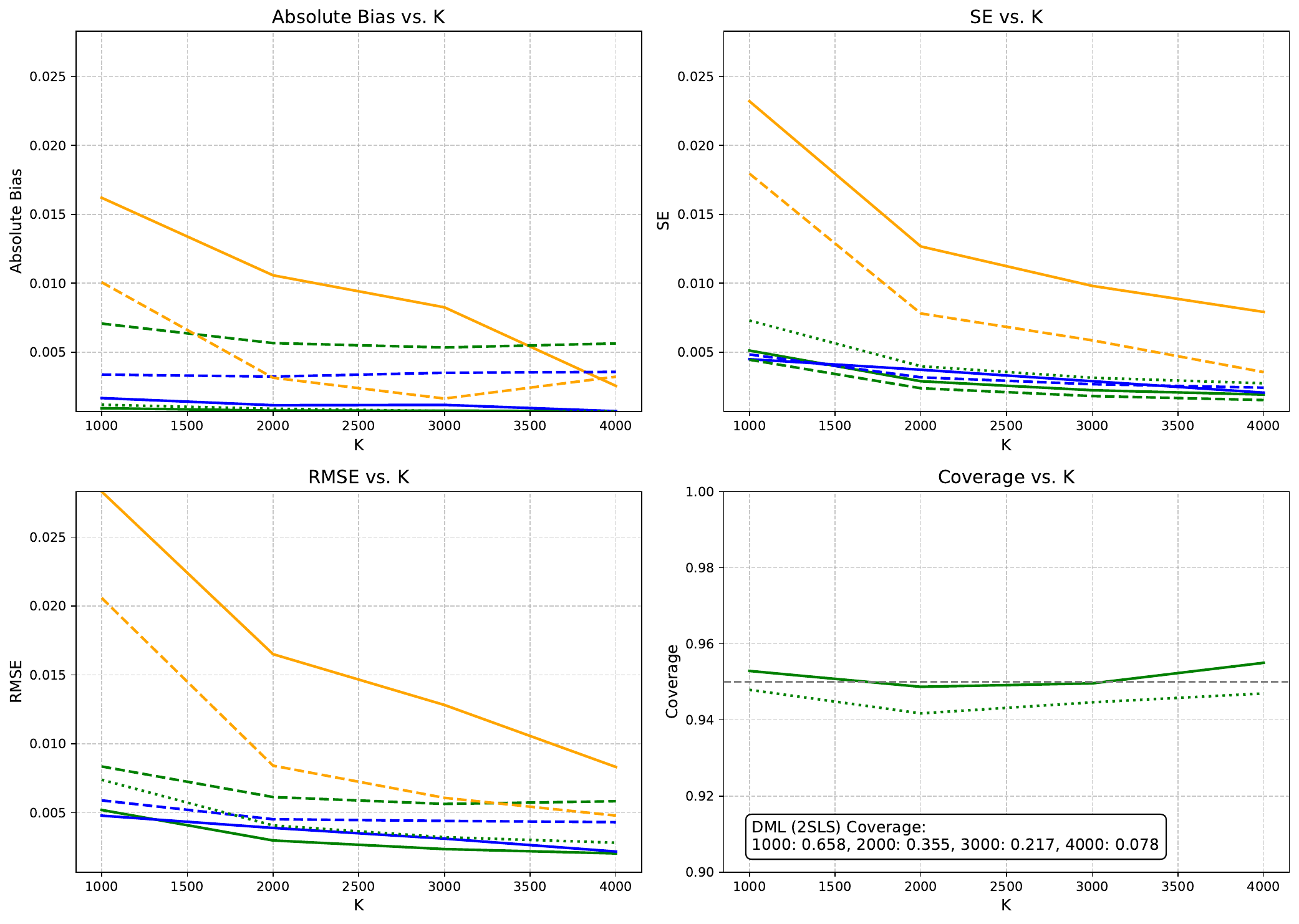}}
        \caption{$n = 100, a_{\text{new}} = 2$, $b_{\text{new}} = 3.5$}
    \end{subfigure}
  \begin{subfigure}[b]{0.5\textwidth}
        \centering
        \adjustbox{valign=t}{\includegraphics[width=\textwidth]{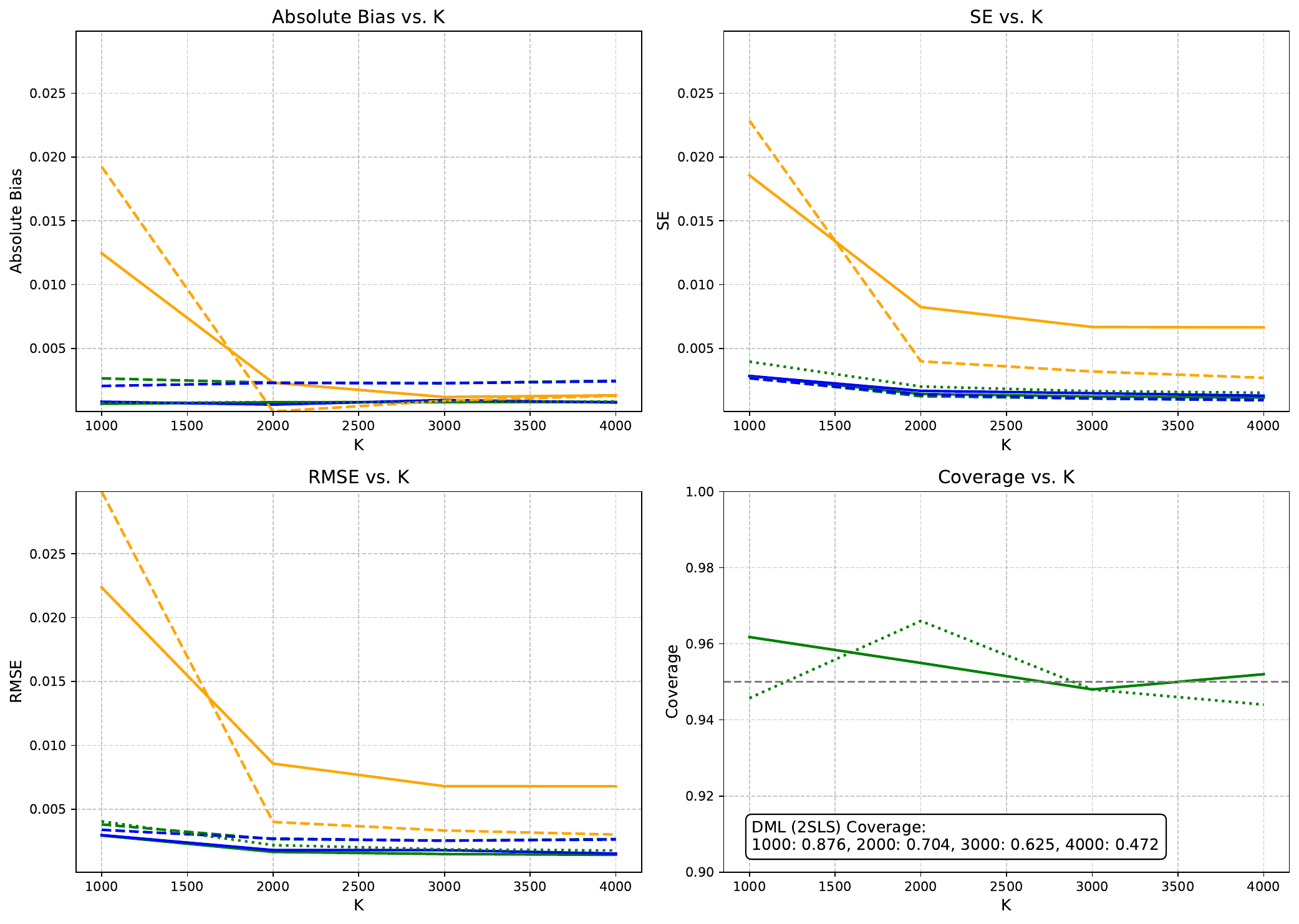}}
        \caption{$n = 300, a_{\text{new}} = 2$, $b_{\text{new}} = 3.5$}
    \end{subfigure}
      \includegraphics[width=0.6\textwidth]{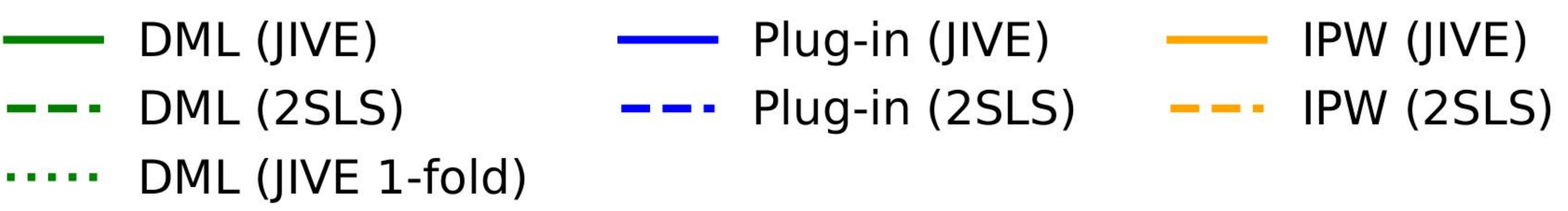}
\caption{Monte Carlo evaluation of estimation performance for the proposed cross-fold DML-JIVE estimators, compared to single-fold JIVE DML, adversarial NPIV DML (2SLS), and plug-in and IPW estimators with and without cross-fold splitting. The plots show absolute bias, standard error (SE), root mean squared error (RMSE), and 95\% confidence interval coverage for various $K$ and $n$. DML (2SLS) coverage values for each $K$ are reported in the coverage plots.}
    \label{fig:exp1}
\end{figure} 
 
\subsection{Simulation Results}

Figure \ref{fig:exp1} displays Monte Carlo estimates of the absolute bias, standard error (SE), root mean squared error (RMSE), and 95\% confidence interval coverage for the cross-fold (JIVE) and single-fold (JIVE 1-fold) DML-JIVE estimators, as well as the baseline adversarial NPIV estimator (2SLS) that does not use cross-fold splitting. The figure also includes results for the plug-in and IPW estimators with and without cross-fold splitting (JIVE vs. 2SLS). Monte Carlo estimates are plotted for increasing numbers of instruments \( K = 1000, 2000, 3000, 4000, 5000 \) and group-specific sample sizes \( n = 30 \), \( 100 \), and \( 300 \).

The cross-fold  DML-JIVE estimator consistently achieves lower bias and RMSE compared to the single-fold JIVE DML and adversarial NPIV (2SLS) estimators, particularly when the number of instruments \( K \) is large. Incorporating cross-fold splitting in the plug-in and IPW estimators also improves performance over their non-cross-fold counterparts, demonstrating the advantage of using JIVE. The coverage of the 95\% confidence intervals for the DML (2SLS) estimator deteriorates rapidly as \( K \) increases, especially for smaller sample sizes (\( n = 30 \)), where coverage rates drop to zero for \( K = 4000 \) and fall below 10\% for \( K \geq 2000 \). In contrast, the cross-fold DML-JIVE estimators maintain reliable coverage across all settings, highlighting the importance of cross-fold splitting in improving both estimation accuracy and uncertainty quantification. We find that even with sample sizes as small as \( n = 30 \), the cross-fold DML-JIVE estimator achieves nominal coverage comparable to its single-fold counterpart and better estimation performance in terms of bias, standard error, and RMSE. This suggests that the potential lack of normality of the cross-fold DML-JIVE estimator when \( n_k \not \rightarrow \infty \) is not a significant issue in this experiment.
\section{Conclusion and future work}

Our analysis develops a framework for inference on linear functionals under
many--weak--instrument asymptotics with discrete instruments, clarifying the
roles of approximate and asymptotic identification and establishing efficient
estimators for both the identified projection parameter and the true target.
Several avenues for further investigation remain. Models with discrete
covariates can be accommodated by augmenting the instrument and regressor
spaces---for example, by setting $Z=(Z,W)$ and $X=(X,W)$---thereby preserving the
structure of our analysis while allowing for richer designs and applications
\citep{newey2003instrumental}.

One important direction is to extend the analysis to continuous instruments. In
such settings, the range of the conditional expectation operator is typically
non-closed, so point identification and strong identification need not coincide,
and additional conditions are often required
\citep{severini2006some, darolles2011nonparametric, bennett2023inference}. A
promising approach is to define the approximating parameter as a functional of a
Tikhonov-regularized solution \citep{engl2015regularization}, which remains
pathwise differentiable under broad conditions. Developing a full efficiency
theory for discrete and continuous NPIV models without imposing operator
surjectivity remains an important direction for future work
\citep{ai2003efficient, khan2010semiparametric, ai2012semiparametric,
chen2018overidentification}.

Our efficiency framework also suggests a general strategy for establishing
regularity and efficiency of irregular functionals that can be approximated by
sequences of pathwise differentiable parameters
\citep{bibaut2017data, luedtke2024one}. While we focus on efficiency theory for
regular estimators, extending this analysis to local asymptotic minimax bounds
is another promising direction \citep{hajek1972local}.

Nonlinear smooth functionals and functionals depending on additional nuisance
components can be handled through linearization and nuisance influence-function
corrections \citep{van2025automatic2}. Finally, our framework may extend beyond
NPIV to other ill-posed inverse problems, including proximal causal inference
with discrete proxies \citep{tchetgen2020introduction, cui2024semiparametric}.

\bibliography{biblio}

\appendix

\section{Background on confounding-robust surrogate indices from many weak experiments}\label{sec: surrogates}

We illustrate how NPIV methods can be used to combine many historical experiments to construct surrogate indices for long-term causal inference in a new experiment. Long-term effects are often of primary scientific interest. Examples include the effect of early-childhood education on lifetime earnings \citep{chetty2011does}, promotions on long-term customer value \citep{yang2020targeting}, and digital platform design on long-term user retention \citep{hohnhold2015focusing}. Although randomized experiments remain the gold standard, the long delay between treatment assignment and the realization of long-term outcomes often means that, even when interventions can be randomized, the outcome of interest is not observed.

This gap poses significant challenges and, as noted by \citet{budish2015firms} in the context of cancer drug development, can even distort incentives by encouraging scientists to focus on interventions that can be evaluated using short-run outcomes. Nevertheless, many relevant post-treatment outcomes are available much earlier. For instance, in AIDS treatment trials, short-term viral loads and CD4 counts are observed well before mortality outcomes \citep{fleming1994surrogate}. In digital experimentation, short-term engagement metrics are available long before retention or revenue outcomes materialize.

A natural strategy for leveraging these short-term observations is to construct a \emph{surrogate index}, which imputes unobserved long-term outcomes using predictions derived from multiple short-term surrogate measures.

In what follows, we use $S$ to denote the surrogate variable (corresponding to the treatment $X$ in the main text).

\subsection{Background: statistical surrogate indices}

A widely used approach assumes that short-term outcomes form a \emph{statistical surrogate} \citep{prentice1989surrogate}. This requires that the long-term outcome be conditionally independent of treatment given the short-term outcomes. Two restrictions are implicit: there is no unobserved confounding between short- and long-term outcomes, and all of the treatment’s long-term effect is mediated by the short-term outcomes. Because the latter becomes more plausible as one includes additional short-term outcomes to capture a larger share of the treatment’s effects, \citet{athey2019surrogate} propose combining many short-term observations into a surrogate index. Their construction assumes statistical surrogacy and uses historical data to regress long-term outcomes on short-term measures (or to adjust via weighting). The method is simple, effective, and widely adopted.

However, even if short-term outcomes fully mediate the long-term effect, they need not satisfy statistical surrogacy. Consider the causal diagram in Figure \ref{fig:simple}, where $S$ perfectly mediates the effect of $A$ on $Y$ (i.e., the exclusion restriction holds) but $S$ and $Y$ share an unobserved confounder $U$, while the treatment $A$ itself is unconfounded. In this setting, $S$ is a collider: conditioning on it opens a path from $A$ to $Y$ through $U$, violating surrogacy and undermining analyses based on methods such as \citet{athey2019surrogate}.

Such scenarios are common. For instance, on subscription-based streaming platforms, a user's amount of free time affects both short-term engagement and long-term retention in the same direction. Failing to account for this confounding leads to surrogate-based estimates that overstate true long-term effects. More extreme situations can arise when an intervention substantially increases short-term engagement for a subgroup unlikely to unsubscribe while slightly decreasing engagement for a subgroup highly likely to unsubscribe. This may produce an overall increase in short-term engagement but a decrease in long-term retention---a phenomenon known as the surrogate paradox \citep{elliott2015surrogacy}.

\subsection{Experiments as instruments and surrogates as proxies}

\citet{athey2019surrogate} study a setting in which historical data from prior experiments include only $S$, $Y$, and baseline covariates. In that context, one must either worry about unobserved confounders or assume that the observed covariates are sufficient to ensure ignorability between $S$ and $Y$. In organizations that routinely run many digital experiments, however, historical data from past experiments also include the corresponding randomized treatments~$A$. In the setting of Figure \ref{fig:simple}, these treatment assignments act as instrumental variables (IVs), allowing us to identify the causal effect of $S$ on $Y$ and, in turn, to infer the long-term effect of a novel treatment by using only its effect on~$S$.

\begin{figure}[b]
\centering
\caption{Causal diagrams for surrogate settings with unobserved confounders. Dashed circles ($U$) indicate unobserved variables. Dotted circles ($Y$) indicate variables observed historically, but unobserved for novel treatments.}
\centering
\begin{tikzpicture}
\node[draw, circle, text centered, minimum size=0.75cm, line width= 1] (a) {$A$};
\node[draw, circle, right=1 of a, text centered, minimum size=0.75cm, line width= 1] (s) {$S$};
\node[draw, circle, above right=0.5 and 0.325 of s,text centered, minimum size=0.75cm, dashed,line width= 1] (u) {$U$};
\node[draw, circle, right=1 of s, text centered, minimum size=0.75cm, dotted, line width= 1] (y) {$Y$};
\draw[-latex, line width= 1] (a) -- (s);
\draw[-latex, line width= 1] (s) -- (y);
\draw[-latex, line width= 1] (u) -- (s);
\draw[-latex, line width= 1] (u) -- (y);
\end{tikzpicture}\vspace{1em}
\caption{A setting with unconfounded treatment ($A$) but confounded surrogate ($S$) and outcome ($Y$).}\label{fig:simple}

\end{figure}

\subsection{Causal model}

  We now introduce a causal model tailored to the surrogate setting under study, and show how the NPIV problem aids in identifying a causal effect.

\paragraph{Potential outcomes.}
We state the causal model in terms of potential outcomes. The model is summarized in \cref{fig:simple}, which depicts a causal diagram consistent with our assumptions. We posit the existence of random variables $\widetilde A$, $U$, $S(a)$, and $Y(s)$, where $S(a)$ and $Y(s)$ denote potential outcomes, and define $Y(a) = Y(S(a))$.

We let $\widetilde A \in \{1,\ldots,K\} \cup \{\new\}$, where $\widetilde A \in \{1,\ldots,K\}$ indexes the historical experiments and $\widetilde A=\new$ denotes assignment to the novel treatment (we return to the data-generating process after introducing the causal model). Our goal is inference on the average long-term outcome under the novel treatment,
\[
\theta_0 = E_0\big[Y(\new)\big].
\]

Our causal model is defined by the following assumptions.

\begin{assumption}[Structural relationship]\label{asm:structural_relationship}
There exists a function $h^\star$ such that, for all $a \in \{1,\ldots,K\}$,
\[
E_0\!\left[h^\star\!\big(S(a)\big)\right] = E_0\!\left[Y(a)\right].
\]
\end{assumption}

\begin{assumption}[Independence of potential outcomes given the unmeasured confounder]\label{asm:indep_pot_outcomes}
For every $s$ and $a$,
\[
Y(s)\independent S(a) \mid U.
\]
\end{assumption}

\begin{assumption}[Randomization]\label{asm:randomization}
$\widetilde A$ is independent of all potential outcomes and of~$U$.
\end{assumption}

Assumption~\ref{asm:indep_pot_outcomes} captures the presence of unmeasured confounding between short- and long-term outcomes. Assumption~\ref{asm:randomization} reflects the randomization of historical experiments, which ensures that treatment assignment serves as a valid instrument.

\paragraph{Example.}
Our assumptions are satisfied when potential outcomes arise from the nonparametric structural causal model (SCM)
\begin{align*}
Y &= f_Y(S, \epsilon_Y) + g_Y(U, \epsilon_Y),\\
S &= f_{S}(\widetilde A, U, \epsilon_{S}),\\
U &= f_{U}(\epsilon_{U}),\\
\widetilde A &= f_{A}(\epsilon_{A}),
\end{align*}
where $\epsilon_Y$, $\epsilon_S$, $\epsilon_U$, and $\epsilon_A$ are mutually independent. Assumption~\ref{asm:structural_relationship} holds here because $U$ and $S$ enter the outcome equation additively and therefore induce an additive separability that ensures the existence of a function~$h^\star$ satisfying the structural relationship.

\subsection{Data}

Let $\widetilde S = S(\widetilde A)$ and $\widetilde Y = Y(\widetilde A)$ denote the factually observed short- and long-term outcomes.

\paragraph{Historical data.}
Let $A$ be uniformly distributed over $\{1,\ldots,K\}$, and let $(Y,S)$ satisfy
\[
(Y,S) \mid A = a \;\sim\; (\widetilde Y, \widetilde S) \mid \widetilde A = a.
\]
Define $O = (A,S,Y)$. We observe $N$ historical units with observations $O_i = (A_i, S_i, Y_i)$, $i=1,\ldots,N$, where each $O_i \sim O$. The treatment assignments $A_1,\ldots,A_N$ arise from a completely randomized design in which each treatment arm contains exactly $n$ units. For any $i \neq j$, the pairs $(S_i,Y_i)$ and $(S_j,Y_j)$ are independent conditional on $(A_i, A_j)$.

\paragraph{Novel data.}
Let $S^{\new} \sim \widetilde S \mid (\widetilde A = \new)$. We observe $n'$ independent copies $S^{\new}_1,\ldots,S^{\new}_{n'}$ of $S^{\new}$.

\medskip

We now turn to the question of identification.  
If we observed $Y^{\new} \sim \widetilde Y \mid (\widetilde A = \new)$ for units assigned to the novel treatment, then by Assumption~\ref{asm:randomization} we would identify $\theta_0$ trivially as $E_0[Y^{\new}]$. In the absence of long-term outcomes for the novel treatment, our goal is instead to combine the historical and novel data to identify $\theta_0$.

\subsection{Surrogacy and set identification}

Under the structural relationship in Assumption~\ref{asm:structural_relationship}, it follows immediately that $h^\star(S)$ is a valid surrogate index. In particular, $h^\star$ links causal effects on $Y$ to causal effects on $S$: if $h^\star$ were known, then the average potential outcome of $Y$ under any intervention could be recovered from the distribution of the average potential outcome of $S$ under that same intervention. We therefore define the causal parameter of interest as $\theta^\star = \theta(h^\star)$, where
\[
\theta(h) = E_0\!\left[h\big(S(\new)\big)\right]
\]
for the novel intervention $\new$.

Let $\mathcal{H} \subseteq L_2(\mathcal{S})$ denote a nuisance space. For $h \in \mathcal{H}$ and $a \in [K]$, define the operator $\mathcal{T}_K : \mathcal{H} \to ([K] \to \mathbb{R})$ by
\[
(\mathcal{T}_K h)(a) = E_0[h(S)\mid A=a].
\]

\paragraph{Set identification of $h^\star$ and point identification of $\theta_0$.}

The next lemma characterizes $h^\star$ in terms of the observed data distribution in the historical experiments.

\begin{lemma}[Set identification of $h^\star$]\label{lemma:set_id_h_star}
Under Assumptions~\ref{asm:structural_relationship}--\ref{asm:randomization}, it holds that $h^\star$ belongs to the set
\[
\mathcal{H}_0 \;=\; 
\left\{h : \mathcal{S}\to\mathbb{R} \;:\; 
E_0\!\left[h(S)\mid A=a\right] = E_0\!\left[Y\mid A=a\right]\;\; \text{for all } a\in[K]\right\}.
\]
That is, $h^\star$ satisfies the conditional moment restriction
\[
E_0\!\left[h^\star(S)\mid A=a\right] = E_0\!\left[Y\mid A=a\right],\qquad a\in[K].
\]
\end{lemma}

\subsection{Identification and strong identification}

Although $h^\star$ is not point identified a priori, it is not itself the parameter of interest; rather, our target is $\theta(h^\star)$. We begin by formalizing what it means for $\theta(h^\star)$ to be identified.

\begin{definition}[Identification]
Let $r_{0,K}$ be the function $[K]\to\mathbb{R}$ defined by
\[
r_{0,K}(a) = E_0[Y \mid A=a], \qquad a\in[K].
\]
We say that $\theta(h^\star)$ is \emph{identified} if, for any $h \in \mathcal{H}$ such that $\mathcal{T}_K h = r_{0,K}$---that is, for any $h \in \mathcal{H}_{0,K}\cap\mathcal{H}$---we have
\[
\theta(h) = \theta(h^\star).
\]
\end{definition}

Let $\rho_K$ denote the importance sampling ratio comparing the distribution of $S$ under the novel treatment ($A=\new$) with that under the historical experiments ($A\in[K]$). For $s\in\mathcal{S}$,
\[
\rho_K(s)
= \frac{p_{S(\new)}(s)}{\frac{1}{K}\sum_{a=1}^K p_{S(a)}(s)}
= \frac{p_{S^{\new}}(s)}{p_S(s)}.
\]
For any $h\in\mathcal{H}$, we may rewrite $\theta(h)$ as an expectation under the distribution of $(S,Y)$ in the historical experiments using this ratio:
\[
\theta(h)
= E_0\!\big[\rho_K(S)\, h(S)\big]
= E_0\!\big[\alpha_K(S)\, h(S)\big],
\]
where
\[
\alpha_K = \Pi\!\left(\rho_K \mid \mathcal{H}\right)
\]
denotes the $L_2$ projection of $\rho_K$ onto $\mathcal{H}$. Hence $\alpha_K$ is the Riesz representer of the linear functional $\theta$ on $\mathcal{H}$.

It follows immediately that $\theta(h^\star)$ is point identified if and only if
\[
\alpha_K \,\in\, \mathcal{N}(\mathcal{T}_K)^\perp
= \overline{\mathcal{R}(\mathcal{T}_K^*)}.
\]
Identification is therefore a strong requirement: the space $\mathcal{N}(\mathcal{T}_K)^\perp$ is $K$-dimensional, imposing substantial restrictions on the allowable complexity of~$\alpha_K$.

 \section{Additional experimental results}

 \begin{figure}[ht]
    \centering
    \begin{subfigure}[b]{0.32\textwidth}
        \centering
        \adjustbox{valign=t}{\includegraphics[width=\textwidth]{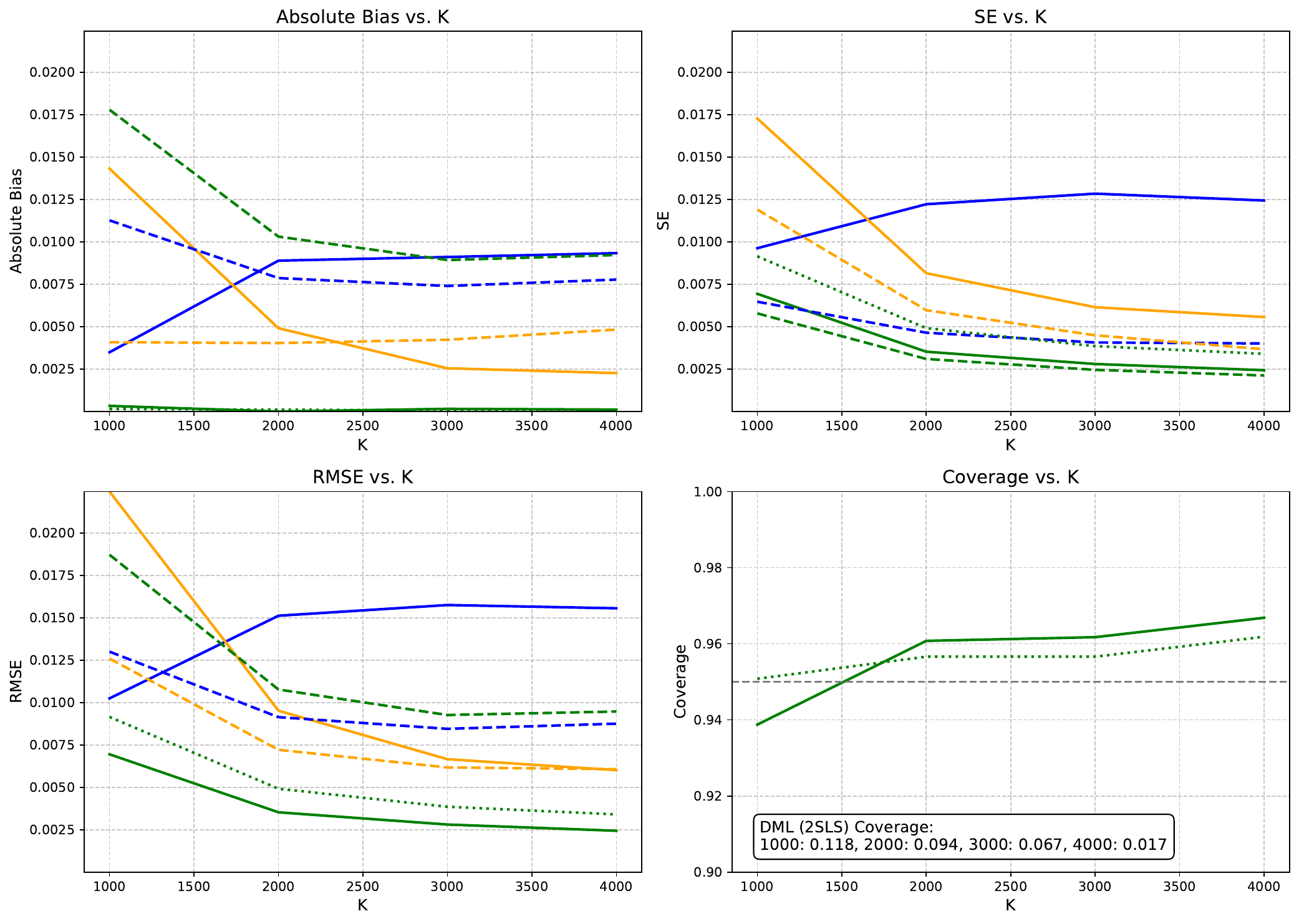}}
        \caption{$a_{\text{new}} = 2$, $b_{\text{new}} = 2.5$}
    \end{subfigure}
    \hfill
    \begin{subfigure}[b]{0.32\textwidth}
        \centering
        \adjustbox{valign=t}{\includegraphics[width=\textwidth]{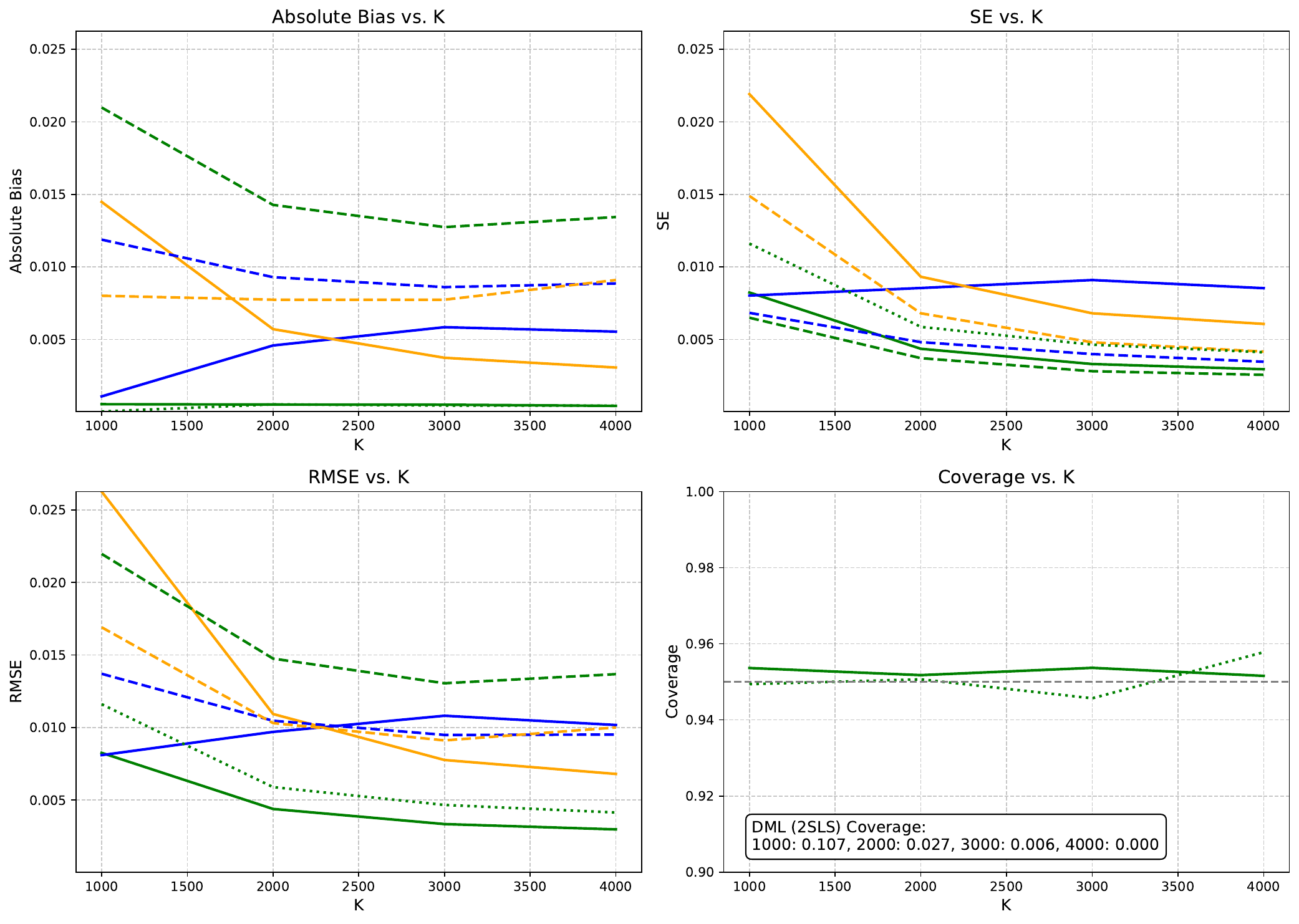}}
        \caption{$a_{\text{new}} = 2$, $b_{\text{new}} = 3 $}
    \end{subfigure}
    \hfill
    \begin{subfigure}[b]{0.32\textwidth}
        \centering
        \adjustbox{valign=t}{\includegraphics[width=\textwidth]{figures/plot_results_30_K4000_18_ab2_3.5_123.pdf}}
        \caption{$a_{\text{new}} = 2$, $b_{\text{new}} = 3.5$}
    \end{subfigure}
     \includegraphics[width=0.6\textwidth]{figures/legend_extracted.pdf}
    \caption{Comparison of Performance Metrics for Different Values of Overlap between treatment distribution in new experiment and marginal distribution of historical experiments.}
    \label{fig:comparison_ab}
\end{figure}

\section{Nuisance estimation without asymptotic identification}

\label{appendix::nuisanceestimation3}

The one-step estimator proposed in Section \ref{sec::estimators2} provides valid inference for the minimum norm projection estimand $\psi_K(h^\star)$ when asymptotic identification fails. This estimator requires estimation of additional nuisance functions to ensure debiasedness and asymptotic normality for $\psi_K(h^\star)$, specifically, the Riesz representer $\alpha_K$ of the linear functional $\psi$, its projection $\Pi_K \alpha_K$ onto $\mathcal{H}_K$, and the debiasing nuisances $\rho_K = (\mathcal{T}_K^* \mathcal{T}_K)^+ h_K$ and $\beta_K = (\mathcal{T}_K^* \mathcal{T}_K)^+ \Pi_K \alpha_K$.

The Riesz representer $\alpha_K$ can be flexibly estimated using Riesz regression \citep{chernozhukov2022automatic} based on the following empirical risk:
 $$\alpha \mapsto \frac{1}{N}\sum_{k=1}^K \sum_{i=1}^{n_k} \{\alpha(X_{ki})\}^2 - 2\psi(\alpha).$$
To estimate its projection onto $\mathcal{H}_K = \mathcal{N}(\mathcal{T}_K)^{\perp}$, we observe that $\Pi_K \alpha_K$ is the minimum norm solution to a null space constraint:
$$\Pi_K \alpha_K = \argmin_{\alpha \in  L^2_K(X)}\, \|\alpha\|_{L^2_K(X)} \quad \text{subject to} \quad \|\mathcal{T}_K(\alpha_K - \alpha)\|_{L^2_K(Z)} = 0.$$
Given an estimator $\widehat{\alpha}_K$ of $\alpha_K$, this suggests estimating $\Pi_K \alpha_K$ by minimizing the npJIVE empirical risk with Tikhonov regularization:
 \begin{align*}
  \widehat{\Pi}_K\widehat{\alpha}_K := \argmin_{\alpha \in \mathcal{F}_{\alpha}} \sum_{k=1}^K \frac{n_k}{N} \widehat{\mathcal{T}}_{K,0}(\widehat{\alpha}_K - \alpha)(k)\widehat{\mathcal{T}}_{K,1}(\widehat{\alpha}_K - \alpha)(k) + \lambda_K \|\alpha\|_{2, K}^2,
 \end{align*}
 where $\mathcal{F}_{\alpha}$ is a convex function class and the regularization parameter \( \lambda_K > 0 \) is chosen to tend to zero as \( K \to \infty \), ensuring convergence to the minimum-norm solution. Estimation rates for $ \widehat{\Pi}_K\widehat{\alpha}_K - \Pi_K \alpha_K$ in weak and strong norms akin to those of $\widehat{h}_K - h_K$ can be attained by modifying the proof of Theorem \ref{theorem::weakstrongratessupprimary}, where an additional error term arises from the first-stage estimation error of $\widehat{\alpha}_K$ for $\alpha_K$.  Alternatively, to directly estimate $\alpha_K - \Pi_K \alpha_K$, we observe that $\alpha_K - \Pi_K \alpha_K = \Pi_K^\perp \alpha_K$ is the projection onto the orthogonal complement $\mathcal{H}_K^\perp = \mathcal{N}(\mathcal{T}_K)$ and, hence, is the solution to the following constrained Riesz regression:
$$\Pi_K^\perp \alpha_K = \argmin_{\alpha \in  L^2_K(X)}\, \|\alpha\|_{L^2_K(X)}^2 - 2\psi(\alpha)\quad \text{subject to} \quad \|\mathcal{T}_K(\alpha)\|_{L^2_K(Z)}^2 = 0.$$

To estimate the debiasing nuisances $\beta_K$ and $\rho_K$, we recall from the previous section that the dual inverse problem $\mathcal{T}_K^* \mathcal{T}_K \beta_K = \Pi_K \alpha_K$ corresponds to the optimal conditions of the first order of the following risk:
\[
R_K^*(f) := \|\mathcal{T}_{K}(\beta_K)\|^2_{L^2_K(Z)} -2 \langle \Pi_K \alpha_K , \beta_K \rangle_{L^2_K(X)}.
\]
Using an estimate $\widehat{\Pi}_K\widehat{\alpha}_K$ of $\Pi_K \alpha_K$, an estimator of $\beta_K$ is given by
\[
\widehat{\beta}_K := \argmin_{\beta \in \mathcal{F}_{\mathrm{dual}}} \sum_{k=1}^K \frac{n_k}{N} \widehat{\mathcal{T}}_{K,0}(  \beta)(k)\widehat{\mathcal{T}}_{K,1}(  \beta)(k) - \sum_{k=2}^K \sum_{i=1}^{n_k} \widehat{\Pi}_K\widehat{\alpha}_K(X_{ki})\beta(X_{ki}) + \lambda_K \|\beta\|_{2, K}^2.
\]
Similarly, using that $\mathcal{T}_K^* \mathcal{T}_K \rho_K = h_K$ and an estimate $ \widehat{h}_K$ of $h_K$, an estimator of $\rho_K$ is given by
\[
\widehat{\rho}_K := \argmin_{\rho \in \mathcal{F}_{\rho}} \sum_{k=1}^K \frac{n_k}{N} \widehat{\mathcal{T}}_{K,0}(  \rho)(k)\widehat{\mathcal{T}}_{K,1}(  \rho)(k) -\frac{2}{N}\sum_{k=2}^K \sum_{i=1}^{n_k} \widehat{h}_K(X_{ki})\rho(X_{ki}) + \lambda_K \|\rho\|_{2, K}^2.
\]
where $\mathcal{F}_{\rho}$ is a convex function class. Estimation rates for $\widehat{\rho}_K $ and $\widehat{\beta}_K$ in weak and strong norms akin to those of $\widehat{h}_K$ can be attained by modifying the proof of Theorem \ref{theorem::weakstrongratessupprimary}, where an additional error term arises from the first-stage estimation errors of loss nuisances.

\section{Additional examples}

\subsection{Asymptotic identification}

 \setcounter{example}{0}
\begin{example}[continued]
For the counterfactual mean \( \psi(h^\star) = \int h^\star(x) p_X^\star(x)dx \), asymptotic identification requires that, with asymptotically vanishing error as $K \rightarrow \infty$, either (i) \( p_X^\star \) can be approximated by an element of the linear combination of factual densities \( \sum_{k=1}^K  q(k)p^{(k)}(X_k = x)\) for some \(q \in L^2_K(Z)  \), or (ii) the structural function \( h_K^{\star} \) can be approximately expressed as a convolution  $\sum_{k=1}^K q(k) \frac{p^{(k)}(X_k = x)}{\sum_{j=1}^K \frac{n_{j}}{N} p^{(j)}(X_{j} = x)}$ for some $q \in L^2_K(Z)$. The latter condition is equivalent to requiring that the product of $h_K^{\star}$ with the marginal density of $X$, given by $x \mapsto h_K^{\star}(x) \cdot \sum_{j=1}^K \frac{n_{j}}{N} p^{(j)}(X_{j} = x)$, can be approximated by a linear combination \( \sum_{k=1}^K q(k)p^{(k)}(X_k = x) \) for some \( q \in L^2_K(Z) \) with vanishing error. When the instrument \( Z \) indexes multiple historical randomized experiments, the former condition requires that the distribution of the surrogate in a novel experiment \( P_X^\star \) eventually lies in the linear span of the surrogate distributions from past experiments, up to an asymptotically vanishing error. 
\end{example}

The second example illustrates how a discrete intended treatment variable can serve as an instrument when the realized treatment consists only of small deviations.

 \setcounter{example}{2}
\begin{example}[Intended dose as an ordinal instrument]
In a multi-arm trial, individuals are randomized into \( K \) arms, with the intended dose for arm \( k \) given by \( Z_{ki} = k/K \), and the actual dose by \( X_{ki} = k/K + \varepsilon_{ki} \), where \( \varepsilon_{ki} \) is an i.i.d.\ mean-zero normal error with standard deviation \( 1/K \), arising from protocol deviations or measurement error. The actual dose may be confounded, but the intended dose is randomized and predicts $X$, making it a valid instrument. Define $\widetilde{h}^{\star}: k \mapsto h^\star(k/K)$ as an $L^2_K(Z)$ approximation of $h^\star$. Since $\Pi_K$ projects onto $\mathcal{R}(\mathcal{T}_K^*)$, we obtain $\|\Pi_K h^\star - h^\star\|_{L^2_K(X)} \lesssim \|\mathcal{T}_K^*(\widetilde{h}^{\star}) - h^\star\|_{L^2_K(X)} \lesssim \|\widetilde{h}^{\star} -  h^\star\|_{L^2_K(X,Z)}.$
If $h^\star$ is $L$-Lipschitz, then $|\widetilde{h}^{\star}(Z_{ki}) - h^\star(X_{ki})| = O_p(K^{-1})$, yielding $\|\widetilde{h}^{\star} -  h^\star\|_{L^2_K(X,Z)} = O(K^{-1})$ and $\|\Pi_K h^\star - h^\star\|_{L^2_K(X)} = O(K^{-1})$. If $\alpha_K$ is also $L$-Lipschitz, then $
\| \alpha_K - \Pi_{K} \alpha_K \|_{L^2_K(X)} \| h^\star  - \Pi_{K} h^\star\|_{L^2_K(X)} = O(K^{-2}).$ 
\end{example}
 
The next example demonstrates that if \( h^\star \) and \( \alpha_K \) are sufficiently smooth so that they asymptotically satisfy a source condition as \( K \to \infty \), then precise rates on the approximation errors can be obtained. Source conditions on solutions to inverse problems are commonly assumed to guarantee consistent and sufficiently fast estimation \citep{bennett2023source}. However, source conditions typically assumed in the literature are often violated at finite \( K \), and can only plausibly hold in the limit as \( K \to \infty \).

\begin{example}[Asymptotic source condition]
      Let $\lambda_K$ be the $K$th largest eigenvalue of $\mathcal{T}_K^*\mathcal{T}_K$. Suppose that $h^\star$ satisfies the \emph{asymptotic source condition}, meaning that there exists a sequence of functions $f_K \in L^2(X)$ and a residual term $\eta_K$ such that $h^\star = (\mathcal{T}_K^* \mathcal{T}_K)^r f_K + \eta_K$ for some $0 < r \leq 1$, where the residual satisfies $\|\eta_K\|_{L^2(X)} \lesssim \lambda_K^r$. This condition implies that, asymptotically, $h^\star$ is smooth with respect to the orthonormal eigenbasis of $\mathcal{T}_K^* \mathcal{T}_K$, with its coefficients decaying at a fractional power of the corresponding eigenvalues. Then, we can show that $\|\Pi_K h^\star - h^\star\|_{L^2(X)} \lesssim \lambda_K^r$. If, in addition, $\alpha_K$ satisfies the asymptotic source condition with exponent $0 < s \leq 1$, then  $\| \alpha_K - \Pi_{K} \alpha_K \|_{L^2_K(X)} \| h^\star  - \Pi_{K} h^\star\|_{L^2_K(X)} \lesssim \lambda_K^{r+s},$ which is $o(N^{-1/2})$ as long as $\lambda_K = o(N^{-\frac{1}{2(r+s)}}) $.
\end{example}

 \subsection{Weak identification}

\label{appendix::exampleweak}

\begin{example}[Least-identified linear functional]
Since \( \mathcal{T}_K : L^2_K(X) \to L^2_K(Z) \) is compact and injective on \( \mathcal{H}_K \), it admits a singular value decomposition with strictly positive singular values on \( \mathcal{H}_K \). Specifically, there exist orthonormal systems \( \{\varphi_j\} \subset \mathcal{H}_K \) and \( \{\psi_j\} \subset L^2_K(Z) \) such that \( \mathcal{T}_K(\varphi_j) = \sigma_j \psi_j \) and \( \mathcal{T}_K^*(\psi_j) = \sigma_j \varphi_j \) for \( j = 1, \dots, r_K \), where \( \sigma_1 \ge \cdots \ge \sigma_{r_K} > 0 \) are the nonzero singular values of \( \mathcal{T}_K \), and \( r_K = \operatorname{rank}(\mathcal{T}_K) \). For clarity, we suppress the dependence on \( K \) in the notation. Let \( \sigma_{\min} := \sigma_{r_K} \) denote the smallest nonzero singular value, with corresponding singular functions \( \varphi_{\min} \in \mathcal{H}_K \) and \( \psi_{\min} \in L^2_K(Z) \). Define the linear functional \( \psi(h_K) := \langle \varphi_{\min}, h_K \rangle_{L^2_K(X)} \). Its Riesz representer is \( \alpha_K = \Pi_K \alpha_K = \varphi_{\min} \), which belongs to \( \mathcal{H}_K \) by construction. Since \( \mathcal{T}_K^*(\sigma_{\min}^{-1} \psi_{\min}) = \varphi_{\min} \),  we have \( q_K = \sigma_{\min}^{-1} \psi_{\min} \) and, hence, \( \|q_K\|_{L^2_K(Z)} = \sigma_{\min}^{-1} \). If \( \sigma_{\min} \to 0 \) as \( K \to \infty \), then \( \|q_K\|_{L^2_K(Z)} \to \infty \).
\end{example}

\begin{example}[Weak identification in gaussian linear models]
Consider a linear IV model under a random-instrument design, where \((Z, X)\) are jointly Gaussian \citep{bound1995problems, han2006gmm, newey2009generalized}.  Assume \(Z \sim \mathcal{N}(0,1)\), \(U \sim \mathcal{N}(0, 1 - \pi_K^2)\), \(X = \pi_K Z + U\), and \(Y = b_K X + V\), with \(Z \perp U\) and \(E_K[V \mid Z] = 0\). The 2SLS estimand is \(b_K = E_K[m_K(Z) Y] / E_K[m_K(Z) X] = \pi_K^{-2} E_K[\pi_K Z\, Y]\), where \(m_K(Z) = E_K[X \mid Z] = \pi_K Z\). The slope $b_K$ corresponds to the linear functional \(\psi(h) = \pi_K^{-2} E_K[m_K(Z)\, h(X)]\), with Riesz representer \(\alpha_K(x) = \pi_K^{-2} E_K[m_K(Z) \mid X = x] = \pi_K^{-2} (\mathcal{T}_K^* m_K)(x)\). The dual solution \(q_K(Z) = Z / \pi_K\) satisfies \(\mathcal{T}_K^* q_K = \alpha_K\), and the EIF is \(\varphi_K^*(Z, X, Y) = \frac{Z}{\pi_K} \{Y - h_K(X)\}\), with norm \(\|q_K\|_{L^2(Z)}^2 = 1 / \pi_K^2\). If \(V \sim \mathcal{N}(0, \sigma^2)\), the EIF variance is \(\sigma_K^2 = \sigma^2 / \pi_K^2\). The effective sample size \(N \pi_K^2\) governs the convergence rate of estimators such as JIVE \citep{angrist1999jackknife}, with the quantity \(N \cdot \frac{\pi_K^2}{1 - \pi_K^2}\) known as the concentration parameter \citep{han2006gmm}. As \(\pi_K^2 \to 0\), we have \(\|q_K\|_{L^2_K(Z)} \to \infty\) and loss of the \(\sqrt{N}\)-rate; if \(\inf_K \pi_K^2 > 0\), then \(\|q_K\|_{L^2_K(Z)}\) remains bounded and valid \(\sqrt{N}\)-inference is achievable.
\end{example}

\begin{example}[$\sqrt{N}$-inference under source conditions]
  For each \(K < \infty\), suppose \(\Pi_K \alpha_K = (\mathcal{T}_K^* \mathcal{T}_K)^{1+s} v_K\) for some \(s \in [0,1]\) and \(\sup_K \|v_K\|_{L^2_K(X)} \le C\), which implies \(q_K = \mathcal{T}_K(\mathcal{T}_K^* \mathcal{T}_K)^s v_K\). This source condition imposes smoothness on the projected Riesz representer \(\Pi_K \alpha_K\), rather than on the true Riesz representer \(\alpha_K\), which aligns with its projection only under strong identification. It ensures that, in the eigenbasis expansion of \(\Pi_K \alpha_K\), smaller coefficients are placed on eigenfunctions corresponding to smaller eigenvalues. Let \(\{(\lambda_j, \varphi_j)\}\) denote the nonzero eigenpairs of \(\mathcal{T}_K^* \mathcal{T}_K\), ordered so that \(\lambda_1 \ge \lambda_2 \ge \cdots > 0\). Since \(q_K = \mathcal{T}_K (\mathcal{T}_K^* \mathcal{T}_K)^s v_K\), we can write \(q_K = \sum_j \lambda_j^{s + 1/2} \langle v_K, \varphi_j \rangle \psi_j\), where \(\psi_j := \lambda_j^{-1/2} \mathcal{T}_K \varphi_j\) is the corresponding orthonormal basis of \(L^2_K(Z)\). Then $\|q_K\|_{L^2_K(Z)}^2 = \sum_j \lambda_j^{2s + 1} \langle v_K, \varphi_j \rangle^2 \le \lambda_1^{2s + 1} \|v_K\|_{L^2_K(X)}^2 \le C^2 \lambda_1^{2s + 1},$
so \(\|q_K\|_{L^2_K(Z)} \le C \lambda_1^{s + 1/2}\). In particular, if \(\sup_K \lambda_1 < \infty\), then \(\sup_K \|q_K\|_{L^2_K(Z)} < \infty\). By an identical argument, an analogous source condition on \(h_K\) ensures that \(\sup_K \|r_K\|_{L^2_K(Z)} < \infty\).
\end{example}

\section{Proofs for Section \ref{sec::identification} and Section \ref{sec::theorypathwise}}

 \begin{proof}[Proof of Theorem \ref{theorem::approx}]
By properties of the pseudoinverse, it holds that \( h_K = \Pi_{K} h^\star \). By the Riesz representation theorem, we have
\begin{align*}
    \Psi^{(K)}(P^{(K)}) - \psi(h^\star) 
    &= \langle \alpha_K, h_K \rangle_{L^2_K(X)} - \langle \alpha_K, h^\star \rangle_{L^2_K(X)} \\
    &= \langle \alpha_K, h_K - h^\star \rangle_{L^2_K(X)} \\
    &= \langle \alpha_K, \Pi_{K} h^\star - h^\star \rangle_{L^2_K(X)} \\
    &= \langle \alpha_K - \Pi_{K} \alpha_K, \Pi_{K} h^\star - h^\star \rangle_{L^2_K(X)},
\end{align*}
where the final equality follows from the orthogonality condition of the projection, which ensures that
\[
\langle \Pi_{K} \alpha_K, \Pi_{K} h^\star - h^\star \rangle_{L^2_K(X)} = 0.
\]

The final statement of the theorem is a direct consequence of the Cauchy-Schwarz inequality.

 \end{proof}

\subsection{Proof of EIF and von Mises expansion}

We first derive the efficient influence function (EIF) in the instrument-random setting, where \( Z \) is drawn randomly from \([K] \) according to \( P(Z = k) = \frac{n_k}{N} \). To this end, let \( P \) denote the marginal distribution of \( O = (Z, X, Y) \), where \( Z \) is random and \( (X, Y) \mid Z = k \) under \( P \) is equal in distribution to \( (X_{ki}, Y_{ki}) \) under the fixed-instrument distribution \( P^{(K)} \). We will then identify this EIF with the EIF for our parameter in the fixed-instrument setting. This identification is valid because \( P \) can be recovered from \( P^{(K)} \), and all expectations and nuisance functions depend on \( P^{(K)} \) only through \( P \).

Let \( \mathcal{M} \) denote the statistical model for \( P \) that is compatible with $\mathcal{P}^{(K)}$. Define the operators
\[
\mathcal{T}_P: h \mapsto E_P[h(X) \mid Z = \cdot] \quad \text{and} \quad \mathcal{T}_P^*: q \mapsto E_P[q(Z) \mid X = \cdot],
\]
which, by construction of \( P \), coincide with the operators \( \mathcal{T}_K = \mathcal{T}_{P^{(K)}} \) and \( \mathcal{T}_K^* = \mathcal{T}_{P^{(K)}}^* \) defined in the main text. As in the main text, define the pseudoinverse \( \mathcal{T}_P^+: \mathcal{R}(\mathcal{T}_P) \to L^2_P(X) \) of \( \mathcal{T}_P \), given by \( \mathcal{T}_P^* (\mathcal{T}_P \mathcal{T}_P^*)^{-1} \). We define our target parameter \( \Psi: \mathcal{M} \rightarrow \mathbb{R} \) as $\Psi(P) = \mathcal{T}_P^+ \mu_P,$
where \( \mu_P: x \mapsto E_P[Y \mid X = x] \) denotes the outcome regression. In this section, we identify all nuisance functions depending on \( P^{(K)} \) as defined in the main text with their instrument-random analogues depending on \( P \); for example, by replacing \( \mathcal{T}_{P^{(K)}} \) with \( \mathcal{T}_P \).

In the following theorem, we recall that \( \Pi_K : L_K^2(X) \to L_K^2(X) \) denotes the orthogonal projection onto \( \mathcal{H}_K := \mathcal{R}(\mathcal{T}_K^*) \), and we define the projection onto its orthogonal complement \( \mathcal{H}_K^\perp \subset L_K^2(X) \) by \( \Pi_K^\perp := I_{L_K^2(X)} - \Pi_K \), where \( I_{L_K^2(X)} \) is the identity operator on \( L_K^2(X) \). We also recall that \( r_K := \mathcal{T}_K \rho_K \) is the minimum-norm solution to the inverse problem \( \mathcal{T}_K^* r_K = h_K \), where \( \rho_K := (\mathcal{T}_K^* \mathcal{T}_K)^+ h_K \). Finally, we recall the extended operator $\widetilde{\mathcal{T}}_P: L^2_P(X,Y) \rightarrow L^2_P(Z)$ by $(\widetilde{\mathcal{T}}_Pf)(k) := E_{P}[f(Y, X) \mid Z = k]$.

\begin{theorem}[EIF for random-instrument design]
Assume \ref{cond::continousfun}. Then, the parameter $\Psi: \mathcal{M} \rightarrow \mathbb{R}$ is pathwise differentiable at $P$ with efficient influence function:
    \begin{align*}
    \varphi_P^*(O) 
     = q_P(Z) \cdot \left( Y - h_P(X) \right) + (\operatorname{I} - \Pi_P) \alpha_P(X) \cdot \left( r_P(Z) - h_P(X) \right),
\end{align*}
\label{theorem::EIFrandom}
\end{theorem}
\begin{proof}[Proof of theorem \ref{theorem::EIFrandom}]
For every $\phi \in L_0^2(O)$, we let $\{ P_{\epsilon, \phi} : \epsilon \in [-\epsilon_0, \epsilon_0] \}$ be a one-dimensional submodel within $\mathcal{M}$ passing through $P$ at $\epsilon = 0$ with score $\phi$ at $\epsilon = 0$. By Condition~\ref{cond::rightinverse}, we have that \( \mathcal{T}_{P_{\varepsilon, \phi}} \) is surjective, and thus full rank, at \( \varepsilon = 0 \). Moreover, since the set of full-rank operators is open in the operator norm topology, and the map \( P \mapsto \mathcal{T}_{P, \phi} \) is continuous when the domain is equipped with the Hellinger norm and the range with the operator norm, it follows that \( \mathcal{T}_{P_{\varepsilon, \phi}} \) is full rank for all \( \varepsilon \) in a neighborhood of \( \varepsilon = 0 \).

Since the range of \( \mathcal{T}_{P_{\epsilon, \phi}} \) is finite-dimensional and therefore closed, the operator \( \mathcal{T}_{P_{\epsilon, \phi}} \mathcal{T}_{P_{\epsilon, \phi}}^* \) is invertible on $L^2_P(Z)$ with bounded inverse \( (\mathcal{T}_{P_{\epsilon, \phi}} \mathcal{T}_{P_{\epsilon, \phi}}^*)^{-1} \). Hence, we have that
\[
\mathcal{T}_{P_{\epsilon, \phi}}^+ = \mathcal{T}_{P_{\epsilon, \phi}}^* (\mathcal{T}_{P_{\epsilon, \phi}} \mathcal{T}_{P_{\epsilon, \phi}}^*)^{-1}
\]
for each \( \epsilon \). The conditional expectation operators $P \mapsto \mathcal{T}_P$ and $P \mapsto \mathcal{T}_P^*$ are pathwise differentiable in operator norm \citep{luedtke2024one, luedtke2024simplifying}. Hence, for any $r \in \mathcal{R}(\mathcal{T}_P)$, the chain rule applies and the pathwise differential
\begin{align*}
    d \mathcal{T}_P^+(r): \phi \in L^2_0(O) \mapsto \left. \frac{d}{d\epsilon} \left( \mathcal{T}_{P_{\epsilon, \phi}}^+ r \right) \right|_{\epsilon = 0}
\end{align*}
exists, where the derivative is taken in the sense of limits in the Hilbert space \( L^2_P(X) \).

Furthermore, denoting \( A_P := \mathcal{T}_P \mathcal{T}_P^* \), by the chain rule (see Lemma S3 in \cite{luedtke2024simplifying}), and for every \( r \in L^2_P(Z) \), \( \phi \in L^2_0(O) \),
\begin{align*}
    d \mathcal{T}_P^+(r)(\phi) 
    = d \mathcal{T}_P^* ( A_P^{-1} r)(\phi) 
    + \mathcal{T}_P^*\, d(A_P^{-1})(r)(\phi).
\end{align*}

We now focus on the pathwise differential \( d(A_P^{-1}) \). Using the identity \( d(A_P^{-1}) \cdot A_P + A_P^{-1} \cdot dA_P = 0 \) in operator notation, the chain rule yields
\begin{align*}
    d(A_P^{-1})(r)(\phi) = - A_P^{-1} \cdot dA_P (A_P^{-1} r)(\phi).
\end{align*}

Putting these identities together and applying the chain rule once more, we obtain
\begin{align*}
    d \mathcal{T}_P^+(r)(\phi) 
    = d \mathcal{T}_P^* ( A_P^{-1} r)(\phi)
    - \mathcal{T}_P^+ \cdot d \mathcal{T}_P (\mathcal{T}_P^+ r)(\phi)
    - \Pi_P \cdot d \mathcal{T}_P^* (A_P^{-1} r)(\phi),
\end{align*}
where \( \Pi_P := \mathcal{T}_P^* A_P^{-1} \) denotes the projection operator onto the range of \( \mathcal{T}_P^* \).

With this identity in hand, we can now derive the EIF of \( \Psi \). Via differentiation under the integral sign and the product rule, we have
\begin{align*}
    \frac{d}{d \epsilon} \Psi(P_{\epsilon, \phi}) \big|_{\epsilon = 0} 
    &= \frac{d}{d \epsilon} \langle \alpha_{P_{\epsilon, \phi}}, h_{P_{\epsilon, \phi}} \rangle_{P_{\epsilon, \phi}} \big|_{\epsilon = 0} \\
    &= \frac{d}{d \epsilon} \langle \alpha_P, h_{P_{\epsilon, \phi}} \rangle_P \big|_{\epsilon = 0} 
    + \frac{d}{d \epsilon} \langle \alpha_{P_{\epsilon, \phi}}, h_P \rangle_{P_{\epsilon, \phi}} \big|_{\epsilon = 0}.
\end{align*}

We begin with the second term of the right-hand side. By the Riesz representation theorem for the linear functional \( \psi \), we have, for any fixed \( h \in L^2_P(X) \), and for all \( \epsilon \) in a neighborhood of 0,
\begin{align*}
    \langle \alpha_{P_{\epsilon, \phi}}, h \rangle_{P_{\epsilon, \phi}} = \psi(h),
\end{align*}
where the right-hand side does not depend on \( \epsilon \). Hence, we obtain
\begin{align*}
    \frac{d}{d \epsilon} \langle \alpha_{P_{\epsilon, \phi}}, h_P \rangle_{P_{\epsilon, \phi}} \bigg|_{\epsilon = 0} = 0.
\end{align*}

We now turn to the first term $ \frac{d}{d \epsilon} \langle \alpha_P, h_{P_{\epsilon, \phi}} \rangle_P \big|_{\epsilon = 0} $. Define the operator \( \tilde{\mathcal{T}}_P : f \in L^2_P(Y) \mapsto (z \mapsto E_P[f(Y) \mid Z = z]) \), and observe that \( \mu_P = \tilde{\mathcal{T}}_P(\mathrm{id}_Y) \). Recall that \( r_P \) is the minimum-norm solution to the equation \( \mathcal{T}_P^* r_P = h_P \). Hence,
\[
r_P = (\mathcal{T}_P^*)^+ h_P 
= (\mathcal{T}_P \mathcal{T}_P^*)^{-1} \mathcal{T}_P h_P 
= (\mathcal{T}_P \mathcal{T}_P^*)^{-1} \mu_P 
= (\mathcal{T}_P \mathcal{T}_P^*)^{-1} \tilde{\mathcal{T}}_P(\mathrm{id}_Y).
\]
By differentiation under the integral sign and the chain rule, we have
\begin{align*}
    &\frac{d}{d \epsilon} \langle \alpha_P, h_{P_{\epsilon, \phi}} \rangle_P \Big|_{\epsilon = 0} \\
    =& \langle \alpha_P, d(\mathcal{T}_P^+ \circ \tilde{\mathcal{T}}_P)(\mathrm{id}_Y)(\phi) \rangle_P \\
    =& \langle \alpha_P, d \mathcal{T}_P^+(\mu_P)(\phi) \rangle_P + \langle \alpha_P, \mathcal{T}_P^+ \, d \tilde{\mathcal{T}}_P(\mathrm{id}_Y)(\phi) \rangle_P \\
    =& \langle \alpha_P, d \mathcal{T}_P^*(r_P)(\phi) - \mathcal{T}_P^+ \, d \mathcal{T}_P(h_P)(\phi) - \Pi_P \, d \mathcal{T}_P^*(r_P)(\phi) + \mathcal{T}_P^+ \, d \tilde{\mathcal{T}}_P(\mathrm{id}_Y)(\phi) \rangle_P \\
    =& \langle (\mathcal{T}_P^*)^+ \alpha_P, - d \mathcal{T}_P(h_P)(\phi) + d \tilde{\mathcal{T}}_P(\mathrm{id}_Y)(\phi) \rangle_P \\
    &\quad + \langle (I_{L^2_P(X)} - \Pi_P) \alpha_P, d \mathcal{T}_P^*(r_P)(\phi) \rangle_P,
\end{align*}
where the last line follows from the adjoint properties of the pseudoinverse and of orthogonal projections.

Computing the pathwise derivatives of the conditional expectation operators explicitly, it is straightforward to verify that for any \( q \in L^2_P(Z) \) and \( f \in L^2_P(X) \),
\begin{align*}
    \langle q, d \tilde{\mathcal{T}}_P(\mathrm{id}_Y)(\phi) \rangle_P &= \langle \phi, q \cdot (\mathrm{id}_Y - \mu_P) \rangle_P, \\
    \langle q, d \mathcal{T}_P(h_P)(\phi) \rangle_P &= \langle \phi, q \cdot (h_P - \mathcal{T}_P h_P) \rangle_P, \\
    \langle f, d \mathcal{T}_P^*(r_P)(\phi) \rangle_P &= \langle \phi, f \cdot (r_P - \mathcal{T}_P^* r_P) \rangle_P.
\end{align*}
See, for example, the proof of Theorem 3 in \cite{van2025automatic}.

Putting the pieces together, we conclude that the efficient influence function \( \varphi_P \in L_0^2(O) \) for the parameter \( \Psi(P) = \mathcal{T}_P^+ \mu_P \) is given by
\begin{align*}
 \varphi_P^*(O) 
    &= (\mathcal{T}_P^*)^+ \alpha_P(Z) \cdot \left\{ (Y - \mu_P(X)) - \left(h_P(X) - \mathcal{T}_P h_P(Z)\right) \right\} \\
    &\quad + (\operatorname{I} - \Pi_P)\alpha_P(X) \cdot \left\{ r_P(Z) - \mathcal{T}_P^* r_P(X) \right\} \\
    &= q_P(Z) \cdot \left( Y - h_P(X) \right) + (\operatorname{I} - \Pi_P) \alpha_P(X) \cdot \left( r_P(Z) - h_P(X) \right),
\end{align*}
as desired.

\end{proof}
 
\begin{proof}[Proof of Theorem \ref{theorem::EIF}.]

To this end, note that the conditions of Theorem~\ref{theorem::EIF} imply those of Theorem~\ref{theorem::EIFrandom}. Hence, we can apply Theorem~\ref{theorem::EIFrandom} to conclude that the efficient influence function of \( \Psi \) is given by \( \varphi_{P^{(K)}}^* \) in the instrument-random setting. We now show that the instrument-random EIF determines the EIF in the instrument-fixed design.

Let \( \Theta: \mathcal{P}^{(K)} \rightarrow \mathcal{M} \) be the map that takes \( P^{(K)} \in \mathcal{P}^{(K)} \) and maps it to its compatible distribution \( P \in \mathcal{M} \). Specifically, treating \( X \) and \( Y \) as discrete for notational convenience, we have
\[
P(Z = k, X = x, Y = y) = \frac{n_k}{N} \cdot P^{(K)}(X_{k1} = x, Y_{k1} = y).
\]
Hence, we define \( \Theta \) as the map
\[
\Theta(P^{(K)})(k, x, y) := \frac{n_k}{N} \cdot P^{(K)}(X_{k1} = x, Y_{k1} = y).
\]
Let \( D \in T_{\mathcal{P}^{(K)}}(P^{(K)}) \) be a score in the tangent space at \( P^{(K)} \). Then, under our independence assumptions and the alignment of the conditional distribution of \( (Y_{ki}, X_{ki}) \) across observations \( i \) within each stratum \( k \), we have
\[
D(o^{(K)}) :=   \sum_{k=1}^K \sum_{i=1}^{n_k} \varphi(o_{ki}),
\]
where \( \varphi \) is a function defined on individual units and evaluated at a representative observation from each stratum, and is mean-zero in the sense that
\[
  E_{P^{(K)}}[\varphi(O_{ki})] = 0.
\]

Let \( \{P_{\varepsilon, D}^{(K)} : \varepsilon\} \subset \mathcal{P}^{(K)} \) be a smooth submodel through \( P^{(K)} \) with score \( D \) at \( \varepsilon = 0 \). Then, it is straightforward to verify that \( \{\Theta(P_{\varepsilon, D}^{(K)}) : \varepsilon\} \subset \mathcal{M} \) defines a smooth submodel through \( \Theta(P^{(K)}) \) with score \( \varphi \), since \( \varphi(k, O_{k1}, Z_{k1}) \) corresponds exactly to the score component arising from perturbing the distribution $P^{(K)}(X_{k1}  , Y_{k1}  )$ of the single unit \( (O_{k1}, Z_{k1}) \). It follows that the pathwise derivative of the mapping $P^{(K)} \mapsto \Theta(P^{(K)})$ along this path satisfies $d\Theta(P^{(K)})(D) = \varphi$.

Now, observe that 
\[
\Psi^{(K)}\bigl(P^{(K)}\bigr) = \Psi\bigl(\Theta(P^{(K)})\bigr).
\]
Hence, by the chain rule for pathwise derivatives, we have
\[
d\Psi^{(K)}\bigl(P^{(K)}\bigr) = d\Psi\bigl(\Theta(P^{(K)})\bigr) \circ d\Theta\bigl(P^{(K)}\bigr).
\]
We know that \( d\Theta(P^{(K)})(D) = \varphi \). Letting \( P = \Theta(P^{(K)}) \), it follows that
\[
d\Psi^{(K)}\bigl(P^{(K)}\bigr)(D) = d\Psi(P)(\varphi) = \langle \varphi, \varphi^*_P \rangle_{L^2_P(O)},
\]
where \( \varphi^*_P \) is the efficient influence function of \( \Psi \). By independence of observations and the mean-zero property of \( \varphi \), we have
\begin{align*}
    \langle \varphi, \varphi^*_P \rangle_{L^2_P(O)}  &=     \sum_{k=1}^K \frac{n_k}{N} E_{P^{(K)}}\left[ \varphi(O_{ki})  \varphi^*_P(O_{ki})  \right] \\
    &=   \frac{1}{N} \sum_{k=1}^K \sum_{i=1}^{n_k} E_{P^{(K)}}\left[ \varphi(O_{ki})  \varphi^*_P(O_{ki})  \right] \\
    &= E_{P^{(K)}}\left[ \left(   \sum_{k=1}^K \sum_{i=1}^{n_k} \varphi(O_{ki}) \right)
    \left( \frac{1}{N} \sum_{k=1}^K \sum_{i=1}^{n_k} \varphi^*_P(O_{ki}) \right) \right] \\
    &= E_{P^{(K)}}\left[ D(O^{(K)}) \, D^*_{P^{(K)}}(O^{(K)}) \right],
\end{align*}
where
\[
D^*_{P^{(K)}}(o^{(K)}) := \frac{1}{N} \sum_{k=1}^K \sum_{i=1}^{n_k} \varphi^*_P(o_{ki}).
\]
Since \( d\Psi^{(K)}\bigl(P^{(K)}\bigr)(D) = d\Psi(P)(\varphi) = \langle \varphi, \varphi^*_P \rangle_{L^2_P(O)} \) and \( D^*_{P^{(K)}} \in T_{\mathcal{P}^{(K)}}(P^{(K)}) \), it follows that \( D^*_{P^{(K)}} \) is the efficient influence function of \( \Psi^{(K)} \).

\end{proof}

\begin{proof}[Proof of Theorem \ref{theorem::vonmises}]

By the Riesz representation property, we have $ \psi(\bar{h}_K) - \psi(h_K) =  \langle \alpha_K, \bar h_K -  h_K \rangle_{L_K^2(X)}$. Hence, 
\begin{align}
    \psi(\bar{h}_K) - \psi(h_K) + P^{(K)} \bar D_K 
    &=  \langle \alpha_K, \bar h_K -  h_K \rangle_{L_K^2(X)} \notag \\
    &\quad + \langle \bar q_K, h_K - \bar h_K \rangle_{L_K^2(X)} \notag \\
    &\quad + \langle (I_{L_K^2(X)} - \bar \Pi_K) \bar \alpha_K, \bar r_K - \bar h_K \rangle_{L_K^2(X)}.
\end{align}
By the law of total expectation, the second term satisfies 
\begin{align*}
    \langle \bar q_K, h_K - \bar h_K \rangle_{L_K^2(X)} &= \langle \bar q_K, \mathcal{T}_K(h_K - \bar h_K) \rangle_{L_K^2(Z)}\\
    &= \langle \bar q_K - q_K, \mathcal{T}_K(h_K - \bar h_K) \rangle_{L_K^2(Z)} + \langle \bar q_K, \mathcal{T}_K(h_K - \bar h_K) \rangle_{L_K^2(Z)}\\
     &= \langle \bar q_K - q_K, \mathcal{T}_K(h_K - \bar h_K) \rangle_{L_K^2(Z)} + \langle \bar q_K, h_K - \bar h_K\rangle_{L_K^2(X)}.
\end{align*} 
Hence, we have
\begin{align*}
    \psi(\bar{h}_K) - \psi(h_K) + P^{(K)} \bar D_K 
        &=  \langle \alpha_K,  \bar h_K -  h_K \rangle_{L_K^2(X)} \notag \\
    & \quad + \langle \bar q_K - q_K, \mathcal{T}_K( h_K - \bar h_K ) \rangle_{L_K^2(X)}\notag \\
      & \quad - \langle  q_K, \bar h_K -  h_K \rangle_{L_K^2(X)}\notag \\
    &\quad + \langle (I_{L_K^2(X)} - \bar \Pi_K) \bar \alpha_K, \bar r_K - \bar h_K \rangle_{L_K^2(X)}.
\end{align*}
By the law of total expectation, property of the adjoint, and the fact that $ \mathcal{T}_K^* q_K = \Pi_K \alpha_K$, we have 
\begin{align*}
    \langle  q_K, \bar h_K -  h_K \rangle_{L_K^2(X)} &=  \langle  q_K, \mathcal{T}_K(\bar h_K -  h_K) \rangle_{L_K^2(Z)} \\
    &=  \langle  \mathcal{T}_K^* q_K, \bar h_K -  h_K \rangle_{L_K^2(X)} \\
     &=  \langle \Pi_K \alpha_K, \bar h_K -  h_K \rangle_{L_K^2(X)} 
\end{align*} 
Hence, we have
\begin{align*}
    \psi(\bar{h}_K) - \psi(h_K) + P^{(K)} \bar D_K 
        &=  \langle (I_{L_K^2(X)} - \Pi_K)\alpha_K  ,  \bar h_K -  h_K \rangle_{L_K^2(X)} \notag \\
    & \quad + \langle \bar q_K - q_K, \mathcal{T}_K( h_K - \bar h_K ) \rangle_{L_K^2(X)}\notag \\
    &\quad + \langle (I_{L_K^2(X)} - \bar \Pi_K) \bar \alpha_K, \bar r_K - \bar h_K \rangle_{L_K^2(X)}
\end{align*}
Applying the law of total expectation, adding and subtracting, and using that $h_K  = \mathcal{T}_K^*(r_K)$, we find
\begin{align*}
    \langle (I_{L_K^2(X)} - \bar \Pi_K) \bar \alpha_K, \bar r_K - \bar h_K \rangle_{L_K^2(X)}  &= \langle (I_{L_K^2(X)} - \bar \Pi_K) \bar \alpha_K, \mathcal{T}_K^*(\bar r_K) - \bar h_K \rangle_{L_K^2(X)}\\
 &=   \langle (I_{L_K^2(X)} - \bar \Pi_K) \bar \alpha_K, h_K - \bar h_K \rangle_{L_K^2(X)}\\
     &\quad + \langle (I_{L_K^2(X)} - \bar \Pi_K) \bar \alpha_K, \mathcal{T}_K^*(\bar r_K) - h_K \rangle_{L_K^2(X)}\\
      &=   \langle (I_{L_K^2(X)} - \bar \Pi_K) \bar \alpha_K, h_K - \bar h_K \rangle_{L_K^2(X)}\\
     &\quad + \langle (I_{L_K^2(X)} - \bar \Pi_K) \bar \alpha_K, \mathcal{T}_K^*(\bar r_K - r_K) \rangle_{L_K^2(X)}
\end{align*}
Hence, we have
\begin{align*}
    \psi(\bar{h}_K) - \psi(h_K) + P^{(K)} \bar D_K 
        &=    \langle \bar q_K - q_K, \mathcal{T}_K( h_K - \bar h_K ) \rangle_{L_K^2(X)}\notag \\
    &\quad +   \langle (I_{L_K^2(X)} - \bar \Pi_K) \bar \alpha_K - (I_{L_K^2(X)} - \Pi_K)\alpha_K , h_K - \bar h_K \rangle_{L_K^2(X)}\\
     &\quad + \langle (I_{L_K^2(X)} - \bar \Pi_K) \bar \alpha_K, \mathcal{T}_K^*(\bar r_K - r_K) \rangle_{L_K^2(X)}.
\end{align*}
Finally, the orthogonality conditions of the projection $\Pi_K \alpha_K$ imply that
$$ \langle (I_{L_K^2(X)} - \Pi_K)  \alpha_K, \mathcal{T}_K^*(\bar r_K - r_K) \rangle_{L_K^2(X)} = 0.$$
We conclude that
\begin{align*}
    \psi(\bar{h}_K) - \psi(h_K) + P^{(K)} \bar D_K 
        &=    \langle \bar q_K - q_K, \mathcal{T}_K( h_K - \bar h_K ) \rangle_{L_K^2(X)}\notag \\
    &\quad +   \langle (I_{L_K^2(X)} - \bar \Pi_K) \bar \alpha_K - (I_{L_K^2(X)} - \Pi_K)\alpha_K , h_K - \bar h_K \rangle_{L_K^2(X)}\\
     &\quad + \langle (I_{L_K^2(X)} - \bar \Pi_K) \bar \alpha_K - (I_{L_K^2(X)} - \Pi_K)\alpha_K, \mathcal{T}_K^*(\bar r_K - r_K) \rangle_{L_K^2(X)}\\
     &=    \langle \bar q_K - q_K, \mathcal{T}_K( h_K - \bar h_K ) \rangle_{L_K^2(X)}\notag \\
    &\quad +   \langle (I_{L_K^2(X)} - \bar \Pi_K) \bar \alpha_K - (I_{L_K^2(X)} - \Pi_K)\alpha_K ,  \mathcal{T}_K^*(\bar r_K - r_K)  - ( \bar h_K - h_K) \rangle_{L_K^2(X)}.
\end{align*}

\end{proof}

 \section{Proofs for Section \ref{sec::theory}}

\subsection{Technical lemmas for Sections \ref{sec::theory1}-\ref{sec::theory3}}

\begin{lemma}[Bernstein inequality with growing envelope]
Let \( X_1, \dots, X_n \) be i.i.d.\ random variables and let \( \mathcal{E}_n \) be an independent dataset (or \(\sigma\)-field). Let \( f_n : \mathcal{X} \to \mathbb{R} \) be a function measurable with respect to \( \mathcal{E}_n \), satisfying \( \mathbb{E}[f_n(X_i) \mid \mathcal{E}_n] = P f_n \), \( \operatorname{Var}(f_n(X_i) \mid \mathcal{E}_n) \le \sigma_{n}^2 \), and \( |f_n(X_i)| \le B_n \) almost surely, for some possibly growing deterministic bound \( B_n \). Define the empirical mean \( P_n f_n := \frac{1}{n} \sum_{i=1}^n f_n(X_i) \). Then for any \( t > 0 \),
\[
\mathbb{P}\left( P_n f_n - P f_n \ge t \mid \mathcal{E}_n \right) \le \exp\left( - \frac{n t^2}{2 \sigma_{n}^2 + \frac{2}{3} B_n t} \right).
\]
In particular, $P_n f_n - P f_n = \mathcal{O}_p\left( \sqrt{\frac{\sigma_n^2}{n}} + \frac{B_n}{n} \right).$
\label{lemma::bernstein}
\end{lemma}

\begin{lemma}\label{lemma::sup_residual_bound}
Assume \ref{cond::split} and \ref{cond::bound}. Fix \( v \in \{0,1\} \) and define \( \widehat{q}_{K,v}(k) := \frac{2}{n_k} \sum_{j \in [n_k] : V_{kj} = v} \widehat{\beta}_K(X_{kj}) \). Suppose that, conditional on the dataset \( \mathcal{E}_K \), the residuals \( R_{k,i} := \widehat{q}_{K,v}(k) - \mathcal{T}_K(\widehat{\beta}_K)(k) \) satisfy $\mathbb{E}[R_{k,i} \mid \mathcal{E}_K] = 0$
and assume \( \|\widehat{\beta}_K\|_{L^2_K(X)} = O_p(\sigma_{\widehat{\beta}_K}^2) \), \( \|\widehat{\beta}_K\|_\infty = O_p(B_{\widehat{\beta}_K}) \). Let $n = \min_{k \in [K] n_k}$.  Then,
\[
\sup_{k \in [K]} \left| \widehat{q}_{K,v}(k) - \mathcal{T}_K(\widehat{\beta}_K)(k) \right| = O_p\left( \sqrt{\frac{\sigma_{\widehat{\beta}_K}^2 \log K}{n}} + \frac{B_{\widehat{\beta}_K} \log K}{n} \right).
\]
\end{lemma}

\begin{proof}
Fix \( k \in [K] \) and apply Lemma~\ref{lemma::bernstein} conditional on \( \mathcal{E}_K \) to \( f_{n,k}(O_i) := R_{k,i} \). For any \( t > 0 \),
\[
\mathbb{P}\left( \left| \frac{1}{n} \sum_{i=1}^n R_{k,i} \right| \ge t \mid \mathcal{E}_K \right) \le 2 \exp\left( - \frac{n t^2}{2 \sigma_{\widehat{\beta}_K}^2 + \frac{2}{3} B_{\widehat{\beta}_K} t} \right).
\]
A union bound over \( k \in [K] \) gives
\[
\mathbb{P}\left( \sup_{k \in [K]} \left| \widehat{q}_{K,v}(k) - \mathcal{T}_K(\widehat{\beta}_K)(k) \right| \ge t \mid \mathcal{E}_K \right) \le 2K \exp\left( - \frac{n t^2}{2 \sigma_{\widehat{\beta}_K}^2 + \frac{2}{3} B_{\widehat{\beta}_K} t} \right).
\]
Setting \( t \asymp \sqrt{ \sigma_{\widehat{\beta}_K}^2 \log K / n } + B_{\widehat{\beta}_K} \log K / n \) yields the result.
\end{proof}

\begin{lemma}[Asymptotic equivalence of debiasing term]
Assume \ref{cond::split} and \ref{cond::bound}. Let $n = \min_{k \in [K] n_k}$. Then,
    \begin{align*}
      &   \frac{2}{N} \sum_{v \in \{0,1\}} \sum_{k=1}^K \sum_{i \in [n_k]: V_{ki} = 1-v}   \left\{\frac{2}{n_k}  \sum_{j \in [n_k]:  V_{kj} = v} \widehat{\beta}_K(X_{kj})\right\}\Bigl(Y_{ki} - \widehat{h}_K\bigl(X_{ki}\bigr)\Bigr)\\
         &=  \frac{1}{N}   \sum_{k=1}^K \sum_{i=1}^{n_k}  \mathcal{T}_{K}(\widehat{\beta}_K)(k) \Bigl(Y_{ki} - \widehat{h}_K\bigl(X_{ki}\bigr)\Bigr) +  O_p(N^{-1/2}) O_p\left( \sqrt{\frac{\|\widehat{\beta}_K\|_{L^2_K(X)}^2 \log K}{n}} + \frac{\|\widehat{\beta}_K\|_{\infty} \log K}{n} \right).
\end{align*} 
\label{lemma::biasExpectation}
\end{lemma} 
\begin{proof}[Proof of Lemma \ref{lemma::biasExpectation}]
Fix $v \in \{0,1\}$ and denote $\widehat{q}_{K,v}(k) := \frac{2}{n_k}  \sum_{j \in [n_k]:  V_{kj} = v} \widehat{\beta}_K(X_{kj})$. Then, adding and subtracting, we have
    \begin{align*}
         \frac{2}{N}  \sum_{k=1}^K \sum_{i \in [n_k]: V_{ki} = 1-v} & \widehat{q}_{K,v}(k) \Bigl(Y_{ki} - \widehat{h}_K\bigl(X_{ki}\bigr)\Bigr)\\
         &=  \frac{2}{N}   \sum_{k=1}^K \sum_{i \in [n_k]: V_{ki} = 1-v}  \mathcal{T}_{K}(\widehat{\beta}_K) \Bigl(Y_{ki} - \widehat{h}_K\bigl(X_{ki}\bigr)\Bigr) \\
         & \quad +  \frac{2}{N}   \sum_{k=1}^K \sum_{i \in [n_k]: V_{ki} = 1-v}  \left\{\widehat{q}_{K,v}(k) - \mathcal{T}_{K}(\widehat{\beta}_K)(k)\right\} \Bigl(Y_{ki} - \widehat{h}_K\bigl(X_{ki}\bigr)\Bigr).
    \end{align*}

We consider the second term on the right-hand side. Let \( f_{\widehat{\beta}_K,\widehat{h}_K} := \left\{\widehat{q}_{K,v} - \mathcal{T}_{K}(\widehat{\beta}_K)\right\} \left(\mathrm{id}_Y - \widehat{h}_K\right) \). Then,
  \begin{align*}
    \frac{2}{N}   \sum_{k=1}^K \sum_{i \in [n_k]: V_{ki} = 1-v} f_{\widehat{\beta}_K,\widehat{h}_K}(O_{ki})=    \frac{2}{N}   \sum_{k=1}^K \sum_{i \in [n_k]: V_{ki} = 1-v}  \left\{\widehat{q}_{K,v}(k) - \mathcal{T}_{K}  (\widehat{\beta}_K)(k)\right\} \Bigl(Y_{ki} - \widehat{h}_K\bigl(X_{ki}\bigr)\Bigr). 
  \end{align*}  
By Condition~\ref{cond::bound}, the function $f_{\widehat{\beta}_K,\widehat{h}_K}$ is uniformly bounded by \( O(\|\widehat{\beta}_K\|_{\infty}) \) almost surely, since \( \widehat{\beta}_K \) is uniformly bounded by \( \|\widehat{\beta}_K\|_{\infty} \), and both \( \widehat{q}_{K,v} \) and \( \mathcal{T}_{K}(\widehat{\beta}_K) \), being averages involving \( \widehat{\beta}_K \), are also uniformly bounded by \( \|\widehat{\beta}_K\|_{\infty} \). In addition, conditional on the external training data for \( \widehat{\beta}_K \) and \( \widehat{h}_K \), denoted by \( \mathcal{E}_K \), the random variable \( \widehat{q}_{K,v}(k) - \mathcal{T}_{K}(\widehat{\beta}_K)(k) \) has mean zero. Consequently, $f_{\widehat{\beta}_K,\widehat{h}_K}(O_{ki})$ is a mean zero random variable conditional on $\mathcal{E}_K$, since by independence of the data folds, 
$$E_K[f_{\widehat{\beta}_K,\widehat{h}_K}(O_{ki}) \mid \mathcal{E}_K] = E_K\left[ \frac{2}{n_k}  \sum_{j \in [n_k]:  V_{kj} = v} \widehat{\beta}_K(X_{kj}) - \mathcal{T}_{K}(\widehat{\beta}_K)(k) \mid \mathcal{E}_K\right] \cdot E_K[Y_k - \widehat{h}_K(X_k)  \mid \mathcal{E}_K]  = 0.$$
Moreover, applying Lemma \ref{lemma::sup_residual_bound} , we have
\begin{align*}
    \left\|f_{\widehat{\beta}_K,\widehat{h}_K}  \right\|_{L^2} &= \sup_{k \in [K]} \left|\widehat{q}_{K,v}(k) - \mathcal{T}_{K}(\widehat{\beta}_K)(k) \right| O_p(1)  = O_p\left( \sqrt{\frac{\|\widehat{\beta}_K\|_{L^2_K(X)}^2 \log K}{n}} + \frac{\|\widehat{\beta}_K\|_{\infty} \log K}{n} \right).
\end{align*}
It follows that
\begin{align*}
        & \frac{2}{N}  \sum_{k=1}^K \sum_{i \in [n_k]: V_{ki} = 1-v}   \widehat{q}_{K,v}(k) \Bigl(Y_{ki} - \widehat{h}_K\bigl(X_{ki}\bigr)\Bigr)\\
         &=  \frac{2}{N}   \sum_{k=1}^K \sum_{i \in [n_k]: V_{ki} = 1-v}  \mathcal{T}_{K}(\widehat{\beta}_K) \Bigl(Y_{ki} - \widehat{h}_K\bigl(X_{ki}\bigr)\Bigr) +  O_p(N^{-1/2}) O_p\left( \sqrt{\frac{\|\widehat{\beta}_K\|_{L^2_K(X)}^2 \log K}{n}} + \frac{\|\widehat{\beta}_K\|_{\infty} \log K}{n} \right).
\end{align*} 
Summing over $v$, we conclude that
\begin{align*}
         & \frac{2}{N} \sum_{v \in \{0,1\}} \sum_{k=1}^K \sum_{i \in [n_k]: V_{ki} = 1-v}  \widehat{q}_{K,v}(k) \Bigl(Y_{ki} - \widehat{h}_K\bigl(X_{ki}\bigr)\Bigr)\\
         &=  \frac{1}{N}   \sum_{k=1}^K \sum_{i=1}^{n_k}  \mathcal{T}_{K}(\widehat{\beta}_K)(k) \Bigl(Y_{ki} - \widehat{h}_K\bigl(X_{ki}\bigr)\Bigr) +  O_p(N^{-1/2}) O_p\left( \sqrt{\frac{\|\widehat{\beta}_K\|_{L^2_K(X)}^2 \log K}{n}} + \frac{\|\widehat{\beta}_K\|_{\infty} \log K}{n} \right).
\end{align*}

\end{proof}

In the following, denote $\widehat{D}_K :=   \mathcal{T}_{K}(\widehat{\beta}_K) (\text{id}_Y - \widehat{h}_K) $ and $D_K := \mathcal{T}_{K}(\beta_{K}) (\text{id}_Y - h_K)$.

\begin{lemma}[Asymptotically linear expansion of debiasing term]
Assume \ref{cond::split}, \ref{cond::bound}, and \ref{cond::nuisrate}. Suppose $\| \mathcal{T}_{K}(\widehat{\beta}_K) -  \mathcal{T}_{K}(\beta_{K})\|_{L^2_K(Z)} = o_p(1)$, $\|\widehat{h}_K - h_K\|_{L^2_K(X)} = o_p(1)$. Then,
  \begin{align*}
      \frac{1}{N}   \sum_{k=1}^K \sum_{i=1}^{n_k} \widehat{D}_K(O_{ki})  &=    \frac{1}{N}   \sum_{k=1}^K \sum_{i=1}^{n_k} D_K(O_{ki}) \\
      & \quad +  \sum_{k=1}^K  \frac{n_k}{N} \int  \widehat{D}_K(o_k)  P^{(K)}(do_k) + o_p(N^{-1/2}) .
  \end{align*}  
\label{lemma::asymlinear}
\end{lemma}

\begin{proof}[Proof of Lemma \ref{lemma::asymlinear}]

By \ref{cond::split}, the function $\widehat{D}_K$ is fixed conditional on an external independent dataset $\mathcal{E}_K$. Moreover, by \ref{cond::nuisrate}(i) and the triangle inequality, $\|\widehat{D}_K - D_K\|_{L^2_K} \lesssim \| \mathcal{T}_{K}( \widehat{\beta}_K - \beta_{K}) \|_{L^2_K(Z)} + \|  \widehat{h}_K - h_K\|_{L^2_K(X)} = o_p(1)$. Hence, applying Chebyschev's inequality conditional on $\mathcal{E}_K$, we conclude that
  \begin{align*}
      \frac{1}{N}   \sum_{k=1}^K \sum_{i=1}^{n_k} \{ \widehat{D}_K(O_{ki}) - P^{(K)} \widehat{D}_K \}  &=    \frac{1}{N}   \sum_{k=1}^K \sum_{i=1}^{n_k} \{D_K(O_{ki}) - P^{(K)} D_K\}  + O_p\left( \frac{\|\widehat{D}_K - D_K\|_{L^2_K} }{\sqrt{N}}\right)\\
      &=   \frac{1}{N}   \sum_{k=1}^K \sum_{i=1}^{n_k} D_K(O_{ki}) + o_p(N^{-1/2}).
  \end{align*}  
since $P^{(K)} D_K = 0$. Rearranging terms, we have
      \begin{align*}
      \frac{1}{N}   \sum_{k=1}^K \sum_{i=1}^{n_k} \widehat{D}_K(O_{ki})  &=    \frac{1}{N}   \sum_{k=1}^K \sum_{i=1}^{n_k} D_K(O_{ki}) \\
      & \quad +  \sum_{k=1}^K  \frac{n_k}{N} \int  \widehat{D}_K(o_k)  P^{(K)}(do_k) + o_p(N^{-1/2}).
  \end{align*}  
\end{proof}

\subsection{Proofs of main results}

\begin{proof}[Proof of Theorem  \ref{theorem::asymnormal}]
By \ref{cond::bound}, \ref{cond::split},  and Lemma \ref{lemma::biasExpectation},
\begin{align*}
         \frac{2}{N} \sum_{v \in \{0,1\}} \sum_{k=1}^K \sum_{i \in [n_k]: V_{ki} = 1-v} & \left\{\frac{2}{n_k}  \sum_{j \in [n_k]:  V_{kj} = v} \widehat{\beta}_K(X_{kj})\right\}\Bigl(Y_{ki} - \widehat{h}_K\bigl(X_{ki}\bigr)\Bigr)\\
         &=  \frac{1}{N}   \sum_{k=1}^K \sum_{i=1}^{n_k}  \mathcal{T}_{K}(\widehat{\beta}_K)(k) \Bigl(Y_{ki} - \widehat{h}_K\bigl(X_{ki}\bigr)\Bigr) \\
         & \quad +  O_p(N^{-1/2}) O_p\left( \sqrt{\frac{\|\widehat{\beta}_K\|_{L^2_K(X)}^2 \log K}{n}} + \frac{\|\widehat{\beta}_K\|_{\infty} \log K}{n} \right)\\
          &=  \frac{1}{N}   \sum_{k=1}^K \sum_{i=1}^{n_k}  \mathcal{T}_{K}(\widehat{\beta}_K)(k) \Bigl(Y_{ki} - \widehat{h}_K\bigl(X_{ki}\bigr)\Bigr) + o_p\left(\sqrt{\frac{\sigma_K^2}{N}}\right), \\
\end{align*} 
where the final equality follows from the boundedness conditions in \ref{cond::bound}. Hence, the one-step estimator satisfies the expansion:
$$\widehat{\psi}_K = \psi(\widehat{h}_K) + \frac{1}{N}   \sum_{k=1}^K \sum_{i=1}^{n_k}  \mathcal{T}_{K}(\widehat{\beta}_K)(k) \Bigl(Y_{ki} - \widehat{h}_K\bigl(X_{ki}\bigr)\Bigr) + o_p(\sigma_K N^{-1/2}). $$

Recall that  $\widehat{D}_K :=   \mathcal{T}_{K}(\widehat{\beta}_K) (\text{id}_Y - \widehat{h}_K) $ and $D_K := \mathcal{T}_{K}(\beta_{K}) (\text{id}_Y - h_K)$. By \ref{cond::bound}, \ref{cond::split}, \ref{cond::nuisrate}, and Lemma \ref{lemma::asymlinear}
  \begin{align*}
      \frac{1}{N}   \sum_{k=1}^K \sum_{i=1}^{n_k} \widehat{D}_K(O_{ki})  &=    \frac{1}{N}   \sum_{k=1}^K \sum_{i=1}^{n_k} D_K(O_{ki}) \\
      & \quad +  \sum_{k=1}^K  \frac{n_k}{N} \int  \widehat{D}_K(o_k)  P^{(K)}(do_k) + o_p(N^{-1/2}).
  \end{align*}  
Thus, $\widehat{\psi}_K = \psi(\widehat{h}_K) +  \frac{1}{N}   \sum_{k=1}^K \sum_{i=1}^{n_k} D_K(O_{ki})  +  \sum_{k=1}^K  \frac{n_k}{N} \int  \widehat{D}_K(o_k)  P^{(K)}(do_k) + o_p(\sigma_K N^{-1/2}).$

Furthermore, using the orthogonality of the projection \( \left( \mathrm{I} - \Pi_{K} \right) \alpha_K \), the von Mises expansion in Theorem~\ref{theorem::EIF}, with \( \bar{q}_K := \mathcal{T}_{K}(\widehat{\beta}_K) \), \( \bar{\alpha}_K := \bar{\Pi}_{K} \bar{\alpha}_K \), and \( \bar{h}_K := \widehat{h}_K \), yields:
\begin{align*}
    \psi(\widehat{h}_K) &- \psi(h_K) + \sum_{k=1}^K  \frac{n_k}{N} \int  \widehat{D}_K(o_k)  P^{(K)}(do_k)   \\
    &= \big \langle \mathcal{T}_{K}(\beta_{K}  - \widehat{\beta}_K) , \mathcal{T}_{K} \big( \widehat{h}_K - h_K \big) \big \rangle_{L^2_K(Z)}   
    + \big \langle \alpha_K - \Pi_{K} \alpha_K , \widehat{h}_K - h_K \big \rangle_{L^2_K(X)} \\
    &= o_p(\sigma_K N^{-1/2}),
\end{align*}
where the final equality follows from \ref{cond::nuisrate} and \ref{cond::nuisrates2}. Adding $\frac{1}{N}   \sum_{k=1}^K \sum_{i=1}^{n_k} D_K(O_{ki})$ to both sides of the above display, we conclude that
\begin{align*}
    \widehat{\psi}_K - \psi(h_K)  = \frac{1}{N}   \sum_{k=1}^K \sum_{i=1}^{n_k} D_K(O_{ki})  + o_p(\sigma_K N^{-1/2}).
\end{align*}
Finally, by Theorem \ref{theorem::approx} and \ref{cond::approxrate}, this further implies that $ \widehat{\psi}_K - \psi(h^\star)  = \frac{1}{N}   \sum_{k=1}^K \sum_{i=1}^{n_k} D_K(O_{ki})  + o_p(\sigma_K  N^{-1/2})$.

\medskip
\noindent\textbf{Application of the Lyapunov CLT.}
To apply the Lyapunov CLT, fix any $\delta > 0$. By Condition~\ref{cond::bound},
\[
E[|D_K(O_{ki})|^{2+\delta}] = O(|q(k)|^{2+\delta}),
\]  
and Condition~\ref{cond::lyaponov} ensures
\[
\frac{1}{\sigma_K^{2 + \delta} N^{1 + \delta/2}} \sum_{k=1}^K \sum_{i=1}^{n_k} E[|D_K(O_{ki})|^{2+\delta}] 
= \frac{1}{  N^{\delta/2}} \frac{\sum_{k=1}^K \frac{n_k}{N} |q(k)|^{2+\delta}}{\sigma_K^{2 + \delta}} \to 0.
\]
Hence, the Lyapunov condition holds, and
\[
\frac{1}{\sigma_K \sqrt{N}} \sum_{k=1}^K \sum_{i=1}^{n_k} D_K(O_{ki}) \xrightarrow{d} \mathcal{N}(0,1),
\]
so that
\[
\frac{\widehat\psi_K - \psi(h^\star)}{\sigma_K N^{-1/2}} \xrightarrow{d} \mathcal{N}(0,1).
\]

Since \(D_K\) is bounded and has finite variance by \ref{cond::bound}, an application of Lindeberg's CLT gives 
\(\sqrt{N/\sigma_K^{2}} \big( \widehat{\psi}_K - \psi(h^\star) \big) \overset{d}{\to} \text{N}(0,1) \text{ as } K \to \infty\),
where  
\(\sigma_K^2 := \sum_{k=1}^K \frac{n_k}{N} q_{K}^2(k) \sigma_{\epsilon}^{2}(k)\),  
with \(\sigma_{\epsilon}^{2}(k) = \mathrm{Var}_{K}[Y_k - h_K(X_k)]\).

\end{proof}

\begin{proof}[Proof of Theorem  \ref{theorem::asymbounded}]

Denote $\widehat{q}_{K,v}(k) := \frac{2}{n_k}  \sum_{j \in [n_k]:  V_{kj} = v} \widehat{\beta}_K(X_{kj})$. The proof of Lemma \ref{lemma::biasExpectation} establishes that
    \begin{align*}
         \frac{2}{N}  \sum_{k=1}^K \sum_{i \in [n_k]: V_{ki} = 1} & \widehat{q}_{K,0}(k) \Bigl(Y_{ki} - \widehat{h}_K\bigl(X_{ki}\bigr)\Bigr)\\
         &=  \frac{2}{N}   \sum_{k=1}^K \sum_{i \in [n_k]: V_{ki} = 1}  \mathcal{T}_{K}(\widehat{\beta}_K) \Bigl(Y_{ki} - \widehat{h}_K\bigl(X_{ki}\bigr)\Bigr) \\
         & \quad +  \frac{2}{N}   \sum_{k=1}^K \sum_{i \in [n_k]: V_{ki} = 1}  \left\{\widehat{q}_{K,v}(k) - \mathcal{T}_{K}(\widehat{\beta}_K)(k)\right\} \Bigl(Y_{ki} - \widehat{h}_K\bigl(X_{ki}\bigr)\Bigr).
\end{align*} 
By Lemma \ref{lemma::asymlinear}, the first term on the right-hand side of the above display satisfies the asymptotic expansion:
  \begin{align*}
      \frac{2}{N}   \sum_{k=1}^K \sum_{i \in [n_k]: V_{ki} = 1}  \mathcal{T}_{K}(\widehat{\beta}_K) \Bigl(Y_{ki} - \widehat{h}_K\bigl(X_{ki}\bigr)\Bigr)  &=    \frac{1}{N}   \sum_{k=1}^K \sum_{i=1}^{n_k} D_K(O_{ki}) \\
      & \quad +  \sum_{k=1}^K  \frac{n_k}{N} \int  \widehat{D}_K(o_k)  P^{(K)}(do_k) + o_p(N^{-1/2}).
  \end{align*}  
It follows that
   \begin{align*}
         \frac{2}{N}  \sum_{k=1}^K \sum_{i \in [n_k]: V_{ki} = 1} & \widehat{q}_{K,0}(k) \Bigl(Y_{ki} - \widehat{h}_K\bigl(X_{ki}\bigr)\Bigr)\\
         &=  \frac{2}{N}   \sum_{k=1}^K \sum_{i \in [n_k]: V_{ki} = 1} D_K(O_{ki}) + \sum_{k=1}^K  \frac{n_k}{N} \int  \widehat{D}_K(o_k)  P^{(K)}(do_k) + o_p(N^{-1/2}) \\
         & \quad +  \frac{2}{N}   \sum_{k=1}^K \sum_{i \in [n_k]: V_{ki} = 1}  \left\{\widehat{q}_{K,0}(k) - \mathcal{T}_{K}(\widehat{\beta}_K)(k)\right\} \Bigl(Y_{ki} - \widehat{h}_K\bigl(X_{ki}\bigr)\Bigr).
\end{align*} 
Arguing as in the proofs of Lemma \ref{lemma::biasExpectation} and Lemma \ref{lemma::asymlinear},  the second term satisfies:
\begin{align*}
    \frac{2}{N}   \sum_{k=1}^K \sum_{i \in [n_k]: V_{ki} = 1} &   \left\{\widehat{q}_{K,v}(k)  - \mathcal{T}_{K}(\widehat{\beta}_K)(k)\right\} \Bigl(Y_{ki} - \widehat{h}_K\bigl(X_{ki}\bigr)\Bigr) \\
    &= \frac{2}{N}   \sum_{k=1}^K \sum_{i \in [n_k]: V_{ki} = 1}  \left\{  \frac{2}{n_k}  \sum_{j \in [n_k]:  V_{kj} = 0} \beta_K(X_{kj})  - \mathcal{T}_{K}(\beta_K)(k)\right\} \Bigl(Y_{ki} -h_K\bigl(X_{ki}\bigr)\Bigr) + o_p(N^{-1/2}).
\end{align*}
Thus, arguing as in the proof of Theorem \ref{theorem::asymnormal} for the multi-fold one-step estimator, we have that
 \begin{align*}
   \widehat{\psi}_K^{\diamond} - \psi(h_K)    &=  \psi(\widehat{h}_K) - \psi(h_K)  +  \frac{2}{N}  \sum_{k=1}^K \sum_{i \in [n_k]: V_{ki} = 1}    \widehat{q}_{K,0}(k) \Bigl(Y_{ki} - \widehat{h}_K\bigl(X_{ki}\bigr)\Bigr) \\
   &=    \frac{2}{N}   \sum_{k=1}^K \sum_{i \in [n_k]: V_{ki} = 1} D_K(O_{ki})  + o_p(N^{-1/2}) \\
         & \quad + \frac{2}{N}   \sum_{k=1}^K \sum_{i \in [n_k]: V_{ki} = 1}  \left\{  \frac{2}{n_k}  \sum_{j \in [n_k]:  V_{kj} = 0} \beta_K(X_{kj})  - \mathcal{T}_{K}(\beta_K)(k)\right\} \Bigl(Y_{ki} -h_K\bigl(X_{ki}\bigr)\Bigr)  \\
          &   = \frac{2}{N}   \sum_{k=1}^K \sum_{i \in [n_k]: V_{ki} = 1}  \left\{  \frac{2}{n_k}  \sum_{j \in [n_k]:  V_{kj} = 0} \beta_K(X_{kj}) \right\} \Bigl(Y_{ki} -h_K\bigl(X_{ki}\bigr)\Bigr) + o_p(N^{-1/2}) .
\end{align*}
We conclude that
\begin{align*}
    \widehat{\psi}_K^{\diamond} - \psi(h_K)  &= \frac{2}{N}   \sum_{k=1}^K \sum_{i=1}^{n_k} 1\{V_{ki} = 1\} \left\{\frac{2}{n_k} \sum_{i=1}^{n_k} 1\{V_{ki} = 0\} \beta_K(X_{ki}) \right\}  \left\{Y_{ki} - h_K(X_{ki}) \right\}  + o_p(N^{-1/2}).
\end{align*}
Applying  \ref{cond::bound}, noting that $o_p(N^{-1/2}) = o_p(\sigma_K N^{-1/2})$ by \ref{cond::lyaponov}, and Lindeberg's CLT gives 
\(\sqrt{N/(2\sigma_K^{\diamond 2})}\, \big( \widehat{\psi}_K^{\diamond} - \psi(h_K) \big) \overset{d}{\to} \text{N}(0,1) \text{ as } K \to \infty\),
where  
\begin{align*}
    \sigma_K^{\diamond 2} := \sum_{k=1}^K \frac{n_k}{N} E_K \left[ \left\{\frac{2}{n_k} \sum_{i=1}^{n_k} 1\{V_{ki} = 0\} \beta_K(X_{ki}) \right\}^2\left\{Y_{ki} - h_K(X_{ki}) \right\}^2 \right].
\end{align*} 
Some algebra then shows that
\[
\sigma_K^{\diamond 2} = \sigma_K^2  +    \frac{2K}{N} \sum_{k=1}^K \frac{N}{Kn_k} \frac{n_k}{N} \,\sigma_{\beta}^{2}(k) \,\sigma_{\epsilon}^{2}(k),
\]
where \(\sigma_{\beta}^{2}(k) = \mathrm{Var}_{K}[\beta_{K}(X_k)]\), \(\sigma_{\epsilon}^{2}(k) = \mathrm{Var}_{K}[Y_k - h_K(X_k)]\), and \( \sigma_K^2 \) is as defined in Theorem \ref{theorem::asymnormal}. Notably, \( \sigma_K^{\diamond 2} - \sigma_K^2 \to 0 \) as \( N \to \infty \) provided that \( K/N \to 0 \), \( \lim \frac{1}{K} \sum_{k=1}^K \frac{N}{K n_k} < \infty \), and $\lim \sum_{k=1}^K \frac{n_k}{N} \left\{ \sigma_{\beta}^{2}(k) \, \sigma_{\epsilon}^{2}(k) \right\}^2 < \infty.$

\end{proof}

\begin{proof}[Proof of Theorem  \ref{theorem::asymnormalgeneral}]

Observe that the one-step estimator can be written as $\widehat{\psi}_K^* = \psi(\widehat{h}_K) + \widehat{B}_1 + \widehat{B}_2$ for debiasing terms $\widehat{B}_1$ and $\widehat{B}_2$ that we define below. By Lemma \ref{lemma::biasExpectation} with \ref{cond::grow} and \ref{cond::nuisrate2}, we have that the first debiasing term satisfies:
\begin{align*}
        \widehat{B}_1 :=  \frac{2}{N} \sum_{v \in \{0,1\}} \sum_{k=1}^K \sum_{i \in [n_k]: V_{ki} = 1-v} & \left\{\frac{2}{n_k}  \sum_{j \in [n_k]:  V_{kj} = v} \widehat{\beta}_K(X_{kj})\right\}\Bigl(Y_{ki} - \widehat{h}_K\bigl(X_{ki}\bigr)\Bigr)\\
         &=  \frac{1}{N}   \sum_{k=1}^K \sum_{i=1}^{n_k}  \mathcal{T}_{K}(\widehat{\beta}_K)(k) \Bigl(Y_{ki} - \widehat{h}_K\bigl(X_{ki}\bigr)\Bigr) + o_p(N^{-1/2}).
\end{align*} 
An argument identical to the proof of Lemma \ref{lemma::biasExpectation} similarly establishes that second debiasing term satisfies
\begin{align*}
     \widehat{B}_2 :=   \frac{2}{N} \sum_{v=0,1} \sum_{k=1}^K & \sum_{i=1}^{n_k} \mathbbm{1}\{V_{ki} = v\} \bigl(\widehat{\alpha}_K(X_{ki}) - \widehat{\Pi}_K\widehat{\alpha}_K(X_{ki})\bigr) \bigl(\widehat{r}_{K,-v}(k) - \widehat{h}_K(X_{ki})\bigr) \\
    &= \frac{2}{N} \sum_{v=0,1} \sum_{k=1}^K \sum_{i=1}^{n_k} \mathbbm{1}\{V_{ki} = v\} \bigl(\widehat{\alpha}_K(X_{ki}) - \widehat{\Pi}_K\widehat{\alpha}_K(X_{ki})\bigr) \bigl(\mathcal{T}_K(\widehat{\rho}_K)(k) - \widehat{h}_K(X_{ki})\bigr) + o_p(N^{-1/2}).
\end{align*}
Denote $\widehat{\varphi}_K^*(o) :=   \mathcal{T}_K(\widehat{\beta}_K)(z) \big(y - \widehat{h}_K(x)\big)  - \big( \widehat{\alpha}_K(x) - \widehat{\Pi}_{K} \widehat{\alpha}_K(x)\big) \big( \mathcal{T}_K(\widehat{\rho}_K)(z) - \widehat{h}_K(x)\big)$. Then, combining the previous displays, we find that 
$$  \widehat{B}_1  +   \widehat{B}_2  =   \frac{1}{N}   \sum_{k=1}^K \sum_{i=1}^{n_k} \widehat{\varphi}_K^*(O_{ki})  + o_p(N^{-1/2}).$$

Applying \ref{cond::grow}, \ref{cond::bound2}-\ref{cond::nuisrate2}, an argument identical to the proof of Lemma \ref{lemma::biasExpectation} establishes that
  \begin{align*}
      \frac{1}{N}   \sum_{k=1}^K \sum_{i=1}^{n_k} \widehat{\varphi}_K^*(O_{ki})  &=    \frac{1}{N}   \sum_{k=1}^K \sum_{i=1}^{n_k}D_K^*(O_{ki}) \\
      & \quad +  \sum_{k=1}^K  \frac{n_k}{N} \int  \widehat{\varphi}_K^*(o_k)  P^{(K)}(do_k) + o_p(N^{-1/2}).
  \end{align*}

By Theorem \ref{theorem::EIF} and \ref{cond::nuisrate2}, we have
\begin{align*}
    \psi(\widehat{h}_{K}) - \psi(h_K) + P^{(K)} \widehat{D}_{K}^* 
    &= \big   \langle   \mathcal{T}_K(\widehat{\beta}_K - \beta_K) ,   \mathcal{T}_{K} \big( \widehat{h}_K - h_K \big)  \big \rangle_{L^2_K(Z)}   \\
    & \quad - \big \langle   \big( \text{I} - \widehat{\Pi}_K \big) \widehat{\alpha}_{K} - \big( \text{I} - \Pi_{K} \big) \alpha_K \big),   \mathcal{T}_{K}^*  \mathcal{T}_K(\widehat{\rho}_K - \rho_K)  +  \widehat{h}_K - h_K   \big  \rangle_{L^2_K(X)} \\
    &= o_P(N^{-1/2}).
\end{align*}
Hence, we have the asymptotic expansion:
\begin{align*}
    \widehat{\psi}^*_K- \psi(h_K)  &=    \psi(\widehat{h}_{K})  - \psi(h_K) +  \widehat{B}_1 +   \widehat{B}_2 \\
    &= \psi(\widehat{h}_{K})  - \psi(h_K) + \frac{1}{N}   \sum_{k=1}^K \sum_{i=1}^{n_k} \widehat{\varphi}_K^*(O_{ki})  + o_p(N^{-1/2})\\
  &  = \frac{1}{N}   \sum_{k=1}^K \sum_{i=1}^{n_k}D_K^*(O_{ki}) + o_p(N^{-1/2}) \\
      & \quad + \psi(\widehat{h}_{K}) - \psi(h_K))   + P^{(K)} \widehat{D}_{K}^* \\
       &  = \frac{1}{N}   \sum_{k=1}^K \sum_{i=1}^{n_k}D_K^*(O_{ki}) + o_p(N^{-1/2}).
\end{align*}
Arguing as in the proof of Theorem \ref{theorem::asymnormal} and applying Lindeberg CLT, the result then follows.

\end{proof}

 \section{Proofs for Section~\ref{sec::nuisance}: Tikhonov regularization in npJIVE}

\subsection{Some technical lemmas}

\begin{lemma}[A matrix identity]
\label{lemma:matrixidentity} For any self-adjoint operator $A$, we have   $(A + \lambda I)^{-1} A - I = -\lambda (A + \lambda I)^{-1}$.
 
\end{lemma}
\begin{proof}
    This follows from the algebra:
\begin{align*}
(A + \lambda I)^{-1}A - I
&= (A + \lambda I)^{-1}(A + \lambda I - \lambda I) - I\\
&= (A + \lambda I)^{-1}(A + \lambda I) - \lambda(A + \lambda I)^{-1} - I\\
&= I - \lambda(A + \lambda I)^{-1} - I\\
&= -\lambda(A + \lambda I)^{-1}.
\end{align*}
\end{proof}

\begin{lemma}
Suppose that \(A\) and \(B\) satisfy, for some \(0 < C < \infty\), the inequality \(A^2 + \lambda B^2 \leq C\lambda\, A^{\nu} B^{1 - \nu}\). Then \(A = O(\lambda^{(1 + \nu)/2})\) and \(B = O(\lambda^{\nu/2})\).
\label{lemma::youngs}
\end{lemma}
\begin{proof}
By Young's inequality, \(ab \leq \frac{\nu}{2} a^{2/\nu} + \frac{2 - \nu}{\nu} b^{2/(2 - \nu)}\). Letting \(a = A^\nu\), \(b = C\lambda B^{1 - \nu}\), we obtain \(C \lambda A^\nu B^{1 - \nu} \leq \frac{\nu}{2} A^2 + \frac{2 - \nu}{\nu} (C\lambda)^{2/(2 - \nu)} B^{2(1 - \nu)/(2 - \nu)}\). Combining with the assumption, we get \(A^2 + \lambda B^2 \leq \frac{\nu}{2} A^2 + \frac{2 - \nu}{\nu} (C\lambda)^{2/(2 - \nu)} B^{2(1 - \nu)/(2 - \nu)}\), hence \(\frac{\nu}{2} A^2 + \lambda B^2 \lesssim \lambda^{2/(2 - \nu)} B^{2(1 - \nu)/(2 - \nu)}\), which implies \(B \lesssim \lambda^{\nu/2}\).

Returning to the original inequality, we have \(A^2 \leq A^2 + \lambda B^2 \lesssim \lambda A^{2\nu} B^{1 - 2\nu} \lesssim \lambda A^{2\nu} (\lambda^{\nu/2})^{1 - 2\nu}\), yielding \(A^{2 - 2\nu} \lesssim \lambda^{1 + \nu(1 - 2\nu)/2}\), so \(A \lesssim \lambda^{(2 + \nu - 2\nu^2)/(4(1 - \nu))} \leq \lambda^{(1 + \nu)/2}\), as claimed.
\end{proof}

\subsection{Lemmas for primal solution}
\label{appendix::primalreg}
Recall that \( h_K \) is the minimum-norm solution to the inverse problem \( \mathcal{T}_K(h_K) = \mu_K \). Let \( h_K^*(\lambda) \) denote the unconstrained Tikhonov-regularized solution to the inverse problem, given by  
\[
h_K^*(\lambda) := \argmin_{h \in L^2_K(X)} \quad \|\mathcal{T}_K(h) - \mu_K \|^2_{L^2_K(Z)} + \lambda \|h\|_{L^2_K(X)}.
\]
The closed-form solution is given by  
\(
h_K^*(\lambda) = \left( \mathcal{T}_K^* \mathcal{T}_K + \lambda I \right)^{-1} \mathcal{T}_K^* \mu_K.\)
For a convex function class \( \mathcal{F}_{\mathrm{primal}} \subseteq L^2_K(X) \), let \( h_K(\lambda) \) be the constrained Tikhonov-regularized solution, given by  
\[
h_K(\lambda) := \argmin_{h \in \mathcal{F}_{\mathrm{primal}}} \quad \|\mathcal{T}_K(h) - \mu_K \|^2_{L^2_K(Z)} + \lambda \|h\|_{L^2_K(X)}.
\]

We bound the regularization bias between the unconstrained Tikhonov-regularized solution \( h_K^*(\lambda) \) and the minimum-norm solution \( h_K \). This result is established in Lemma 3 of \cite{bennett2023source} (see also \cite{carrasco2007linear}), and we sketch a proof here.

\begin{lemma}[Regularization bias for unconstrained solution]
\label{lemma::regbiasunconstrained}
Under  \ref{cond::sourceprimary}, it holds that
    \begin{align*}
    \bigl\|\mathcal{T}_K\bigl(h_K^*(\lambda)-h_K\bigr)\bigr\|_{L^2_K(Z)} &\leq   \lambda^{\nu+\frac12}\,\|u_K\|_{L^2_K(X)}\\
    \| h_K^*(\lambda) - h_K \|_{L^2_K(X)} &\leq    \lambda^\nu\,\|u_K\|_{L^2_K(X)}
\end{align*}
Under  \ref{cond::funclasscontainsprimary}, we have $h_K^*(\lambda) = h_K(\lambda)$ so that the above holds for $h_K(\lambda)$ as well.
\end{lemma}
\begin{proof}
    The unconstrained Tikhonov-regularized solution is given in closed form by
    \[
    h_K^*(\lambda) = \left( \mathcal{T}_K^* \mathcal{T}_K + \lambda I \right)^{-1} \mathcal{T}_K^* \mu_K.
    \]
    The minimum-norm solution \( h_K \) satisfies \( \mathcal{T}_K(h_K) = \mu_K \).  Hence,
 \begin{align*}
\| h_K^*(\lambda) - h_K \|_{L^2_K(X)}
&= \left\|   \left( \mathcal{T}_K^* \mathcal{T}_K + \lambda I \right)^{-1} \mathcal{T}_K^*\mu_K  -h_K \right\|_{L^2_K(X)} \\
&= \left\| \left( \left( \mathcal{T}_K^* \mathcal{T}_K + \lambda I \right)^{-1} \mathcal{T}_K^*\mathcal{T}_K  - I  \right) h_K \right\|_{L^2_K(X)}\\
&= \left\| -\lambda\, \left( \mathcal{T}_K^* \mathcal{T}_K + \lambda I \right)^{-1} h_K\right\|_{L^2_K(X)} ,
\end{align*}
where the final equality follows from Lemma \ref{lemma:matrixidentity} with $A = \mathcal{T}_K^* \mathcal{T}_K $.

The source condition in  \ref{cond::sourceprimary} says that \(h_K = (\mathcal{T}_{K}^*\mathcal{T}_{K} )^{ \nu} u_K \) for some \(\nu \in [0,1]\) and \(u_K \in L^2_K(Z) \). Hence, by our previous display,
 \begin{align*}
\| h_K^*(\lambda) - h_K \|_{L^2_K(X)}&=   \left\| -\lambda\, \left( \mathcal{T}_K^* \mathcal{T}_K + \lambda I \right)^{-1}  (\mathcal{T}_{K}^*\mathcal{T}_{K} )^{ \nu} u_K\right\|_{L^2_K(X)} \\
&\leq   \lambda \left\|  \left( \mathcal{T}_K^* \mathcal{T}_K + \lambda I \right)^{-1}  (\mathcal{T}_{K}^*\mathcal{T}_{K} )^{ \nu} u_K\right\|_{L^2_K(X)}. 
\end{align*}
Now, the spectral theorem for self-adjoint operators \citep{kato2013perturbation} shows that
\begin{align*}
\lambda \bigl\|   (\mathcal{T}_K^*\mathcal{T}_K + \lambda I)^{-1} (\mathcal{T}_K^*\mathcal{T}_K)^\nu u_K\bigr\|_{L^2_K(X)} &= \lambda \bigl\|   (\mathcal{T}_K^*\mathcal{T}_K + \lambda I)^{-1} (\mathcal{T}_K^*\mathcal{T}_K)^\nu u_K\bigr\|_{L^2_K(X)}\\
&\le  \sup_{x\ge0}\frac{\lambda\,x^\nu}{x + \lambda}\,\|u_K\|_{L^2_K(X)}\\
&\le\lambda^\nu\,\|u_K\|_{L^2_K(X)}.
\end{align*}
Here, we used that for \(0 \le \nu \le 1\), the scalar function $x \mapsto \frac{\lambda\,x^\nu}{x + \lambda}$ attains its maximum \(\lambda^\nu\) at \(x = \lambda\). Hence, combining the above with our previous displays, we obtain the strong norm bound:
 $$\| h_K^*(\lambda) - h_K \|_{L^2_K(X)} \leq    \lambda^\nu\,\|u_K\|_{L^2_K(X)}.$$
 A nearly identical proof for the weak norm gives
\begin{align*}
  \bigl\|\mathcal{T}_K\bigl(h_K^*(\lambda)-h_K\bigr)\bigr\|_{L^2_K(Z)} &=   \bigl\|(\mathcal{T}_K^*\mathcal{T}_K)^{\frac{1}{2}}\bigl(h_K^*(\lambda)-h_K\bigr)\bigr\|_{L^2_K(X)}\\
  &\leq    \lambda\,\bigl\|( \mathcal{T}_K^*\mathcal{T}_K )^{\tfrac12}
      ( \mathcal{T}_K^*\mathcal{T}_K + \lambda I )^{-1}
      ( \mathcal{T}_K^*\mathcal{T}_K )^\nu\,u_K\bigr\|_{L^2_K(X)}\\
  &\le  \lambda\,\sup_{x\ge0}\frac{x^{\nu+\frac12}}{\,x+\lambda\,}\;\|u_K\|_{L^2_K(X)}\\
  &\le  \lambda^{\nu+\frac12}\,\|u_K\|_{L^2_K(X)}.
\end{align*}

\end{proof}

In our work, we assume that the constrained solution \( h_K(\lambda) \) equals the unconstrained one \( h_K^*(\lambda) \) (Conditions~\ref{cond::funclasscontainsprimary} and~\ref{cond::funclasscontainsdual}), so that the above lemma can be applied to obtain fast rates. The next lemma shows that these conditions can be relaxed, potentially at the cost of a slower convergence rate. In particular, the regularization bias of the constrained solution \( h_K(\lambda) \in \mathcal{F}_{\mathrm{primal}} \) may exhibit worse dependence on the source condition exponent \( \nu \) than in Lemma~\ref{lemma::regbiasunconstrained} when \( \nu > 1/2 \). It is unclear whether this degradation is an artifact of our proof technique or a fundamental limitation. In the case where \( \mathcal{F}_{\mathrm{primal}} \) is correctly specified for \( h_K(\lambda) \), so that \( h_K(\lambda) = h_K^*(\lambda) \), the faster rate from Lemma~\ref{lemma::regbiasunconstrained} can still be applied.

\begin{lemma}[Regularization bias with convexity constraints]
\label{lemma::regbiasconstrained}
    Assume  \ref{cond::sourceprimary} with $\nu \in [0,1]$ and  $\sup_K  \|   u_K \| < \infty$, and that $h_K \in \mathcal{F}_{\mathrm{primal}}$. If $\nu \leq \frac{1}{2}$,  we have $\|\mathcal{T}_K (h_K(\lambda) - h_K)\| = O(\lambda^{\frac{1}{2} + \nu})$ and $\|h_K(\lambda) - h_K\| = O(\lambda^{\nu})$. Otherwise if $\frac{1}{2} \leq \nu < 1$, we have $\|\mathcal{T}_K(h_K(\lambda) - h_K)\| = O(\lambda^{\frac{1 + \nu}{2} })$ and $\|(h_K(\lambda) - h_K)\| = O(\lambda^{\frac{\nu}{2}})$
    \end{lemma}
\begin{proof}
Since $h_K(\lambda)$ minimizes the population risk over a convex function class $\mathcal{F}_{\mathrm{primal}}$, the KKT conditions imply, for all $h \in \mathcal{F}_{\mathrm{primal}}$, that
\[
\langle \mathcal{T}_K(h - h_K(\lambda)),\, \mathcal{T}_K(h_K(\lambda)) - \mathcal{T}_K(h) \rangle_{L^2_K(Z)} + \lambda \langle h_K(\lambda),\, h - h_K(\lambda) \rangle_{L^2_K(X)} \geq 0.
\]
Substituting $h = h_K$, we find that
\[
\langle \mathcal{T}_K(h_K - h_K(\lambda)),\, \mathcal{T}_K(h_K(\lambda)) - \mathcal{T}_K(h_K) \rangle_{L^2_K(Z)} + \lambda \langle h_K(\lambda),\, h_K - h_K(\lambda) \rangle_{L^2_K(X)} \geq 0.
\]
Letting $e_K = h_K - h_K(\lambda)$, this yields the inequality
\begin{equation}
\|\mathcal{T}_K e_K\|_{L^2_K(Z)}^2 + \lambda \|e_K\|_{L^2_K(X)}^2 \leq \lambda \left| \langle e_K, h_K  \rangle \right|. \label{eqn::basicKKT}
\end{equation}
Condition \ref{cond::sourceprimary} says that  $h_K = (\mathcal{T}_K^* \mathcal{T}_K)^\nu u_K$ for some $u_K \in L^2_K(X)$. Hence,
\begin{align*}
\left| \langle e_K, h_K \rangle \right| 
&= \left| \langle e_K,\, (\mathcal{T}_K^* \mathcal{T}_K)^\nu u_K \rangle \right| \\
&= \left| \langle (\mathcal{T}_K^* \mathcal{T}_K)^\nu   e_K,\,u_K \rangle \right| \\
&\le \| (\mathcal{T}_K^* \mathcal{T}_K)^{\nu} e_K \| \cdot \|   u_K \|.
\end{align*}
For any self-adjoint operator $A$, the Heinz--Kato interpolation inequality \citep{steinerberger2019refined} implies
\[
\|A^c x\| \le \|A x\|^c \|x\|^{1-c} \quad \text{for any } 0 \leq c \leq 1.
\]
\newline \textbf{Fast rate ($0 \leq \nu \leq 1/2$).} Applying the interpolation inequality with $A = (\mathcal{T}_K^* \mathcal{T}_K)^{1/2}$, we find, for any $\nu \leq 1/2$, that
\begin{align*}
\|(\mathcal{T}_K^* \mathcal{T}_K)^{\nu} e_K\| 
&= \|A^{2\nu} e_K\| \\
&\le \|A e_K\|^{2\nu} \|e_K\|^{1-2\nu} \\
&= \|\mathcal{T}_K e_K\|^{2\nu} \|e_K\|^{1-2\nu}.
\end{align*}
Combining the above bound with \eqref{eqn::basicKKT} and using that $\sup_K  \|   u_K \| < \infty$ gives
\[
\|\mathcal{T}_K e_K\|^2 + \lambda \|e_K\|^2 
\lesssim  \lambda\,   \|\mathcal{T}_K e_K\|^{2\nu} \|e_K\|^{1 - 2\nu}.
\]
By Lemma  \ref{lemma::youngs} applied with exponent $2\nu$, we have $\|\mathcal{T}_K e_K\| = O(\lambda^{\frac{1}{2} + \nu})$ and $\|e_K\| = O(\lambda^{\nu})$. 
\newline \textbf{Slow rate ($1/2 \leq \nu \leq 1$).} Applying the interpolation inequality with $A = (\mathcal{T}_K^* \mathcal{T}_K) $, we find, for any $\nu \leq 1$,
\begin{align*}
\|(\mathcal{T}_K^* \mathcal{T}_K)^{\nu} e_K\| 
&= \|A^{\nu} e_K\| \\
&\le \|A e_K\|^{\nu} \|e_K\|^{1-\nu} \\
&= \|\mathcal{T}_K^*(\mathcal{T}_K e_K)\|^{\nu} \|e_K\|^{1-\nu} \\
&\le \|\mathcal{T}_K^*\|^{\nu} \|\mathcal{T}_K e_K\|^{\nu} \|e_K\|^{1-\nu}\\
&\lesssim \|\mathcal{T}_K e_K\|^{\nu} \|e_K\|^{1-\nu}.
\end{align*}
Combining the above bound with \eqref{eqn::basicKKT}  gives
\[
\|\mathcal{T}_K e_K\|^2 + \lambda \|e_K\|^2 
\lesssim  \lambda\,   \|\mathcal{T}_K e_K\|^{\nu} \|e_K\|^{1-\nu}.
\]
By Lemma  \ref{lemma::youngs} applied with exponent $\nu$, we have $\|\mathcal{T}_K e_K\| = O(\lambda^{\frac{1 + \nu}{2} })$ and $\|e_K\| = O(\lambda^{\nu/2})$.

\end{proof}

\subsection{Dual solution: proof of Lemma \ref{lemma::approxstrongident}}

  Recall that \( \beta_{K}(\lambda) \) is the solution to the regularized inverse problem \( (\mathcal{T}_{K}^* \mathcal{T}_{K} + \lambda I) \beta_{K}(\lambda) = \alpha_K \), and that $\beta_K$ is the solution to \( (\mathcal{T}_{K}^* \mathcal{T}_{K} ) \beta_{K}  = \Pi_K \alpha_K \)

\begin{proof}[Proof of Lemma \ref{lemma::approxstrongident}]
 \textbf{Strong norm rate.} For the strong norm, we will establish the following two bounds:
\begin{align}
     \|(\mathcal{T}_{K}^* \mathcal{T}_{K} + \lambda I)^{-1} \Pi_{K} \alpha_K - \beta_{K} \|& \leq O(\lambda^\nu)\,\|w_{K}\| \\
     \|(\mathcal{T}_{K}^* \mathcal{T}_{K} + \lambda I)^{-1} (\alpha_K - \Pi_{K} \alpha_K)\| &\leq O(\lambda^{-1})\,\|\alpha_K - \Pi_{K} \alpha_K\|
\end{align}
Taking the above bounds as given for the moment, observe that:
\[
\beta_{K}(\lambda) - \beta_{K}
\;=\;
\Bigl[(\mathcal{T}_{K}^* \mathcal{T}_{K} + \lambda I)^{-1}\,\Pi_{K} \alpha_K - \beta_{K}\Bigr]
\;+\;
(\mathcal{T}_{K}^* \mathcal{T}_{K}  + \lambda I)^{-1}\,(\alpha_K - \Pi_{K} \alpha_K).
\]
Thus,
\[
\|\beta_{K}(\lambda) - \beta_{K}\|
\;=\; O(\lambda^\nu)\,\|w_{K}\|
\;+\;
O(\lambda^{-1})\,\|\alpha_K - \Pi_{K} \alpha_K\|.
\]

\paragraph{Weak norm rate.} For the weak norm, observe that:
\[
\mathcal{T}_{K} \beta_{K}(\lambda) - \mathcal{T}_{K}  \beta_{K}
\;=\;
\Bigl[\mathcal{T}_{K} (\mathcal{T}_{K}^* \mathcal{T}_{K} + \lambda I)^{-1}\,\Pi_{K} \alpha_K - \mathcal{T}_{K} \beta_{K}\Bigr] 
\;+\;
\mathcal{T}_{K} (\mathcal{T}_{K}^* \mathcal{T}_{K}  + \lambda I)^{-1}\,(\alpha_K - \Pi_{K} \alpha_K).
\]
We will establish the following two bounds:
\begin{align}
     \|\mathcal{T}_{K}(\mathcal{T}_{K}^* \mathcal{T}_{K} + \lambda I)^{-1} \Pi_{K} \alpha_K - \beta_{K} \|& \leq O(\lambda^{\nu + \frac{1}{2}})\,\|w_{K}\| \\
     \|\mathcal{T}_{K}(\mathcal{T}_{K}^* \mathcal{T}_{K} + \lambda I)^{-1} (\alpha_K - \Pi_{K} \alpha_K)\| &\leq O(\lambda^{-1/2})\,\|\alpha_K - \Pi_{K} \alpha_K\|.
\end{align}
Thus,
\[
\|\mathcal{T}_{K}\beta_{K}(\lambda) - \mathcal{T}_{K}\beta_{K}\|
\;=\; O(\lambda^{\nu + \frac{1}{2}})\,\|w_{K}\|
\;+\;
O(\lambda^{-1/2})\,\|\alpha_K - \Pi_{K} \alpha_K\|.
\]

\paragraph{Proof of regularization error bounds.} These bounds follow from the proof of Lemma 3 in \cite{bennett2023source}. Since \(\Pi_{K} \alpha_K \in \mathcal{R}(\mathcal{T}_{K}^*)\) and \(\mathcal{T}_{K}^* \mathcal{T}_{K}\) is invertible on \(\mathcal{R}(\mathcal{T}_{K}^*)\), note $\Pi_{K} \alpha_K \;=\; \mathcal{T}_{K}^* \mathcal{T}_{K}\,\beta_{K}$ implies that
\[
(\mathcal{T}_{K}^* \mathcal{T}_{K} + \lambda I)^{-1}\,\Pi_{K} \alpha_K 
\;-\;\beta_{K} 
\;=\;
(\mathcal{T}_{K}^* \mathcal{T}_{K} + \lambda I)^{-1}\,\mathcal{T}_{K}^* \mathcal{T}_{K}\,\beta_{K} 
\;-\;\beta_{K}.
\]
 A standard argument in regularization theory shows that if 
$\beta_{K}  = (\mathcal{T}_{K}^* \mathcal{T}_{K})^{\nu} w_{K}$, then the norm of the latter term in the previous display satisfies
\[
\Bigl\|\Bigl[I \;-\;(\mathcal{T}_{K}^* \mathcal{T}_{K}+\lambda I)^{-1}\mathcal{T}_{K}^* \mathcal{T}_{K}\Bigr]\beta_{K}\Bigr\|_{L^2_K(X)} 
= O(\lambda^\nu)\,\|w_{K}\|_{L^2_K(X)}.
\]
For details, see the proof of Lemma \ref{lemma::regbiasunconstrained} and the proof of Lemma 3 in \cite{bennett2023source}. Hence, \(\|(\mathcal{T}_{K}^* \mathcal{T}_{K} + \lambda I)^{-1} \Pi_{K} \alpha_K \;-\;\beta_{K}\|_{L^2_K(X)}
\;=\;
O(\lambda^\nu)\,\|w_{K}\|_{L^2_K(X)}.\)
The proof of Lemma \ref{lemma::regbiasunconstrained} also establishes the weak norm bound \(\|\mathcal{T}_{K}((\mathcal{T}_{K}^* \mathcal{T}_{K} + \lambda I)^{-1} \Pi_{K} \alpha_K \;-\;\beta_{K})\|_{L^2_K(X)}
\;=\;
O(\lambda^{\frac{1}{2} + \nu})\,\|w_{K}\|_{L^2_K(X)}.\)

\paragraph{Proof of approximation error bounds.} For the strong norm bound, consider \((A + \lambda I)^{-1}\bigl(\alpha_K - \Pi_{K} \alpha_K\bigr),\)
where \(A = \mathcal{T}_{K}^* \mathcal{T}_{K}\). Then
\begin{align*}
    \|(A + \lambda I)^{-1} (\alpha_K - \Pi_{K} \alpha_K)\|_{L^2_K(X)} 
    & \leq \|(A + \lambda I)^{-1}\|_{\text{op}} \|\alpha_K - \Pi_{K} \alpha_K\|_{L^2_K(X)} \\
    & \leq \frac{1}{\lambda} \|\alpha_K - \Pi_{K} \alpha_K\|_{L^2_K(X)},
\end{align*}
where we use the standard bound \(\|(A + \lambda I)^{-1}\|_{\text{op}} = \sup_{\sigma \in \text{spec}(A)} \frac{1}{\sigma + \lambda} \lesssim \lambda^{-1}.\) For the weak norm bound, similarly, we can show that
\begin{align*}
    \| \mathcal{T}_{K}(\mathcal{T}_{K}^* \mathcal{T}_{K} + \lambda I)^{-1} (\alpha_K - \Pi_{K} \alpha_K) \|_{L^2_K(X)} 
    & \leq \| \mathcal{T}_{K}(\mathcal{T}_{K}^* \mathcal{T}_{K} + \lambda I)^{-1} \|_{\text{op}} \|\alpha_K - \Pi_{K} \alpha_K\|_{L^2_K(X)} \\
    & \leq \sup_{\sigma \in \text{spec}(\mathcal{T}_{K}^* \mathcal{T}_{K})} \frac{\sqrt{\sigma}}{\sigma + \lambda} \|\alpha_K - \Pi_{K} \alpha_K\|_{L^2_K(X)} \\
    & \leq O\left(\lambda^{-1/2} \right) \|\alpha_K - \Pi_{K} \alpha_K\|_{L^2_K(X)}.
\end{align*}

\end{proof}

 \section{Proofs for Section~\ref{sec::nuisance}: Regret bounds in npJIVE}

\subsection{Local maximal inequality}

  Let $O_1,\ldots,O_N \in \mathcal{O}$ be independent random variables. For any function $f:\mathcal{O} \to \mathbb{R}$, define
\begin{align}
    \|f\| := \sqrt{\frac{1}{N} \sum_{i=1}^N \mathbb{E}[f(O_i)^2]}.
\end{align}
 The following lemma provides a local maximal inequality and is a restatement of Lemma 9 in \cite{bibaut2024nonparametric} (see also Lemma~11 of \cite{foster2023orthogonal}).

\begin{lemma}[Local maximal inequality]\label{lemma:loc_max_ineq}
Let \( \mathcal{F} \) be a star-shaped class of functions such that \( \sup_{f \in \mathcal{F}} \|f\|_{\infty} \le B \), and define \( r_{\max} := \sup_{f \in \mathcal{F}} \|f\| \). Let \( \delta_N > 0 \) satisfy the critical radius condition \( \mathcal{R}_N(\mathcal{F}, \delta_N) \le \delta_N^2 \), where
\[
\mathcal{R}_N(\mathcal{F}, r) := \mathbb{E}\left[\sup_{\|f\| \le r} \frac{1}{N} \sum_{i=1}^N \epsilon_i f(O_i)\right],
\]
and \( \epsilon_i \sim \mathrm{Unif}\{\pm 1\} \) are i.i.d. Rademacher random variables. Suppose that $ N^{-1/2} \sqrt{\log \log \delta_N} =  o(\delta_N)$. Then, for \(N\) large enough so that \( N^{-1/2} \sqrt{\log \log (1/\delta_N)} \le \delta_N \), the following holds with probability at least \(1 - e^{-u^2}\) for every \( f \in \mathcal{F} \): 
\[
\frac{1}{N} \sum_{i=1}^N \left(f(O_i) - \mathbb{E}[f(O_i)]\right)
\lesssim \delta_N^2 +  \delta_N \|f\| + \frac{u \|f\|}{\sqrt{N}} + \frac{B u^2}{N}.
\]
\end{lemma}

\begin{proof}
Fix any \( r \ge 0 \) and define \( \mathcal{S}(r) := \{ f \in \mathcal{F} : \|f\| \le r \} \). By Bousquet’s inequality, since each \( f(O_i) \) is bounded by \( B \) and satisfies \( \operatorname{Var}(f(O_i)) \le \mathbb{E}[f(O_i)^2] \le r^2 \), we have, with probability at least \( 1 - e^{-u^2} \),
\[
\sup_{f \in \mathcal{S}(r)} \frac{1}{N} \sum_{i=1}^N (f(O_i) - \mathbb{E}[f(O_i)])
\lesssim \mathbb{E} \left[ \sup_{f \in \mathcal{S}(r)} \frac{1}{N} \sum_{i=1}^N \epsilon_i f(O_i) \right] + \frac{u r}{\sqrt{N}} + \frac{B u^2}{N}.
\]
By definition, the Rademacher complexity term is \( \mathcal{R}_N(\mathcal{F}, r) := \mathbb{E} \left[ \sup_{\|f\| \le r} \frac{1}{N} \sum_{i=1}^N \epsilon_i f(O_i) \right] \), so the above becomes
\[
\sup_{\|f\| \le r} \frac{1}{N} \sum_{i=1}^N (f(O_i) - \mathbb{E}[f(O_i)])
\lesssim \mathcal{R}_N(\mathcal{F}, r) + \frac{u r}{\sqrt{N}} + \frac{B u^2}{N}.
\]
Since \( \mathcal{F} \) is star-shaped and \( \mathcal{R}_N(\mathcal{F}, \delta_N) \le \delta_N^2 \), we have \( \mathcal{R}_N(\mathcal{F}, r) \le (r / \delta_N) \mathcal{R}_N(\mathcal{F}, \delta_N) \le r \delta_N \).  Substituting this yields:
\[
\sup_{\|f\| \le r} \frac{1}{N} \sum_{i=1}^N (f(O_i) - \mathbb{E}[f(O_i)])
\lesssim \delta_N r + \frac{u r}{\sqrt{N}} + \frac{B u^2}{N}.
\]

To extend the bound uniformly over all \(f \in \mathcal{F}\), apply a peeling argument. Let
\( r_{\max} := \sup_{f \in \mathcal{F}} \|f\| \) and
\( J_N := \lceil \log_2(r_{\max}/\delta_N) \rceil \). For \(j = 0, 1, \dots, J_N\), define
\( r_j = 2^j \delta_N \),
\( r_{-1} = 0 \), and
\( \mathcal{S}_j := \{ f \in \mathcal{F} : r_{j-1} < \|f\| \le r_j \} \). By the fixed-\(r\) result applied to \(r = r_j\), each shell satisfies, with probability at least \(1 - e^{-u^2}\),
\[
\sup_{f \in \mathcal{S}_j} (P_N - P)f \lesssim \delta_N r_j + \frac{u r_j}{\sqrt{N}} + \frac{B u^2}{N}.
\]
Applying a union bound over \(j = 0, \dots, J_N\), this holds for all shells simultaneously with probability at least \(1 - (J_N + 1)e^{-u^2}\). Now fix any \(f \in \mathcal{F}\) and let \(j\) be such that \(r_{j-1} < \|f\| \le r_j\). Then, by monotonicity of the bound in \(r\),
\[
(P_N - P)f \le \sup_{\|g\| \le r_j} (P_N - P)g \lesssim \delta_N r_j + \frac{u r_j}{\sqrt{N}} + \frac{B u^2}{N} \lesssim \delta_N \|f\| + \frac{u \|f\|}{\sqrt{N}} + \frac{B u^2}{N}.
\]
Applying a union bound over \(J_N = \lceil \log_2(r_{\max}/\delta_N) \rceil\) shells introduces an additive \(\sqrt{\log J_N}\) overhead into the deviation term. Thus, with probability at least \(1 - e^{-u^2}\), we have for every \(f \in \mathcal{F}\),
\[
(P_N - P)f \lesssim \delta_N \|f\| + \frac{(u + \sqrt{\log \log(1/\delta_N)}) \|f\|}{\sqrt{N}} + \frac{B (u + \sqrt{\log \log(1/\delta_N)})^2}{N}.
\]
Finally, using the assumption \( N^{-1/2} \sqrt{\log \log (1/\delta_N)} = o(\delta_N) \), we conclude that for \(N\) large enough so that \( N^{-1/2} \sqrt{\log \log (1/\delta_N)} \le \delta_N \), the following holds with probability at least \(1 - e^{-u^2}\) for every \( f \in \mathcal{F} \):
\[
(P_N - P)f \lesssim \delta_N^2 + \delta_N \|f\| + \frac{u\|f\|}{\sqrt{N}} + \frac{B u^2 }{N}.
\]

\end{proof}

\subsection{Technical lemmas and notation}

In this section, we provide technical lemmas that bound key remainder terms arising in the regret analysis of npJIVE. These bounds are derived using the local maximal inequality for empirical processes established in the previous subsection.

For a function class $\mathcal{F}$, we define the empirical Rademacher complexities:
\begin{align*}
    \widehat{\mathcal{R}}_N(\mathcal{F}, \delta) &= \sum_{v \in \{0,1\}}\sup_{\substack{f \in \mathcal{F} : \| f \|_{L^2_K} \leq \delta}} \frac{2}{N} \sum_{i: V_{ki} = v} \epsilon_i f(O_{ki}), \\
     \widehat{\mathcal{R}}_{n_k}(\mathcal{F}, \delta) &=   \sup_{\substack{f \in \mathcal{F} : \| f \|_k \leq \delta}} \frac{2}{n_k} \sum_{i  \in [n_k]: V_{ki} = 0} \epsilon_i f(O_{ki}) ,
\end{align*}
where we define the conditional norm  \(\|f\|_k := \left\{ E_K \left[ f^2(O_k) \right] \right\}^{1/2}.\)

\begin{lemma}
\label{lemma::localmax1}
Suppose \(\mathcal{F}\) is a star-shaped, uniformly bounded function class. Define \(\delta_n := \max_{k \in [K]} \delta_{n_k}\) and \(n_{\min} := \min_{k \in [K]} n_k\), where each \(\delta_{n_k}\) satisfies \(E_K[\widehat{\mathcal{R}}_{n_k}(\mathcal{F}, \delta)] \le \delta^2\sqrt{n_k}\). Assume \(n_k^{-1/2}\sqrt{\log\log(1/\delta_{n_k})} = o(\delta_{n_k})\). Then, for any \(\delta \in (0,1)\), with probability at least \(1 - \delta\), the following holds uniformly for all \(f \in \mathcal{F}\) and \(k \in [K]\):
\[
|\widehat{\mathcal{T}}_{K,0}(f)(k) - \mathcal{T}_K(f)(k)| \lesssim \delta_n\|f\|_k + \delta_n^2 + \tfrac{\log(K/\delta)}{n_{\min}}.
\]
\end{lemma}
\begin{proof}
Fix \(\delta \in (0,1)\) and set \(u = \sqrt{\log(K/\delta)}\). By Lemma \ref{lemma:loc_max_ineq}, with probability at least \(1 - e^{-u^2}\), uniformly over all \(f \in \mathcal{F}\) and \(k \in [K]\),
\[
|\widehat{\mathcal{T}}_{K,0}(f)(k) - \mathcal{T}_K(f)(k)| \lesssim \delta_{n_k}^2 + \delta_{n_k}\|f\|_k + \frac{u\,\|f\|_k}{\sqrt{n_k}} + \frac{u^2}{n_k}.
\]
Since \(e^{-u^2} = \delta/K\), applying a union bound over \(k\) yields probability at least \(1 - \delta\). Moreover,
\[
\frac{u}{\sqrt{n_k}} = \sqrt{\frac{\log(K/\delta)}{n_k}}, \quad \frac{u^2}{n_k} = \frac{\log(K/\delta)}{n_k}.
\]
Thus, with probability at least \(1 - \delta\), uniformly over all \(f\) and \(k\),
\[
|\widehat{\mathcal{T}}_{K,0}(f)(k) - \mathcal{T}_K(f)(k)| \lesssim \delta_{n_k}\|f\|_k + \delta_{n_k}^2 + \|f\|_k\sqrt{\frac{\log(K/\delta)}{n_{\min}}} + \frac{\log(K/\delta)}{n_{\min}}.
\]
Finally, using \(\delta_n = \max_k \delta_{n_k}\) and absorbing the \(\|f\|_k\)-term into \(\delta_n\|f\|_k\) up to constants, we obtain
\[
|\widehat{\mathcal{T}}_{K,0}(f)(k) - \mathcal{T}_K(f)(k)| \lesssim \delta_n\|f\|_k + \delta_n^2 + \frac{\log(K/\delta)}{n_{\min}},
\]
as claimed.
\end{proof}


 

\begin{lemma}
\label{lemma::localmax2}
    Suppose that $\mathcal{F}$ is a uniformly bounded function class. Define the function class $\mathcal{G} := \{g_f: f \in \mathcal{F} \}$, where $g_f:= \mathcal{T}_{K}(f)\{f -  \mathcal{T}_{K}(f)\}$.  Let $\delta_N$ satisfy the critical inequality $E_K[\widehat{\mathcal{R}}_N(\text{star}(\mathcal{G}), \delta)]\leq \sqrt{N} \delta^2$. Assume that $ N^{-1/2} \sqrt{\log \log \delta_{N}} =  o(\delta_{N})$.  Then,  with probability at least $1 - e^{-N\delta_N^2}$, we have
\begin{align*}
   \langle  \mathcal{T}_{K}(f), [\widehat{\mathcal{T}}_{K,v} -  \mathcal{T}_{K}](f) \rangle_{L^2_K(Z)}   \lesssim \delta_N   \|\mathcal{T}_{K}(f)\|_{L^2_K(Z)} + \delta_N^2 .
\end{align*} 

\end{lemma}
\begin{proof}

For each $v \in \{0,1\}$, we can write
\begin{align*}
    \langle  \mathcal{T}_{K}(f), \widehat{\mathcal{T}}_{K,v}(f) - \mathcal{T}_{K}(f)\rangle_{L^2_K(Z)} =  \frac{1}{N} \sum_{k=1}^K \sum_{i=1}^{n_k} 1(V_{ki} = v) \mathcal{T}_{K}(f)(k)\left\{f(X_{ki}) - \mathcal{T}_{K}(f)(k) \right\},
\end{align*}
where the right-hand side is an empirical mean of approximately $N/2$ random variables with mean zero. We will use this fact to show that this term uniformly concentrates around zero.

    Define the function class $\mathcal{G} := \{g_f: f \in \mathcal{F} \}$, where $g_f:= \mathcal{T}_{K}(f)\{f -  \mathcal{T}_{K}(f)\}$. Applying the local maximal inequality stated in Lemma \ref{lemma:loc_max_ineq}, we find that for any $u \geq 0$, with probability at least $1 - e^{-u^2}$, the following holds for any $f \in \mathcal{F}$:
\begin{align*}
  \langle  \mathcal{T}_{K}(f), [\widehat{\mathcal{T}}_{K,v} -  \mathcal{T}_{K}](f) \rangle_{L^2_K(Z)} \lesssim \delta_N \|g_f\|_{L^2_K} + \frac{u\|g_f\|_{L^2_K}}{\sqrt{N}} + \frac{u^2}{N}.
\end{align*}
where $\delta_N$ satisfies the critical inequality $E_K[\widehat{\mathcal{R}}_N(\mathcal{G}, \delta)]\leq \sqrt{N} \delta^2$. Since $\mathcal{F}$ is uniformly bounded, we can further bound $\|g_f\|_{L^2_K}$ loosely as
\begin{align*}
    \|g_f\|_{L^2_K} &= \sqrt{\sum_{k=1}^{K} \frac{n_k}{N}\{ \mathcal{T}_{K}(f)(k)\}^2 E_K\left[\{f(X_{k}) - \mathcal{T}_{K}(f)(k)\}^2\right]} \\
        &\lesssim \min \left\{ \|\mathcal{T}_{K}(f)\|_{L^2_K(Z)},\, \|f\|_{L^2_K} \right\}.
\end{align*}
Thus, for any $u \geq 0$, with probability at least $1 - e^{-u^2}$, the following holds for any $f \in \mathcal{F}$:
\begin{align*}
  \langle  \mathcal{T}_{K}(f), [\widehat{\mathcal{T}}_{K,v} -  \mathcal{T}_{K}](f) \rangle_{L^2_K(Z)} \lesssim \delta_N   \|\mathcal{T}_{K}(f)\|_{L^2_K(Z)}  + \frac{u}{\sqrt{N}} \|\mathcal{T}_{K}(f)\|_{L^2_K(Z)}  + \frac{u^2}{N}.
\end{align*}
Setting $u = \delta_n \sqrt{N}$, with probability at least $1 - e^{-N\delta_N^2}$, we have
\begin{align*}
   \langle  \mathcal{T}_{K}(f), [\widehat{\mathcal{T}}_{K,v} -  \mathcal{T}_{K}](f) \rangle_{L^2_K(Z)} \lesssim \delta_N   \|\mathcal{T}_{K}(f)\|_{L^2_K(Z)} + \delta_N^2 .
\end{align*}

\end{proof}

\begin{lemma}
  \label{lemma::localmax3} Suppose that $\mathcal{F}$ be a uniformly bounded function class, and let $f_0 \in \mathcal{F}$. Define the function class $\mathcal{G} := \{g_f: f \in \mathcal{F} \}$, where $g_f:= f \cdot \mathcal{T}_{K}(f - f_0)$. Let $\delta_N$ satisfy the critical inequality $E_K[\widehat{\mathcal{R}}_N(\text{star}(\mathcal{G}), \delta)]\leq \sqrt{N} \delta^2$. Assume that $ N^{-1/2} \sqrt{\log \log \delta_{N}} =  o(\delta_{N})$. Then,  with probability at least $1 - e^{-N\delta_N^2}$, we have
\begin{align*}
       \langle  \widehat{\mathcal{T}}_{K,v}(f), \mathcal{T}_{K}(f - f_0) \rangle_{L^2_K(Z)} -  \langle  \mathcal{T}_{K}(f), \mathcal{T}_{K}(f - f_0) \rangle_{L^2_K(Z)}  \lesssim \delta_N   \|\mathcal{T}_{K}(f - f_0)\|_{L^2_K(Z)} + \delta_N^2 .
\end{align*} 
 
\end{lemma}
\begin{proof}
      The proof of this lemma is nearly identical to the proof of Lemma \ref{lemma::localmax2}. By independence of the folds, we can express \(  \langle  \widehat{\mathcal{T}}_{K,v}(f), \mathcal{T}_{K}(f - f_0) \rangle_{L^2_K(Z)} \) as the empirical mean of mean zero-random variables:
\begin{align*}
  \frac{1}{N} \sum_{k=1}^K \sum_{i=1}^{n_k} 1(V_{ki} = v)\{f(X_{ki}) - \mathcal{T}_{K}(f)\} \mathcal{T}_{K}(f - f_0)(k).
\end{align*}
Applying Lemma \ref{lemma:loc_max_ineq}, with probability at least $1 - e^{-N\delta_N^2}$, we have
\begin{align*}
 \langle  \widehat{\mathcal{T}}_{K,v}(f - \mathcal{T}_{K}f), \mathcal{T}_{K}(f - f_0) \rangle_{L^2_K(Z)}  \lesssim \delta_N   \|\mathcal{T}_{K}(f - f_0)\|_{L^2_K(Z)} + \delta_N^2 .
\end{align*}

\end{proof}

 \begin{lemma}
 \label{lemma::localmax_smalln_1}
    Suppose that $\mathcal{F}$ is a uniformly bounded function class. Define the function class $\widehat{\mathcal{G}} := \{\widehat{g}_f : f \in \mathcal{F}\}$ as the set of all elements of the form $ \widehat{g}_f :=[\widehat{\mathcal{T}}_{K,0} - \mathcal{T}_{K}](f) \{f - \mathcal{T}_{K}(f)\}$. Let $\delta_N$ satisfy the critical inequality $E_K[\widehat{\mathcal{R}}_N(\text{star}(\widehat{\mathcal{G}}), \delta) \mid \mathcal{D}_{n,0}] \leq \sqrt{N} \delta^2$. Let $\delta_n^2 := \max_{k \in [K]} \delta_{n_k}^2$, where each $\delta_{n_k}$ satisfies the critical inequality $E_K[\widehat{\mathcal{R}}_{n_k}(\text{star}(\mathcal{F}), \delta)] \leq \delta^2\sqrt{n_k}$. Assume that $ N^{-1/2} \sqrt{\log \log \delta_{N}} =  o(\delta_{N})$ and $ n_k^{-1/2} \sqrt{\log \log \delta_{n_k}} =  o(\delta_{n_k})$.  Then, for any $\delta > 0$, with probability at least $1- \delta - e^{-N\delta_N^2} $, the following holds uniformly for any $f \in \mathcal{F}$:
\begin{align*}
 \langle  \widehat{\mathcal{T}}_{K,0}(f) - \mathcal{T}_{K}(f), \widehat{\mathcal{T}}_{K,1}(f) - \mathcal{T}_{K}(f) \rangle_{L^2_K(Z)} &\lesssim \delta_N \left( \delta_n +  \sqrt{\log(K/\delta) / n_{\min}} \right) \|f - f_0\|_{L^2_K}   + \delta_N^2.
\end{align*} 
 \end{lemma}

   \begin{lemma}
  \label{lemma::localmax_smalln_2}
    Suppose that $\mathcal{F}$ is a uniformly bounded function class, and let $ f_0 \in \mathcal{F}$. Define the function class $\widehat{\mathcal{G}} := \{\widehat{g}_f : f \in \mathcal{F}\}$ as the set of all elements of the form $ \widehat{g}_f :=\widehat{\mathcal{T}}_{K,v}(f) \{f - f_0 - \mathcal{T}_{K}(f - f_0)\}$. Let $\delta_N$ satisfy the critical inequality $E_K[\widehat{\mathcal{R}}_N(\text{star}(\widehat{\mathcal{G}}), \delta) \mid \mathcal{D}_{n,v}] \leq \sqrt{N} \delta^2$. Let $\delta_n^2 := \max_{k \in [K]} \delta_{n_k}^2$, where each $\delta_{n_k}$ satisfies the critical inequality $E_K[\widehat{\mathcal{R}}_{n_k}(\text{star}(\mathcal{F}), \delta)] \leq \delta^2\sqrt{n_k}$.  Assume that $ N^{-1/2} \sqrt{\log \log \delta_{N}} =  o(\delta_{N})$ and $ n_k^{-1/2} \sqrt{\log \log \delta_{n_k}} =  o(\delta_{n_k})$.  Then, for any $\delta > 0$, with probability at least $1 -  \delta - e^{ -N\delta_N^2} $, the following holds uniformly for any $f \in \mathcal{F}$:
\begin{align*}
 \langle  \widehat{\mathcal{T}}_{K,0}(f), \widehat{\mathcal{T}}_{K,1}(f - f_0) - \mathcal{T}_{K}(f - f_0) \rangle_{L^2_K(Z)} &\lesssim \delta_N \left( \delta_n +  \sqrt{\log(K/\delta) / n_{\min}} \right) \|f - f_0\|_{L^2_K}  \\
 & \quad +   \delta_N \left( \delta_n^2 +  \log(K/\delta) / n_{\min} \right)    + \delta_N^2.
\end{align*} 
 \end{lemma}

 \begin{proof}[Proof of Lemma \ref{lemma::localmax_smalln_1}]

We can express \(\langle  \widehat{\mathcal{T}}_{K,0}(f) - \mathcal{T}_{K}(f), \widehat{\mathcal{T}}_{K,1}(f) - \mathcal{T}_{K}(f) \rangle_{L^2_K(Z)}\) as the empirical mean:
\begin{align*}
  \frac{1}{N} \sum_{k=1}^K \sum_{i=1}^{n_k} 1(V_{ki} = 1) \left\{\widehat{\mathcal{T}}_{K,0}(f) - \mathcal{T}_{K}(f)\right\} \left\{f(X_{ki}) - \mathcal{T}_{K}(f)(k) \right\}.
\end{align*}
Moreover, from the independence of the data folds, the right-hand side is mean zero conditional on the data fold $\mathcal{D}_{n,0}$:
\begin{align*}
\sum_{k=1}^K \frac{n_k}{N} E_K\Big[ \big\{\widehat{\mathcal{T}}_{K,0}(f)(k) &- \mathcal{T}_{K}(f)\big\}  \big\{f(X_{k}) - \mathcal{T}_{K}(f)(k) \big\} \mid \mathcal{D}_{n,0} \Big] \\
&= \sum_{k=1}^K \frac{n_k}{N} \big\{\widehat{\mathcal{T}}_{K,0}(f)(k) - \mathcal{T}_{K}(f)(k)\big\} E_K\big[f(X_{k}) - \mathcal{T}_{K}(f)(k) \big] = 0. 
\end{align*}

Define the function class $\widehat{\mathcal{G}} := \{\widehat{g}_f : f \in \mathcal{F}\}$ as the set of all elements of the form $ \widehat{g}_f :=[\widehat{\mathcal{T}}_{K,0} - \mathcal{T}_{K}](f) \{f - \mathcal{T}_{K}(f)\}$. Note that $ \big\{\widehat{\mathcal{T}}_{K,0}(f) - \mathcal{T}_{K}(f)\big\}  \big\{f - \mathcal{T}_{K}(f)\big\}$ is an element of $\widehat{\mathcal{G}} $. Thus, applying Lemma \ref{lemma:loc_max_ineq}, conditionally on the training fold $\mathcal{D}_{n,0}$, we find, with probability at least $1 - e^{-N\delta_N^2}$, the following holds uniformly for any $f \in \mathcal{F}$:
\begin{align*}
 \langle  \widehat{\mathcal{T}}_{K,0}(f) - \mathcal{T}_{K}(f), \widehat{\mathcal{T}}_{K,1}(f) - \mathcal{T}_{K}(f) \rangle_{L^2_K(Z)} \lesssim \delta_N \|\widehat{g}_f \|_{L^2_K} + \delta_N^2.
\end{align*}
where $\delta_N$ satisfies $E_K[\widehat{\mathcal{R}}_N(\widehat{\mathcal{G}}, \delta) \mid \mathcal{D}_{n,0}] \leq \sqrt{N} \delta^2$.

Fix $\delta > 0$  and define $\delta_n^2 := \max_{k \in [K]} \delta_{n_k}^2$, where each $\delta_{n_k}$ satisfies the critical inequality $E_K[\widehat{\mathcal{R}}_{n_k}(\mathcal{F}, \delta)] \leq \delta^2\sqrt{n_k}$. Then, by Lemma \ref{lemma::localmax1}, there exists a constant $C \geq 0$ such that, with probability at least $1 -\delta$, the following holds uniformly over all $f \in \mathcal{F}$:
\begin{align*}
  \|\widehat{g}_f \|_{L^2_K}^2 &= \|[\widehat{\mathcal{T}}_{K,0} - \mathcal{T}_{K}](f) \{f - \mathcal{T}_{K}(f)\}\|_{L^2_K}^2 \\
  &= \sum_{k=1}^N \frac{n_k}{N} \{[\widehat{\mathcal{T}}_{K,0} - \mathcal{T}_{K}](f)(k)\}^2 \|f - \mathcal{T}_{K}(f)\|_k^2\\
  &\leq \sum_{k=1}^N \frac{n_k}{N} \left\{ \delta_{n}  \|f - \mathcal{T}_{K}(f)\|_k + \delta_{n}^2 + \frac{\log(K/\delta)}{n_{\min}}\right\}^2 \|f - \mathcal{T}_{K}(f)\|_k^2\\
  &\lesssim \delta_{n}^2  \left[ \sum_{k=1}^N \frac{n_k}{N}  \|f - \mathcal{T}_{K}(f)\|_k^4\right] + \delta_{n}^4 \|f - \mathcal{T}_{K}(f)\|_{L^2_K}^2 + \frac{\log(K/\delta)}{n_{\min}}  \|f - \mathcal{T}_{K}(f)\|_{L^2_K}^2\\
    & \lesssim \left[\delta_n^2  + \delta_{n}^4 + \frac{\log(K/\delta)}{n_{\min}} \right] \|f - \mathcal{T}_{K}(f)\|_{L^2_K}^2.
\end{align*}
Thus, we have $\|\widehat{g}_f \|_{L^2_K} \lesssim  [\delta_n  + \sqrt{\log(K/\delta) / n_{\min}} ] \|f - \mathcal{T}_{K}(f)\|_{L^2_K}$.

 Putting it all together, we find, with probability at least $1 - \delta -  e^{-N\delta_N^2} $, the following holds uniformly for any $f \in \mathcal{F}$:
\begin{align*}
 \langle  \widehat{\mathcal{T}}_{K,0}(f) - \mathcal{T}_{K}(f), \widehat{\mathcal{T}}_{K,1}(f) - \mathcal{T}_{K}(f) \rangle_{L^2_K(Z)} &\lesssim   \delta_N\left( \delta_n +  \sqrt{\log(K/\delta) / n_{\min}} \right)  \|f - \mathcal{T}_{K}(f)\|_{L^2_K}   + \delta_N^2.
\end{align*}
 \end{proof}

 \begin{proof}[Proof of Lemma \ref{lemma::localmax_smalln_2}]
     The proof of this lemma is nearly identical to the proof of Lemma \ref{lemma::localmax_smalln_1}. By independence of the folds, we can express \(\langle  \widehat{\mathcal{T}}_{K,0}(f) , \widehat{\mathcal{T}}_{K,1}(f - f_0) - \mathcal{T}_{K}(f - f_0) \rangle_{L^2_K(Z)}\) as the empirical mean of mean zero-random variables:
\begin{align*}
  \frac{1}{N} \sum_{k=1}^K \sum_{i=1}^{n_k} 1(V_{ki} = 0)f(X_{ki}) \left\{\widehat{\mathcal{T}}_{K,1}(f - f_0)(k) - \mathcal{T}_{K}(f - f_0)(k)\right\}.
\end{align*}
Applying Lemma \ref{lemma:loc_max_ineq}, conditionally on the training fold $\mathcal{D}_{n,1}$, we find, with probability at least $1 - e^{-N\delta_N^2}$, the following holds uniformly for any $f \in \mathcal{F}$:
\begin{align*}
 \langle  \widehat{\mathcal{T}}_{K,0}(f), \widehat{\mathcal{T}}_{K,1}(f - f_0) - \mathcal{T}_{K}(f - f_0) \rangle_{L^2_K(Z)} \lesssim \delta_N \|f\{ \widehat{\mathcal{T}}_{K,1}(f - f_0) - \mathcal{T}_{K}(f - f_0)\}\|_{L^2_K} + \delta_N^2.
\end{align*}
To complete the proof, we will now show that $ \|f\{ \widehat{\mathcal{T}}_{K,1}(f - f_0) - \mathcal{T}_{K}(f - f_0)\}\|_{L^2_K}^2  \lesssim \left( \delta_n +  \sqrt{\log(K/\delta) / n_{\min}} \right) \|f - f_0\|_{L^2_K(X)} + \left( \delta_n^2 +  \log(K/\delta) / n_{\min} \right)$ with high probability. By Lemma \ref{lemma::localmax1}, there exists a constant $C \geq 0$ such that, with probability at least $1 -  \delta$, the following holds uniformly over all $f \in \mathcal{F}$:  
     \begin{align*}
 \|f\{ \widehat{\mathcal{T}}_{K,1}(f - f_0) - \mathcal{T}_{K}(f - f_0)\}\|_{L^2_K}^2  &\lesssim \sum_{k=1}^N \frac{n_k}{N} \{[\widehat{\mathcal{T}}_{K,1} - \mathcal{T}_{K}](f - f_0)(k)\}^2 \\
  &\leq \sum_{k=1}^N \frac{n_k}{N} \left\{\left( \delta_n +  \sqrt{\log(K/\delta) / n_{\min}} \right)  \|f - f_0 - \mathcal{T}_{K}(f -f_0)\|_k + \left( \delta_n^2 +  \log(K/\delta) / n_{\min} \right)\right\}^2  \\
  &\lesssim \left( \delta_n^2 +  \log(K/\delta) / n_{\min} \right)  \left[ \sum_{k=1}^N \frac{n_k}{N}  \|f -  f_0\|_k^2\right] +\left( \delta_n +  \sqrt{\log(K/\delta) / n_{\min}} \right)^4\\
   &\lesssim \left( \delta_n^2 +  \log(K/\delta) / n_{\min} \right)  \|f -  f_0\|_{L^2_K(X)}^2  + \left( \delta_n +  \sqrt{\log(K/\delta) / n_{\min}} \right)^4.
\end{align*}

 \end{proof}

\subsection{Error bounds for dual solution}

In the following, we define the function classes: $\mathcal{G} := \{ \mathcal{T}_{K}(\beta)(\beta - \mathcal{T}_{K}(\beta)) : \beta \in \mathcal{F}_{\mathrm{dual}} \} \cup \{ \beta \cdot \mathcal{T}_{K}(\beta - \beta_\lambda) : \beta \in \mathcal{F}_{\mathrm{dual}} \} \cup \mathcal{F}_{\mathrm{dual}}$  and  $\widehat{G}_v := \{ (\widehat{\mathcal{T}}_{K,v} - \mathcal{T}_{K})(\beta)(\beta - \mathcal{T}_{K}(\beta)) : \beta \in \mathcal{F}_{\mathrm{dual}} \} \cup \{ \widehat{\mathcal{T}}_{K,v}(\beta)(\beta - \beta_\lambda - \mathcal{T}_{K}(\beta - \beta_\lambda)) : \beta \in \mathcal{F}_{\mathrm{dual}} \}$.  

\begin{theorem}[Estimation error] 
\label{theorem::estimationerror} 
Assume \ref{cond::funclasscontainsdual}. Let $\delta_N$ be deterministic and almost surely satisfy the critical inequalities $E_K[\widehat{\mathcal{R}}_N(\text{star}(\mathcal{G}), \delta)]\leq \delta^2$ and $E_K[\widehat{\mathcal{R}}_N(\text{star}(\widehat{\mathcal{G}}_v), \delta) \mid \mathcal{D}_{n,v}] \leq \delta^2$ for all $v \in \{0,1\}$. Let $\delta_n^2 := \max_{k \in [K]} \delta_{n_k}^2$, where each $\delta_{n_k}$ satisfies the critical inequality $E_K[\widehat{\mathcal{R}}_{n_k}(\text{star}(\mathcal{F}_{\mathrm{dual}}), \delta)] \leq \delta^2 $.  Assume that $ N^{-1/2} \sqrt{\log \log \delta_{N}} =  o(\delta_{N})$, $N\delta_N^2\rightarrow \infty$, and $ n_k^{-1/2} \sqrt{\log \log \delta_{n_k}} =  o(\delta_{n_k})$. Then, it holds that
\begin{align*}
\|\mathcal{T}_{K}(\widehat{\beta}_K(\lambda)-\beta_{K}(\lambda)\|_{L^2_K(Z)}& = O_p\left( \delta_N \left( 1 + \frac{\delta_n + \sqrt{\log K / n_{\min}}}{\sqrt{\lambda}}  + \sqrt{\lambda} \right)\right)\\
\|\widehat{\beta}_K(\lambda)-\beta_{K}(\lambda)\|_{L^2_K(X)} &= O_p\left( \lambda^{-1/2} \delta_N \left( 1 + \frac{\delta_n + \sqrt{\log K / n_{\min}}}{\sqrt{\lambda}}  + \sqrt{\lambda} \right)\right).
\end{align*} 
\end{theorem}

 \begin{proof}[Proof of Theorem \ref{theorem::estimationerror}]

For ease of notation, we denote $R_K(\beta) :=  R_K^*(\beta)$. Denote the regularized population and empirical risks by $R_{\lambda}(\beta) := R_K(\beta)  + \lambda \|\beta\|^2_{L^2_K(X)}$ and $\widehat{R}_{\lambda}(\beta) := \widehat{R}_{K}(\beta)  + \lambda \|\beta\|^2_{N}$, respectively. As shorthand, let $\beta_{K,\lambda} := \beta_{K}(\lambda) $ and $\widehat{\beta}_{K,\lambda} := \widehat{\beta}_K(\lambda)$ denote the population and empirical risk minimizers, and denote $\mathcal{F} := \mathcal{F}_{\mathrm{dual}}$.  By \ref{cond::funclasscontainsdual}, we have that $\beta_{K}(\lambda) = (\mathcal{T}_{K}^* \mathcal{T}_{K} + \lambda I)^{-1}\,\Pi_{K} \alpha_K$ is contained in $\mathcal{F}$.

We can lower bound the regret by:
\begin{align*}
   R_{\lambda}(\widehat{\beta}_{K,\lambda}) - R_{\lambda}(\beta_{K,\lambda}) &\geq \frac{d}{dt}  R_{\lambda}(\beta_{K,\lambda} + t(\widehat{\beta}_{K,\lambda}- \beta_{K,\lambda}))\big|_{t=0} + \frac{d^2}{dt^2}  R_{\lambda}(\beta_{K,\lambda} + t(\widehat{\beta}_{K,\lambda}- \beta_{K,\lambda}))\big|_{t=0} \\
    &\geq \frac{d^2}{dt^2}  R_{\lambda}(\beta_{K,\lambda} + t(\widehat{\beta}_{K,\lambda}- \beta_{K,\lambda}))\big|_{t=0},
\end{align*}
where we used the first-order optimality condition that $ \frac{d}{dt}  R_{\lambda}(\beta_{K,\lambda} + t(\widehat{\beta}_{K,\lambda}- \beta_{K,\lambda}))\big|_{t=0} \geq 0 $, since $\mathcal{F}$ is convex. Computing the second derivative, we obtain
 \begin{align}
\|\mathcal{T}_{K}(\widehat{\beta}_{K,\lambda}- \beta_{K,\lambda})\|_{L^2_K(Z)}^2 + \lambda \| \widehat{\beta}_{K,\lambda}- \beta_{K,\lambda}\|_{L^2_K(X)}^2 \lesssim   R_{\lambda}(\widehat{\beta}_{K,\lambda}) - R_{\lambda}(\beta_{K,\lambda}).  \label{eqn::lowerriskbound}
\end{align}

Now, using that $\widehat{\beta}_{K,\lambda}$ is the minimizer of the regularized empirical risk, so that $\widehat{R}_{\lambda}(\widehat{\beta}_{K,\lambda}) - \widehat{R}_{\lambda}(\beta_{K,\lambda}) \leq 0$, we can further write:
\begin{align*}
     R_{\lambda}(\widehat{\beta}_{K,\lambda}) - R_{\lambda}(\beta_{K,\lambda}) &\leq  R_{\lambda}(\widehat{\beta}_{K,\lambda}) - R_{\lambda}(\beta_{K,\lambda}) - \left\{ \widehat{R}_{\lambda}(\widehat{\beta}_{K,\lambda}) - \widehat{R}_{\lambda}(\beta_{K,\lambda}) \right\}\\
      &\leq  R_{K}(\widehat{\beta}_{K,\lambda}) - R_{K}(\beta_{K,\lambda}) - \left\{ \widehat{R}_{K}(\widehat{\beta}_{K,\lambda}) - \widehat{R}_{K}(\beta_{K,\lambda}) \right\} \\
      & \quad  + \lambda \left[\|\widehat{\beta}_{K,\lambda}\|^2_N - \|\beta_{K,\lambda}\|^2_N - \left\{\|\widehat{\beta}_{K,\lambda}\|^2_{L^2_K(X)} - \|\beta_{K,\lambda}\|^2_{L^2_K(X)}\right\} \right]
\end{align*}

We now establish concentration of each of the two terms on the right-hand side of the above display. 

It was shown in \cite{bibaut2024nonparametric} (using Lemma \ref{lemma:loc_max_ineq}) that, with probability tending to one:
\begin{align}
    \label{eqn::regtermconc}\lambda \left[\|\widehat{\beta}_{K,\lambda}\|^2_N - \|\beta_{K,\lambda}\|^2_N - \left\{\|\widehat{\beta}_{K,\lambda}\|^2_{L^2_K(X)} - \|\beta_{K,\lambda}\|^2_{L^2_K(X)}\right\} \right] \lesssim \lambda \delta_N \|\widehat{\beta}_{K,\lambda} - \beta_{K,\lambda}\|_{L^2_K(X)} + \lambda \delta_N^2.  
\end{align}

Next, we establish concentration of $ R_{\lambda}(\widehat{\beta}_{K,\lambda}) - R_{\lambda}(\beta_{K,\lambda}) - \left\{ \widehat{R}_{\lambda}(\widehat{\beta}_{K,\lambda}) - \widehat{R}_{\lambda}(\beta_{K,\lambda}) \right\}$. We use the following decomposition of the empirical regret:
\begin{align*}
&  \widehat{R}_K(\widehat{\beta}_{K,\lambda}) - \widehat{R}_K(\beta_{K,\lambda})\\
  &=   \langle  \widehat{\mathcal{T}}_{K,0}(\widehat{\beta}_{K,\lambda}),   \widehat{\mathcal{T}}_{K,1}(\widehat{\beta}_{K,\lambda}) \rangle_{L^2_K(Z)}   -   \langle  \widehat{\mathcal{T}}_{K,0}(\beta_{K,\lambda}), \widehat{\mathcal{T}}_{K,1}(\beta_{K,\lambda}) \rangle_{L^2_K(Z)} - 2\psi(\widehat{\beta}_{K,\lambda} - \beta_{K,\lambda})\\
     &= \langle  \widehat{\mathcal{T}}_{K,0}(\widehat{\beta}_{K,\lambda}), \widehat{\mathcal{T}}_{K,1}(\widehat{\beta}_{K,\lambda} - \beta_{K,\lambda}) \rangle_{L^2_K(Z)}  + \langle  \widehat{\mathcal{T}}_{K,0}(\widehat{\beta}_{K,\lambda} - \beta_{K,\lambda}), \widehat{\mathcal{T}}_{K,1}(\beta_{K,\lambda}) \rangle_{L^2_K(Z)}  - 2\psi(\widehat{\beta}_{K,\lambda} - \beta_{K,\lambda})\\
     &  = \langle  \widehat{\mathcal{T}}_{K,0}(\widehat{\beta}_{K,\lambda} - \beta_{K,\lambda}), \widehat{\mathcal{T}}_{K,1}(\widehat{\beta}_{K,\lambda} - \beta_{K,\lambda}) \rangle_{L^2_K(Z)} \\
     & \quad + \langle  \widehat{\mathcal{T}}_{K,0}(\beta_{K,\lambda}), \widehat{\mathcal{T}}_{K,1}(\widehat{\beta}_{K,\lambda} - \beta_{K,\lambda}) \rangle_{L^2_K(Z)}  + \langle  \widehat{\mathcal{T}}_{K,0}(\widehat{\beta}_{K,\lambda} - \beta_{K,\lambda}), \widehat{\mathcal{T}}_{K,1}(\beta_{K,\lambda}) \rangle_{L^2_K(Z)}  - 2\psi(\widehat{\beta}_{K,\lambda} - \beta_{K,\lambda}).
\end{align*}
Therefore, leveraging an analogous decomposition for the true regret, we have:
\begin{align*}
    R_{K}(\widehat{\beta}_{K,\lambda}) - R_{K}(\beta_{K,\lambda}) &- \left\{ \widehat{R}_{K}(\widehat{\beta}_{K,\lambda}) - \widehat{R}_{K}(\beta_{K,\lambda}) \right\} \\
   & \leq    \langle  \mathcal{T}_{K}(\widehat{\beta}_{K,\lambda} - \beta_{K,\lambda}),  \mathcal{T}_{K}(\widehat{\beta}_{K,\lambda} - \beta_{K,\lambda}) \rangle_{L^2_K(Z)} - \langle  \widehat{\mathcal{T}}_{K,0}(\widehat{\beta}_{K,\lambda} - \beta_{K,\lambda}), \widehat{\mathcal{T}}_{K,1}(\widehat{\beta}_{K,\lambda} - \beta_{K,\lambda}) \rangle_{L^2_K(Z)}  \\
    & \quad + \sup_{v \in \{0,1\}} \left|\langle  \mathcal{T}_{K}(\beta_{K,\lambda}), \mathcal{T}_{K}(\widehat{\beta}_{K,\lambda} - \beta_{K,\lambda}) \rangle_{L^2_K(Z)}  -  \langle  \widehat{\mathcal{T}}_{K,v}(\beta_{K,\lambda}), \widehat{\mathcal{T}}_{K,1-v}(\widehat{\beta}_{K,\lambda} - \beta_{K,\lambda}) \rangle_{L^2_K(Z)} \right|.
\end{align*}

We begin by establishing concentration of the second term on the right-hand side of the previous display. For each $v \in \{0,1\}$, we can write
\begin{align*}
  \langle  \widehat{\mathcal{T}}_{K,v}(\beta_{K,\lambda}), \widehat{\mathcal{T}}_{K,1-v}(\widehat{\beta}_{K,\lambda} - \beta_{K,\lambda}) \rangle_{L^2_K(Z)} &-     \langle  \mathcal{T}_{K}(\beta_{K,\lambda}), \mathcal{T}_{K}(\widehat{\beta}_{K,\lambda}    -  \beta_{K,\lambda}) \rangle_{L^2_K(Z)}  \\
    &=\langle  \widehat{\mathcal{T}}_{K,v}(\widehat{\beta}_{K,\lambda}) - \mathcal{T}_{K}(\widehat{\beta}_{K,\lambda}), \mathcal{T}_{K}(\widehat{\beta}_{K,\lambda} - \beta_{K,\lambda}) \rangle_{L^2_K(Z)} \\
     & \quad +  \langle  \widehat{\mathcal{T}}_{K,v}(\widehat{\beta}_{K,\lambda}), \widehat{\mathcal{T}}_{K,1-v}(\widehat{\beta}_{K,\lambda} - \beta_{K,\lambda}) - \mathcal{T}_{K}(\widehat{\beta}_{K,\lambda} - \beta_{K,\lambda})\rangle_{L^2_K(Z)}.
\end{align*}
Applying Lemma \ref{lemma::localmax3}, we have, with probability tending to one, that
\begin{align*}
       \langle  \widehat{\mathcal{T}}_{K,v}(\widehat{\beta}_{K,\lambda}) - \mathcal{T}_{K}(\widehat{\beta}_{K,\lambda}), \mathcal{T}_{K}(\widehat{\beta}_{K,\lambda} - \beta_{K,\lambda}) \rangle_{L^2_K(Z)}  \lesssim \delta_N   \|\mathcal{T}_{K}(\widehat{\beta}_{K,\lambda} - \beta_{K,\lambda})\|_{L^2_K(Z)} + \delta_N^2 .
\end{align*} 
Next, applying Lemma \ref{lemma::localmax_smalln_2}, we have, for any $\delta > 0$, with probability tending to $1- \delta$, that
\begin{align*}
     \langle  \widehat{\mathcal{T}}_{K,v}(\widehat{\beta}_{K,\lambda}), \widehat{\mathcal{T}}_{K,1-v}(\widehat{\beta}_{K,\lambda} - \beta_{K,\lambda}) - \mathcal{T}_{K}(\widehat{\beta}_{K,\lambda} - \beta_{K,\lambda})\rangle_{L^2_K(Z)}  \lesssim \delta_N \widetilde{\delta}_n \|\widehat{\beta}_{K,\lambda} - \beta_{K,\lambda}\|_{L^2_K}   +   \delta_N \widetilde{\delta}_n^2    + \delta_N^2,
\end{align*} 
where $\widetilde{\delta}_n := \delta_n  + \sqrt{n_{\min}^{-1} \log(K/\delta)}$.
Thus, we conclude, with probability tending to $1 - \delta$, that
\begin{align*}
    \sup_{v \in \{0,1\}} \Big| \langle  &\mathcal{T}_{K}(\beta_{K,\lambda}), \mathcal{T}_{K}(\widehat{\beta}_{K,\lambda} - \beta_{K,\lambda}) \rangle_{L^2_K(Z)}  
    -  \langle  \widehat{\mathcal{T}}_{K,v}(\beta_{K,\lambda}), \widehat{\mathcal{T}}_{K,1-v}(\widehat{\beta}_{K,\lambda} - \beta_{K,\lambda}) \rangle_{L^2_K(Z)} \Big|\\
    &\lesssim \delta_N \|\mathcal{T}_{K}(\widehat{\beta}_{K,\lambda} - \beta_{K,\lambda})\|_{L^2_K(Z)} + \delta_N^2 + \delta_N \widetilde{\delta}_n \|\widehat{\beta}_{K,\lambda} - \beta_{K,\lambda}\|_{L^2_K(X)}  
    + \delta_N \widetilde{\delta}_n^2 + \delta_N^2.
\end{align*}

Next, we establish concentration of the first term on the right-hand side.  We have the expansion:
\begin{align*}
    \langle  \widehat{\mathcal{T}}_{K,0}(\widehat{\beta}_{K,\lambda} - \beta_{K,\lambda})&, \widehat{\mathcal{T}}_{K,1}(\widehat{\beta}_{K,\lambda} - \beta_{K,\lambda}) \rangle_{L^2_K(Z)} -  \langle  \mathcal{T}_{K}(\widehat{\beta}_{K,\lambda} - \beta_{K,\lambda}), \mathcal{T}_{K}(\widehat{\beta}_{K,\lambda} - \beta_{K,\lambda}) \rangle_{L^2_K(Z)}\\
    & =   \langle  \widehat{\mathcal{T}}_{K,0}(\widehat{\beta}_{K,\lambda} - \beta_{K,\lambda}) - \mathcal{T}_{K}(\widehat{\beta}_{K,\lambda} - \beta_{K,\lambda}), \widehat{\mathcal{T}}_{K,1}(\widehat{\beta}_{K,\lambda} - \beta_{K,\lambda}) - \mathcal{T}_{K}(\widehat{\beta}_{K,\lambda} - \beta_{K,\lambda})\rangle_{L^2_K(Z)} \\
    & \quad + \langle  \mathcal{T}_{K}(\widehat{\beta}_{K,\lambda} - \beta_{K,\lambda}), \widehat{\mathcal{T}}_{K,0}(\widehat{\beta}_{K,\lambda} - \beta_{K,\lambda}) - \mathcal{T}_{K}(\widehat{\beta}_{K,\lambda} - \beta_{K,\lambda})\rangle_{L^2_K(Z)}\\
    & \quad +  \langle  \mathcal{T}_{K}(\widehat{\beta}_{K,\lambda} - \beta_{K,\lambda}), \widehat{\mathcal{T}}_{K,1}(\widehat{\beta}_{K,\lambda} - \beta_{K,\lambda}) - \mathcal{T}_{K}(\widehat{\beta}_{K,\lambda} - \beta_{K,\lambda})\rangle_{L^2_K(Z)}.
\end{align*}
We consider each of the terms on the right-hand side of the above display in turn.
 
To bound the first term, we apply Lemma \ref{lemma::localmax_smalln_1} to find, with probability tending to $1 - \delta$, that
\begin{align*}
     \langle  \widehat{\mathcal{T}}_{K,0}(\widehat{\beta}_{K,\lambda} - \beta_{K,\lambda}) - \mathcal{T}_{K}(\widehat{\beta}_{K,\lambda} - \beta_{K,\lambda}), \widehat{\mathcal{T}}_{K,1}(\widehat{\beta}_{K,\lambda} - \beta_{K,\lambda}) - \mathcal{T}_{K}(\widehat{\beta}_{K,\lambda} - \beta_{K,\lambda})\rangle_{L^2_K(Z)} \\
     \lesssim \delta_N \widetilde{\delta}_n  \|\widehat{\beta}_{K,\lambda} - \beta_{K,\lambda}\|_{L^2_K(X)}   +   \delta_N \widetilde{\delta}_n^2  \|\widehat{\beta}_{K,\lambda} - \beta_{K,\lambda}\|_{L^2_K(X)}   + \delta_N^2 .
\end{align*}
To bound the final two terms, we applying Lemma \ref{lemma::localmax2} to find for each $v \in \{0,1\}$, with probability tending to one, that
\begin{align*}
     \langle  \mathcal{T}_{K}(\widehat{\beta}_{K,\lambda} - \beta_{K,\lambda}), \widehat{\mathcal{T}}_{K,v}(\widehat{\beta}_{K,\lambda} - \beta_{K,\lambda}) - \mathcal{T}_{K}(\widehat{\beta}_{K,\lambda} - \beta_{K,\lambda})\rangle_{L^2_K(Z)} \\
     \lesssim    \delta_N   \|\mathcal{T}_{K}(\widehat{\beta}_{K,\lambda} - \beta_{K,\lambda})\|_{L^2_K(Z)} + \delta_N^2 .
\end{align*}
Combining both bounds, we conclude, with probability tending to $1 - \delta$, that
 \begin{align*}
    \langle  \widehat{\mathcal{T}}_{K,0}(\widehat{\beta}_{K,\lambda} - \beta_{K,\lambda})&, \widehat{\mathcal{T}}_{K,1}(\widehat{\beta}_{K,\lambda} - \beta_{K,\lambda}) \rangle_{L^2_K(Z)} -  \langle  \mathcal{T}_{K}(\widehat{\beta}_{K,\lambda} - \beta_{K,\lambda}), \mathcal{T}_{K}(\widehat{\beta}_{K,\lambda} - \beta_{K,\lambda}) \rangle_{L^2_K(Z)}\\
    & =  \delta_N \widetilde{\delta}_n  \|\widehat{\beta}_{K,\lambda} - \beta_{K,\lambda}\|_{L^2_K(X)}   +   \delta_N \widetilde{\delta}_n^2  \|\widehat{\beta}_{K,\lambda} - \beta_{K,\lambda}\|_{L^2_K(X)}   + \delta_N^2 \\
    & \quad +   \delta_N   \|\mathcal{T}_{K}(\widehat{\beta}_{K,\lambda} - \beta_{K,\lambda})\|_{L^2_K(Z)} + \delta_N^2.
\end{align*}

Substituting all our bounds, we find, with probability tending to $1 - \delta$, that
\begin{align*}
    R_{K}(\widehat{\beta}_{K,\lambda}) - R_{K}(\beta_{K,\lambda}) &- \left\{ \widehat{R}_{K}(\widehat{\beta}_{K,\lambda}) - \widehat{R}_{K}(\beta_{K,\lambda}) \right\} \\
   & \leq   \delta_N \widetilde{\delta}_n  \|\widehat{\beta}_{K,\lambda} - \beta_{K,\lambda}\|_{L^2_K(X)}   +   \delta_N \widetilde{\delta}_n^2  \|\widehat{\beta}_{K,\lambda} - \beta_{K,\lambda}\|_{L^2_K(X)}   + \delta_N^2 \\
    & \quad +   \delta_N   \|\mathcal{T}_{K}(\widehat{\beta}_{K,\lambda} - \beta_{K,\lambda})\|_{L^2_K(Z)} + \delta_N^2\\
   & \quad +  \delta_N \widetilde{\delta}_n \|\widehat{\beta}_{K,\lambda} - \beta_{K,\lambda}\|_{L^2_K(X)}   +   \delta_N \widetilde{\delta}_n^2  \|\widehat{\beta}_{K,\lambda} - \beta_{K,\lambda}\|_{L^2_K(X)}   + \delta_N^2 \\
    & \quad +   \delta_N   \|\mathcal{T}_{K}(\widehat{\beta}_{K,\lambda} - \beta_{K,\lambda})\|_{L^2_K(Z)} + \delta_N^2.
\end{align*}
Combining like terms, we have, with probability tending to $1 - \delta$, that
\begin{align*}
    R_{K}(\widehat{\beta}_{K,\lambda}) - R_{K}(\beta_{K,\lambda}) &- \left\{ \widehat{R}_{K}(\widehat{\beta}_{K,\lambda}) - \widehat{R}_{K}(\beta_{K,\lambda}) \right\} \\
   & \leq   \delta_N \widetilde{\delta}_n  \|\widehat{\beta}_{K,\lambda} - \beta_{K,\lambda}\|_{L^2_K(X)} +   \delta_N   \|\mathcal{T}_{K}(\widehat{\beta}_{K,\lambda} - \beta_{K,\lambda})\|_{L^2_K(Z)}    + \delta_N^2.
\end{align*}

Therefore, from \eqref{eqn::lowerriskbound} and \eqref{eqn::regtermconc}, we have, with probability tending to $1 - \delta$, that
 \begin{align*}
& \|\mathcal{T}_{K}(\widehat{\beta}_{K,\lambda}- \beta_{K,\lambda})\|_{L^2_K(Z)}^2 + \lambda \| \widehat{\beta}_{K,\lambda}- \beta_{K,\lambda}\|_{L^2_K(X)}^2\\ &\lesssim   R_{\lambda}(\widehat{\beta}_{K,\lambda}) - R_{\lambda}(\beta_{K,\lambda})\\
 &\lesssim  \delta_N \widetilde{\delta}_n  \|\widehat{\beta}_{K,\lambda} - \beta_{K,\lambda}\|_{L^2_K(X)} +   \delta_N   \|\mathcal{T}_{K}(\widehat{\beta}_{K,\lambda} - \beta_{K,\lambda})\|_{L^2_K(Z)}    + \delta_N^2 \\
 & \quad + \lambda\left\{\delta_N\|\widehat{\beta}_{K,\lambda} - \beta_{K,\lambda}\|_{L^2_K(X)} + \delta_N^2 \right\}\\
  &\lesssim   \lambda \delta_N  \left\{ 1 + \frac{\widetilde{\delta}_n}{\lambda} \right\} \|\widehat{\beta}_{K,\lambda} - \beta_{K,\lambda}\|_{L^2_K(X)} +   \delta_N   \|\mathcal{T}_{K}(\widehat{\beta}_{K,\lambda} - \beta_{K,\lambda})\|_{L^2_K(Z)}    + \delta_N^2 .
\end{align*}

We now obtain weak and strong norm rates from the above inequality. Define $A := \|\mathcal{T}_{K}(\widehat{\beta}_{K,\lambda}-\beta_{K,\lambda})\|_{L^2_K(Z)}$ and $B := \|\widehat{\beta}_{K,\lambda}-\beta_{K,\lambda}\|_{L^2_K(X)}.$ Then, 
\[
A^2 + \lambda B^2 \lesssim \lambda\delta_N\Bigl(1+\frac{\widetilde{\delta}_n}{\lambda}\Bigr)B + \delta_N A + \delta_N^2.
\]
Applying Young’s inequality,
\[
\delta_N A \le \frac{1}{2}A^2 + \frac{1}{2}\delta_N^2, \quad
\lambda\delta_N\Bigl(1+\frac{\widetilde{\delta}_n}{\lambda}\Bigr)B \le \frac{\lambda}{2}B^2 + \frac{\lambda\delta_N^2}{2}\Bigl(1+\frac{\widetilde{\delta}_n}{\lambda}\Bigr)^2.
\]
Substituting, we get
\begin{align*}
    A^2 + \lambda B^2 &\lesssim \frac{1}{2}A^2 + \frac{\lambda}{2}B^2 
+ \frac{\lambda\delta_N^2}{2}\Bigl(1+\frac{\widetilde{\delta}_n}{\lambda}\Bigr)^2 + \frac{3}{2}\delta_N^2\\
& \lesssim \frac{1}{2} A^2 + \frac{\lambda}{2} B^2 + \delta_N^2 \Bigl(\frac{(\lambda+\widetilde{\delta}_n)^2}{2\lambda} + \frac{3}{2}\Bigr).
\end{align*}
Rearranging and using that $(\lambda+\widetilde{\delta}_n)^2 \lesssim \lambda^2 + \widetilde{\delta}_n^2$ :
\begin{align*}
    \frac{1}{2} A^2 + \frac{\lambda}{2} B^2 & \lesssim \delta_N^2 \Bigl(\frac{\widetilde{\delta}_n^2}{2\lambda} + 1\Bigr) +  \lambda \delta_N^2\\
     & \lesssim  \delta_N^2 \Bigl(1 + \frac{\widetilde{\delta}_n^2}{2\lambda} + \lambda\Bigr).
\end{align*}
It follows that, with probability tending to $1 - \delta$, that
\begin{align*}
\|\mathcal{T}_{K}(\widehat{\beta}_{K,\lambda}-\beta_{K,\lambda})\|_{L^2_K(Z)}& = O_p\left(\delta_N \left( 1 + \frac{\widetilde{\delta}_n}{\sqrt{\lambda}}  + \sqrt{\lambda} \right) \right)\\
\|\widehat{\beta}_{K,\lambda}-\beta_{K,\lambda}\|_{L^2_K(X)} &= O_p\left(\lambda^{-1/2} \delta_N \left( 1 + \frac{\widetilde{\delta}_n}{\sqrt{\lambda}}  + \sqrt{\lambda} \right)\right).
\end{align*} 
Since \(\delta > 0\) is arbitrary and \(\widetilde{\delta}_n = \sqrt{\log(K/\delta)\,n_{\min}^{-1}}\), we conclude that the above bound holds unconditionally when \(\widetilde{\delta}_n\) is replaced by \(\sqrt{\log K / n_{\min}}\) in the \(O_p\) terms.

 \end{proof}

\begin{theorem}[Weak and strong norm rates]
    Suppose that the results of Lemma \ref{lemma::approxstrongident} and Theorem \ref{theorem::estimationerror} hold. Denote $\widetilde{\delta}_n := \delta_n + \sqrt{\log K / n_{\min}}$. Then, for \( \lambda^{\nu + 1} \asymp \delta_N \widetilde{\delta}_n + \|\alpha_K - \Pi_{K} \alpha_K\|_{L^2_K(X)}\), we have
 \begin{align*}
 \|\mathcal{T}_{K}(\widehat{\beta}_K(\lambda) -  \beta_{K})\|_{L^2_K(Z)} &= O_p\left(\delta_N + \left\{ \delta_N \widetilde{\delta}_n + \|\alpha_K - \Pi_{K} \alpha_K\|_{L^2_K(X)} \right\}^{\frac{2\nu +1}{2\nu +2}}\right)\\
     \|\widehat{\beta}_K(\lambda) -  \beta_{K}\|_{L^2_K(Z)} &= O_p\left(\delta_N\left\{ \delta_N \widetilde{\delta}_n + \|\alpha_K - \Pi_{K} \alpha_K\|_{L^2_K(X)} \right\}^{\frac{-1}{2\nu +2}} \right) \\
     & \quad + O_p\left(\left\{ \delta_N \widetilde{\delta}_n + \|\alpha_K - \Pi_{K} \alpha_K\|_{L^2_K(X)} \right\}^{\frac{2\nu}{2\nu +2}}\right).
 \end{align*} 
 \label{theorem::weakstrongrates}
\end{theorem}
\begin{proof}[Proof of Theorem  \ref{theorem::weakstrongrates}]
    From Lemma \ref{lemma::approxstrongident}, we  have
    \[
    \|\mathcal{T}_{K}(\beta_{K}(\lambda) - \beta_{K}) \|_{L^2_K(Z)} = O(\lambda^{\nu + \frac{1}{2}}) \|w_{K}\|_{L^2_K(X)} + O(\lambda^{-1/2}) \|\alpha_K - \Pi_{K} \alpha_K\|_{L^2_K(X)}.
    \] 
    From Theorem \ref{theorem::estimationerror}, we have
    \begin{align*}
\|\mathcal{T}_{K}(\widehat{\beta}_K(\lambda)-\beta_{K}(\lambda))\|_{L^2_K(Z)}& = \delta_N \left( 1 + \frac{\widetilde{\delta}_n}{\sqrt{\lambda}}  + \sqrt{\lambda} \right)\\
\|\widehat{\beta}_K(\lambda)-\beta_{K}(\lambda)\|_{L^2_K(X)} &= \lambda^{-1/2} \delta_N \left( 1 + \frac{\widetilde{\delta}_n}{\sqrt{\lambda}}  + \sqrt{\lambda} \right)
\end{align*} 
Combining these bounds and applying the triangle inequality, we have
\begin{align*}
      \|\mathcal{T}_{K}(\widehat{\beta}_K(\lambda)-\beta_{K}) \|_{L^2_K(Z)} 
      &\lesssim \lambda^{-1/2} \delta_N \left( 1 + \frac{\widetilde{\delta}_n}{\sqrt{\lambda}} \right) 
      +  O(\lambda^{\nu + \frac{1}{2}}) 
      + O(\lambda^{-1/2}) \|\alpha_K - \Pi_{K} \alpha_K\|_{L^2_K(X)} \\
       \sqrt{\lambda}  \|\widehat{\beta}_K(\lambda)-\beta_{K}\|_{L^2_K(Z)} 
       &\lesssim \lambda^{-1/2} \delta_N \left( 1 + \frac{\widetilde{\delta}_n}{\sqrt{\lambda}} \right) 
       +  O(\lambda^{\nu + \frac{1}{2}}) 
       + O(\lambda^{-1/2}) \|\alpha_K - \Pi_{K} \alpha_K\|_{L^2_K(X)}.
\end{align*}

The right-hand side of the weak norm bound is minimized when \( \lambda^{\nu + 1} \asymp \delta_N \widetilde{\delta}_n + \|\alpha_K - \Pi_{K} \alpha_K\|_{L^2_K(X)}\). For this choice, we find
 $$\|\mathcal{T}_{K}(\widehat{\beta}_K(\lambda) -  \beta_{K})\|_{L^2_K(Z)} = O_p\left(\delta_N + \left\{ \delta_N \widetilde{\delta}_n + \|\alpha_K - \Pi_{K} \alpha_K\|_{L^2_K(X)} \right\}^{\frac{2\nu +1}{2\nu +2}}\right).$$
 For this choice, we find the strong norm can be bounded as 
 \begin{align*}
     \|\widehat{\beta}_K(\lambda) -  \beta_{K}\|_{L^2_K(Z)} &= O_p\left(\delta_N\left\{ \delta_N \widetilde{\delta}_n + \|\alpha_K - \Pi_{K} \alpha_K\|_{L^2_K(X)} \right\}^{\frac{-1}{2\nu +2}} \right) \\
     & \quad + O_p\left(\left\{ \delta_N \widetilde{\delta}_n + \|\alpha_K - \Pi_{K} \alpha_K\|_{L^2_K(X)} \right\}^{\frac{2\nu}{2\nu +2}}\right).
 \end{align*} 
\end{proof}

\begin{theorem}
\label{cor::estimationerrorsup}
Under the conditions of Theorem \ref{theorem::weakstrongratessup}, we have
\begin{align*}
\|\mathcal{T}_{K}(\widehat{\beta}_K(\lambda)-\beta_{K}(\lambda)\|_{L^2_K(Z)}& \lesssim N^{-\frac{\gamma}{2\gamma+1}} \left( 1 + \frac{n^{-\frac{\gamma}{2\gamma+1}}}{\sqrt{\lambda}}  + \sqrt{\lambda} \right)\\
\|\widehat{\beta}_K-\beta_{K,\lambda}\|_{L^2_K(X)} &\lesssim  \lambda^{-1/2} N^{-\frac{\gamma}{2\gamma+1}} \left( 1 + \frac{n^{-\frac{\gamma}{2\gamma+1}}}{\sqrt{\lambda}}  + \sqrt{\lambda} \right).
\end{align*} 
\end{theorem}

 \begin{proof}[Proof of Theorem \ref{cor::estimationerrorsup}]
  Note that the conditions of Lemma \ref{lemma::approxstrongident} and Theorem \ref{theorem::estimationerror}  hold under the conditions of Theorem Theorem \ref{cor::estimationerrorsup}. We apply Theorem \ref{theorem::estimationerror} using the sup-norm entropy integral to upper bound the critical radii of the function classes. To this end, we claim that the sup-norm covering numbers \(N_{\infty}(\varepsilon, \mathcal{G})\) and \(N_{\infty}(\varepsilon, \widehat{\mathcal{G}})\) are upper bounded by \(N_{\infty}(\varepsilon, \mathcal{F})\), up to a constant. To see this, observe that for any \(f_1, f_2 \in \mathcal{F}\),
\(\|\mathcal{T}_{K}f_1 - \mathcal{T}_{K}f_2\|_{\infty} \lesssim \|f_1 - f_2\|_{\infty}\)
and
\(\|\widehat{\mathcal{T}}_{K,v}f_1 - \widehat{\mathcal{T}}_{K,v}f_2\|_{\infty} \lesssim \|f_1 - f_2\|_{\infty}\),
since the average of functions is upper bounded by the supremum of these functions. It follows that the sup-norm covering numbers of \(\{\mathcal{T}_{K}f: f \in \mathcal{F}\}\) and \(\{\widehat{\mathcal{T}}_{K,v}f: f \in \mathcal{F}\}\) are upper bounded by \(N_{\infty}(\varepsilon, \mathcal{F})\). Furthermore, the function classes \(\mathcal{G}\) and \(\widehat{\mathcal{G}}\) are obtained through Lipschitz transformations of \(\mathcal{F}\) and the aforementioned function classes. Thus, by Theorem 2.10.20 of \cite{vanderVaartWellner}, we conclude that
\(\mathcal{J}_{\infty}(\delta, \mathcal{G}) \lesssim \mathcal{J}_{\infty}(\delta, \mathcal{F})\)
and
\(\mathcal{J}_{\infty}(\delta, \widehat{\mathcal{G}}) \lesssim \mathcal{J}_{\infty}(\delta, \mathcal{F})\).

    For any function class $\mathcal{F}$ that is totally bounded in supremum norm, Dudley's entropy integral (Lemma 4 of \cite{van2024combining}) establishes that, for any $\delta \geq N^{-1/2}$,
    $E_K[\widehat{\mathcal{R}}_N(\mathcal{F}, \delta)] \lesssim N^{-1/2}\mathcal{J}_{\infty}(\delta, \mathcal{F})$. Let $\delta_N$ and $\delta_n$ be the smallest values that satisfy the critical inequalities  $N^{-1/2}\mathcal{J}_{\infty}(\delta_N, \mathcal{F}) \leq \delta_N^2$ and $n^{-1/2}\mathcal{J}_{\infty}(\delta_n, \mathcal{F}) \leq \delta_n^2$, respectivelly. Then, by \ref{cond::regularityOnActionSpace}, using that $\mathcal{J}_{\infty}(\delta, \mathcal{F}) \lesssim \delta^{1-1/(2\alpha)}$, we have that $\delta_N = O(N^{-\frac{\gamma}{2\gamma+1}})$ and $\delta_n = O(n^{-\frac{\gamma}{2\gamma+1}})$. It can be verified that $ N^{-1/2} \sqrt{\log \log \delta_{N}} =  o(\delta_{N})$ and $ n_k^{-1/2} \sqrt{\log \log \delta_{n_k}} =  o(\delta_{n_k})$. Since $\alpha > 1/2$ by assumption, the conditions of Theorem \ref{theorem::estimationerror} hold. We conclude that
    \begin{align*}
\|\mathcal{T}_{K}(\widehat{\beta}_K(\lambda)-\beta_{K}(\lambda)\|_{L^2_K(Z)}& \lesssim N^{-\frac{\gamma}{2\gamma+1}} \left( 1 + \frac{n^{-\frac{\gamma}{2\gamma+1}} + \sqrt{\log K / n_{\min}}}{\sqrt{\lambda}}  + \sqrt{\lambda} \right)\\
\|\widehat{\beta}_K-\beta_{K,\lambda}\|_{L^2_K(X)} &\lesssim  \lambda^{-1/2} N^{-\frac{\gamma}{2\gamma+1}} \left( 1 + \frac{n^{-\frac{\gamma}{2\gamma+1}} + \sqrt{ \log K / n_{\min}}}{\sqrt{\lambda}}  + \sqrt{\lambda} \right).
\end{align*} 
    
\end{proof}

\subsection{Proofs for Section \ref{sec::nuisance}}

\begin{proof}[Proof of Theorem \ref{theorem::weakstrongratessup}]
    This is a direct application of Theorem \ref{theorem::weakstrongrates} using the estimation rates of Theorem \ref{cor::estimationerrorsup}. 
\end{proof}

\begin{proof}[Proof of Theorem \ref{theorem::weakstrongratessupprimary}]
The proof of Theorem \ref{theorem::weakstrongratessupprimary} is analogous to that of Theorem \ref{theorem::weakstrongratessup}, using regularization bias bounds in Appendix \ref{appendix::primalreg}. In particular, the estimation error bounds in Theorem \ref{theorem::estimationerror} also hold for the primal problem. The rate without the source condition (\(\nu = 0\)) was established in Theorem 1 of \cite{bibaut2024nonparametric}.  
\end{proof}

\section{Efficiency theory}

\subsection{Notation and definitions}

We follow the setup of Section~3.11 in \citet{vanderVaartWellner}. Throughout, we index by \( K \), with the total sample size \( N := N(K) := \sum_{k=1}^{K} n_k \) growing implicitly with \( K \). Recall that we observe \( N \) data units \( O^{(K)} = \{O_{ki} : 1 \leq i \leq n_k,\ 1 \leq k \leq K\} \), drawn jointly from a distribution \( P^{(K)} \), where each unit \( O_{ki} = (Z_{ki} = k, X_{ki}, Y_{ki}) \in \mathcal{O} \). The observations \( O_{ki} \) are mutually independent but not identically distributed. Let \( (\mathcal{O}^K, \mathcal{B}(\mathcal{O}^K)) \) denote the Borel measurable space on which \( P^{(K)} \) is defined, corresponding to the product space \( \mathcal{O}^K := \bigotimes_{k=1}^{K} \bigotimes_{i=1}^{n_k} \mathcal{O} \). Our statistical model \( \mathcal{P}^{(K)} \) consists of all distributions on \( \mathcal{O}^K \), dominated by a common product measure \( \nu^{(K)} \) induced by measures \( \nu_k^{(K)} \) defined on the support of $O_{ki}$, that satisfy the independence assumptions described in Section~\ref{sec:setup}.

Let \( P_k^{(K)} \) denote the marginal distribution of \( O_{ki} = (Z_{ki} = k, X_{ki}, Y_{ki}) \) under \( P^{(K)} \), which does not depend on \( i \). Let \( p_k^{(K)} := \frac{dP_k^{(K)}}{d\nu_k^{(K)}} \) and \( p^{(K)} := \frac{dP^{(K)}}{d\nu^{(K)}} \) denote the corresponding densities. Then, for each realization \( o^{(K)} \) of \( O^{(K)} \), we can write, by mutual independence, $p^{(K)}(o^{(K)}) = \prod_{k=1}^{K} \prod_{i=1}^{n_k} p_k^{(K)}(o_{ki}),$
and the log-likelihood satisfies
\begin{align}
    \log p^{(K)}(o^{(K)}) = \sum_{k=1}^K \sum_{i=1}^{n_k} \log p_k^{(K)}(o_{ki}). \label{eqn::loglik}
\end{align}

We now recall the definition of a regular parametric submodel. Let \( \{P_t^{(K)} : |t| \leq \delta\} \subset \mathcal{P}^{(K)} \) be a one-dimensional parametric submodel through \( P^{(K)} \in \mathcal{P}^{(K)} \), with corresponding densities \( p_t^{(K)} = dP_t^{(K)} / d\nu  \). The submodel is said to be \emph{quadratic mean differentiable} (QMD) at \( t = 0 \) if there exists a measurable function \( s_{P^{(K)}} \in L^2(P^{(K)}) \), called the score function, such that
\[
\int_{\mathcal{O}^K} \left( \sqrt{p_t^{(K)}(o^{(K)})} - \sqrt{p^{(K)}(o^{(K)})} - \tfrac{t}{2} s_{P^{(K)}}(o^{(K)}) \sqrt{p^{(K)}(o^{(K)})} \right)^2 d\nu^{(K)}(o^{(K)}) = o(t^2) \quad \text{as } t \to 0.
\]
We refer to \( s_{P^{(K)}} \) as the score of the submodel at \( t = 0 \), and define the Fisher information as \( I_{P^{(K)}} := \int s_{P^{(K)}}^2(o^{(K)}) \, dP^{(K)}(o^{(K)}) \). We review standard results on quadratic mean differentiability for factorizable models  \citep{bickel1993efficient}. The submodel density admits the factorization $p_t^{(K)}(o^{(K)}) = \prod_{k=1}^K \prod_{i=1}^{n_k} p_{k,t}^{(K)}(o_{ki}),$
and the score function admits a unique decomposition $s_{P^{(K)}}(o^{(K)}) = \sum_{k=1}^K \sum_{i=1}^{n_k} s_{k,P^{(K)}}(o_{ki}),$
where each \( s_{k,P^{(K)}} \in L^2_0(P_k^{(K)}) \), since independence induces an orthogonal decomposition of the model tangent space. Moreover, each factor submodel $t \mapsto p_{k,t}^{(K)}$ is QMD with score function \( s_{k,P^{(K)}} \), satisfying
\[
\int_{\mathcal{O}} \left( \sqrt{p_{k,t}^{(K)}(o)} - \sqrt{p_k^{(K)}(o)} - \tfrac{t}{2} s_{k,P^{(K)}}(o) \sqrt{p_k^{(K)}(o)} \right)^2 d\nu_k^{(K)}(o) = o(t^2) \quad \text{as } t \to 0.
\]
In fact, the converse also holds: the full density \( p_t^{(K)} \) is QMD at \( t = 0 \) if and only if each of the factor densities \( p_{k,t}^{(K)} \) is QMD at \( t = 0 \).

We showed in Theorem~\ref{theorem::EIF} that the parameter \( \Psi^{(K)}: \mathcal{P}^{(K)} \to \mathbb{R} \) is \emph{pathwise differentiable} at \( P^{(K)} \in \mathcal{P}^{(K)} \), with efficient influence function (EIF) \( D_{P^{(K)}}^* \in L_0^2(P^{(K)}) \). In this setting, the EIF \( D_{P^{(K)}}^* \) is defined as a function of the full observed data vector \( O^{(K)} \), rather than of a single observation. However, due to independence across data units, the EIF admits a decomposition of the form
\[
D_{P^{(K)}}^*(O^{(K)}) = \frac{1}{N} \sum_{k=1}^K \sum_{i=1}^{n_k} \varphi^*_{P^{(K)}}(O_{ki}),
\]
where each summand \( \varphi^*_{P^{(K)}} \in L_0^2(P_k^{(K)}) \) has mean zero and finite variance under the distribution \( P_k^{(K)} \) of \( O_{ki} \). This decomposition follows directly from Theorem~\ref{theorem::EIF}.


\subsection{Existence of Fisher-normalizable sequence of least favorable submodels}

\label{appendix:existenceleastfavorable}

We begin by establishing the existence of a least favorable sequence of submodels that satisfies Definition~\ref{def:leastfavorablemain}. To this end, we use the following construction of a QMD submodel, which we refer to as the Le Cam--Hájek QMD path, as it is directly motivated by the definition of quadratic mean differentiability.

\begin{lemma}[Le Cam--Hájek QMD path for the $k$th density factor]
Let $p_k^{(K)}$ be the marginal density of $P^{(K)}$ with respect to $\nu_k^{(K)}$, and let $s_{k,P^{(K)}} \in L_0^2(P_k^{(K)})$ with Fisher information $I_{k,P^{(K)}} := \int s_{k,P^{(K)}}^2(o)\,dP_k^{(K)}(o) < \infty$. Define the path
\[
p_{k,t}^{(K)} := \frac{\left( \sqrt{p_k^{(K)}} + \tfrac{t}{2}\,s_{k,P^{(K)}}\,\sqrt{p_k^{(K)}} \right)^2}{1 + \tfrac{t^2}{4}I_{k,P^{(K)}}}.
\]
Then $\{p_{k,t}^{(K)} : |t| \le \delta\}$ is QMD at $t = 0$ with score $s_{k,P^{(K)}}$ and satisfies
\[
\left\| \sqrt{p_{k,t}^{(K)}} - \sqrt{p_k^{(K)}} - \tfrac{t}{2}\,s_{k,P^{(K)}}\,\sqrt{p_k^{(K)}} \right\|_{L^2(\nu_k^{(K)})}^2
= \left( 1 - \sqrt{1 + \tfrac{t^2}{4}I_{k,P^{(K)}}} \right)^2 \leq I_{k,P^{(K)}}^2 t^4.
\] \label{lemma:hajeklecampath}
\end{lemma}
\begin{proof}
Let $q(o):=\sqrt{p_k^{(K)}(o)}$ and write $q_t(o)\;:=\;q(o)+\frac{t}{2}\,s_{k,P^{(K)}}(o)\,q(o).$
Since 
\[
\int q(o)^2\,d\nu_k^{(K)}(o)=1,\qquad
\int q(o)\bigl[s_{k,P^{(K)}}(o)\,q(o)\bigr]\,d\nu_k^{(K)}(o)
=\int s_{k,P^{(K)}}(o)\,p_k^{(K)}(o)\,d\nu_k^{(K)}(o)=0,
\]
and $\int\bigl[s_{k,P^{(K)}}(o)\,q(o)\bigr]^2\,d\nu_k^{(K)}(o)
=\int s_{k,P^{(K)}}(o)^2\,p_k^{(K)}(o)\,d\nu_k^{(K)}(o)
=I_{k,P^{(K)}},$
we have
\[
Z_{k,t}\;:=\;\int q_t(o)^2\,d\nu_k^{(K)}(o)
\;=\;\int\Bigl[q(o)+\tfrac{t}{2}s_{k,P^{(K)}}(o)\,q(o)\Bigr]^2\,d\nu_k^{(K)}
\;=\;1+\frac{t^2}{4}I_{k,P^{(K)}}.
\]
Hence $q_t^2/Z_{k,t}$ is a valid density on $\mathcal O$ for $|t| < \infty$.  Define
\[
p_{k,t}^{(K)}(o)\;=\;\frac{q_t(o)^2}{Z_{k,t}}
=\frac{\bigl[q(o)+\tfrac{t}{2}s_{k,P^{(K)}}(o)\,q(o)\bigr]^2}
{\,1+\tfrac{t^2}{4}I_{k,P^{(K)}}\,}\,.
\]
Then $p_{k,0}^{(K)}=p_k^{(K)}$.  Moreover, 
\[
\log p_{k,t}^{(K)}(o)
=2\log\bigl[q(o)+\tfrac{t}{2}s_{k,P^{(K)}}(o)\,q(o)\bigr]
-\log\Bigl(1+\tfrac{t^2}{4}I_{k,P^{(K)}}\Bigr).
\]
Differentiating at $t=0$ gives
\[
\left.\frac{\partial}{\partial t}\right|_{t=0}\log p_{k,t}^{(K)}(o)
=2\,\frac{\tfrac12\,s_{k,P^{(K)}}(o)\,q(o)}{q(o)}-0
=s_{k,P^{(K)}}(o).
\]
Thus the score at $t=0$ is $s_{k,P^{(K)}}$.  Next, 
\[
\sqrt{p_{k,t}^{(K)}(o)}
=\frac{q_t(o)}{\sqrt{Z_{k,t}}}
=\frac{q(o)+\tfrac{t}{2}s_{k,P^{(K)}}(o)\,q(o)}{\sqrt{1+\tfrac{t^2}{4}I_{k,P^{(K)}}}}.
\]
Define the pointwise remainder
\[
R_{k,t}(o)
:=\sqrt{p_{k,t}^{(K)}(o)}-q(o)-\frac{t}{2}\,s_{k,P^{(K)}}(o)\,q(o).
\]
Then
\begin{align*}
R_{k,t}(o)
&=\frac{q(o)+\tfrac{t}{2}s_{k,P^{(K)}}(o)\,q(o)}{\sqrt{1+\tfrac{t^2}{4}I_{k,P^{(K)}}}}
-q(o)-\frac{t}{2}\,s_{k,P^{(K)}}(o)\,q(o)\\
&=\Bigl(q(o)+\tfrac{t}{2}s_{k,P^{(K)}}(o)\,q(o)\Bigr)
\Bigl(\tfrac{1}{\sqrt{1+\tfrac{t^2}{4}I_{k,P^{(K)}}}}-1\Bigr).
\end{align*}
Hence 
\[
\bigl\lVert R_{k,t}\bigr\rVert_{L^2(\nu_k^{(K)})}^2
=\int R_{k,t}(o)^2\,d\nu_k^{(K)}(o)
=\Bigl(\tfrac{1}{\sqrt{1+\tfrac{t^2}{4}I_{k,P^{(K)}}}}-1\Bigr)^2
\int\Bigl[q(o)+\tfrac{t}{2}s_{k,P^{(K)}}(o)\,q(o)\Bigr]^2\,d\nu_k^{(K)}(o).
\]
Since $\int[q_t(o)]^2\,d\nu_k^{(K)}=Z_{k,t}=1+\tfrac{t^2}{4}I_{k,P^{(K)}}$, it follows that
\[
\bigl\lVert R_{k,t}\bigr\rVert_{L^2(\nu_k^{(K)})}^2
=\Bigl(\tfrac{1}{\sqrt{1+\tfrac{t^2}{4}I_{k,P^{(K)}}}}-1\Bigr)^2
\Bigl(1+\tfrac{t^2}{4}I_{k,P^{(K)}}\Bigr)
=\Bigl(1-\sqrt{1+\tfrac{t^2}{4}I_{k,P^{(K)}}}\Bigr)^2.
\]
Since $|1 - \sqrt{1+x}|\leq x$, we have 
\[
\bigl\lVert R_{k,t}\bigr\rVert_{L^2(\nu_k^{(K)})}^2
=\Bigl(1-\sqrt{1+\tfrac{t^2}{4}I_{k,P^{(K)}}}\Bigr)^2  \leq t^4  I_{k,P^{(K)}}^2,
\]
which completes the proof.
\end{proof}

We now show that the sequence of least favorable submodels based on the Le Cam--Hájek QMD path is Fisher-normalizable.

\begin{lemma}[Fisher-normalizability of Le Cam--Hájek least favorable submodels]
Let \( P_t^{(K)} \) denote the product submodel through \( P^{(K)} \), defined by factorizing the joint density as
\[
\frac{dP_t^{(K)}}{d\nu^{(K)}}(O^{(K)}) = \prod_{k=1}^K \prod_{i=1}^{n_k} p_{k,t}^{(K)}(O_{ki}),
\]
where, for each \( 1 \leq k \leq K \), the path \( t \mapsto p_{k,t}^{(K)} \) is the Le Cam--Hájek QMD submodel through the marginal density \( p_k^{(K)} \) of \( P^{(K)} \), with score function \( o_{ki} \mapsto \varphi_{P^{(K)}}^*(o_{ki}) \). Suppose
\[
\max_{1 \leq k \leq K} \operatorname{Var}_{P^{(K)}}\left( \varphi_{P^{(K)}}^*(O_{ki}) \right) = o(N).
\]
Then, for any \( \delta > 0 \), the sequence of submodels \( \{P_t^{(K)} : |t| \leq \delta\}_{K=1}^\infty \) is a Fisher-normalizable sequence of least favorable submodels.
\end{lemma}
\begin{proof}
By Lemma~\ref{lemma:hajeklecampath}, the Le Cam--Hájek submodel is QMD with remainder
\[
\left\| \sqrt{p_{k,t}^{(K)}} - \sqrt{p_k^{(K)}} - \tfrac{t}{2}\,s_{k,P^{(K)}}\,\sqrt{p_k^{(K)}} \right\|_{L^2(\nu_k^{(K)})}^2
=  O(I_{k,K}^2 t^4),
\]
where \( I_{k,K} := \operatorname{Var}_{P^{(K)}}(\varphi_{P^{(K)}}^*(O_{ki})) \). Relabeling \( t := t I_K^{-1/2} \), we find
\[
\left\| \sqrt{p_{k,t I_K^{-1/2}}^{(K)}} - \sqrt{p_k^{(K)}} - \tfrac{t}{2}\, I_K^{-1/2} \varphi_{P^{(K)}}^*\,\sqrt{p_k^{(K)}} \right\|_{L^2(\nu_k^{(K)})}^2
= \frac{I_{k,K}^2}{I_K^2} O(t^4).
\]
Averaging over \( k \) and noting that \( I_K = \sum_{k=1}^K n_k I_{k,K} \), we obtain, for each fixed \( |t| \leq \delta \), as \( K \rightarrow \infty \),
\begin{align*}
   \sum_{k=1}^K \frac{n_k}{N} \left\| \sqrt{p_{k,t I_K^{-1/2}}^{(K)}} - \sqrt{p_k^{(K)}} - \tfrac{t}{2}\, I_K^{-1/2} \varphi_{P^{(K)}}^*\,\sqrt{p_k^{(K)}} \right\|_{L^2(\nu_k^{(K)})}^2
&=  O(t^4) \sum_{k=1}^K \frac{n_k}{N}  \frac{I_{k,K}^2}{I_K^2}  \\
&=  O(t^4)  I_K^{-1} \cdot \frac{\sum_{k=1}^K \frac{n_k}{N} I_{k,K}^2}{N \sum_{k=1}^K \frac{n_k}{N} I_{k,K}} \\
&\leq O(t^4) I_K^{-1} \cdot \frac{\max_{1 \leq k \leq K} I_{k,K}}{N} \\
&= o(I_K^{-1}),
\end{align*}
since by assumption \( \max_{1 \leq k \leq K} I_{k,K} / N = o(1) \). Hence, the submodel is a Fisher-normalizable sequence of least favorable submodels.
\end{proof}

\subsection{Local asymptotic normality along least favorable submodels}

The following theorem establishes local asymptotic normality (LAN) along a sequence of Fisher-normalizable least favorable submodels as \( K \to \infty \) and \( N := N(K) \to \infty \). The proof closely follows Theorem~7.2 in \citet{van2000asymptotic}, which establishes LAN in the i.i.d.\ case. To generalize to the many-weak-instruments regime, it relies on the Fisher-normalizable property in Definition~\ref{def:leastfavorablemain}, which ensures that the quadratic mean differentiability (QMD) remainder vanishes after normalization by the total Fisher information.

\begin{theorem}[LAN for Fisher-normalizable least favorable submodels]
\label{theorem::LAN}
Let \( \{(P_t^{(K)} : |t| \leq \delta)\}_{K=1}^\infty \) be a Fisher-normalizable sequence of least favorable submodels satisfying Definition~\ref{def:leastfavorablemain} through \( P^{(K)} \). Suppose the following hold as \( K \to \infty \):
\begin{enumerate}
\item[(i)] \( I_K :=  \sum_{k=1}^K \sum_{i=1}^{n_k} \operatorname{Var}_{P^{(K)}}(\varphi_{P^{(K)}}^*(O_{ki})  ) \to \infty \);
\item[(ii)] \( N := \sum_{k=1}^K n_k \) satisfies \( N/I_K = O(1) \);
\item[(iii)] (Lindeberg condition)
\[
\frac{1}{I_K} \sum_{k=1}^K \sum_{i=1}^{n_k} 
E_{P^{(K)}}\left[ \varphi^{*2}_{P^{(K)}}(O_{ki}) \cdot 
\mathbf{1}\left\{ \varphi^{*2}_{P^{(K)}}(O_{ki}) > \tfrac{I_K}{t^2} \varepsilon^2 \right\} \right] \to 0
\quad \text{for all } \varepsilon > 0.
\]
\end{enumerate}
Then the submodel is locally asymptotically normal, in the sense that
\[
\log \frac{dP_{t I_K^{-1/2}}^{(K)}(O^{(K)})}{dP^{(K)}(O^{(K)})} 
= t\,I_K^{-1/2} \sum_{k=1}^K \sum_{i=1}^{n_k} \varphi_{P^{(K)}}^*(O_{ki}) 
- \frac{t^2}{2} + o_p(1),
\quad \text{as } K \to \infty.
\]
Moreover, by the Lindeberg central limit theorem,
$$\log \frac{p_t^{(K)}(O^{(K)})}{p^{(K)}(O^{(K)})} \overset{d}{\to} \mathcal{N}\left( -\tfrac{t^2}{2},\, t^2 \right).$$
\end{theorem}
\begin{proof}
Following the proof of Theorem 7.2 in \citet{van2000asymptotic}, define the random variable
\[
W_{ki} := 2 \left[\sqrt{\frac{p_{k,\,t I_K^{-1/2}}^{(K)}(O_{ki})}{p_k^{(K)}(O_{ki})}} - 1\right].
\]
Then the log-likelihood ratio admits the expansion
\begin{equation}
    \log \frac{p_{t}^{(K)}(O^{(K)})}{p^{(K)}(O^{(K)})} = 2 \sum_{k=1}^K \sum_{i=1}^{n_k} \log\left(1 + \tfrac{1}{2} W_{ki}\right). \label{eqn::loglikratio}
\end{equation}
 
We aim to show that, as \( K \to \infty \),
\begin{equation}
\sum_{k=1}^K \sum_{i=1}^{n_k} W_{ki} = t\,I_K^{-1/2} \sum_{k=1}^K \sum_{i=1}^{n_k} \varphi_{P^{(K)}}^*(O_{ki}) - \frac{t^2}{4} + o_p(1). \label{eqn::wki-expansion}
\end{equation}
To prove this, it suffices to show that both the mean and the variance of the difference
\[
\sum_{k=1}^K \sum_{i=1}^{n_k} W_{ki} - t\,I_K^{-1/2} \sum_{k=1}^K \sum_{i=1}^{n_k} \varphi_{P^{(K)}}^*(O_{ki}) - \frac{t^2}{4}
\]
converge to zero. We begin with the variance:
\begin{align*}
\operatorname{Var}_{P^{(K)}}\left(\sum_{k=1}^K \sum_{i=1}^{n_k} W_{ki} - t\,I_K^{-1/2} \sum_{k=1}^K \sum_{i=1}^{n_k} \varphi_{P^{(K)}}^*(O_{ki})\right)
&\le \sum_{k=1}^K n_k\,\operatorname{Var}_{P_k^{(K)}}\left(W_{ki} - t\,I_K^{-1/2} \varphi_{P^{(K)}}^*(O_{ki})\right) \\
&= \sum_{k=1}^K n_k \int \left\{\sqrt{\frac{p_{k,\,t I_K^{-1/2}}^{(K)}(o)}{p_k^{(K)}(o)}} - 1 - \tfrac{t}{2} I_K^{-1/2} \varphi_{P^{(K)}}^*(o)\right\}^2 dP_k^{(K)}(o).
\end{align*}
Denote the upper bound by \(\text{RHS}\). Switching to integration with respect to \(\nu_k^{(K)}\), we obtain:
\begin{align*}
\text{RHS}
&= \sum_{k=1}^K n_k \int \left\{ \sqrt{p_{k,\,t I_K^{-1/2}}^{(K)}(o)} - \sqrt{p_k^{(K)}(o)} - \tfrac{t}{2} I_K^{-1/2} \varphi_{P^{(K)}}^*(o) \sqrt{p_k^{(K)}(o)} \right\}^2 d\nu_k^{(K)}(o) \\
&= \frac{N}{I_K} \sum_{k=1}^K \frac{n_k}{N} I_K \int \left\{ \sqrt{p_{k,\,t I_K^{-1/2}}^{(K)}(o)} - \sqrt{p_k^{(K)}(o)} - \tfrac{t}{2} I_K^{-1/2} \varphi_{P^{(K)}}^*(o) \sqrt{p_k^{(K)}(o)} \right\}^2 d\nu_k^{(K)}(o).
\end{align*}
By assumption, \(N / I_K = O(1)\) as \(K \to \infty\). Moreover, by Fisher-normalizability of the least favorable submodel sequence (Property 2 of Definition~\ref{def:leastfavorablemain}), the weighted average term is \(o(1)\). Hence,
\[
\text{RHS} = O(1) \cdot o(1) = o(1),
\]
so the variance vanishes.

Turning to the expectation, we compute:
\begin{align*}
E_{P^{(K)}}\left[ \sum_{k=1}^K \sum_{i=1}^{n_k} W_{ki} \right] 
&= 2 \sum_{k=1}^K n_k \left( \int \sqrt{p_{k,\,t I_K^{-1/2}}^{(K)}(o)} \sqrt{p_k^{(K)}(o)}\, d\nu_k^{(K)}(o) - 1 \right) \\
&= - \sum_{k=1}^K n_k \int \left( \sqrt{p_{k,\,t I_K^{-1/2}}^{(K)}(o)} - \sqrt{p_k^{(K)}(o)} \right)^2 d\nu_k^{(K)}(o) \\
&= - \frac{N}{I_K} \sum_{k=1}^K \frac{n_k}{N} I_K \int \left( \sqrt{p_{k,\,t I_K^{-1/2}}^{(K)}(o)} - \sqrt{p_k^{(K)}(o)} \right)^2 d\nu_k^{(K)}(o).
\end{align*}
 Applying the second-order QMD expansion (again from Property 2 of Definition~\ref{def:leastfavorablemain}), we have:
\begin{align*}
\sum_{k=1}^K \frac{n_k}{N} I_K \int \left\{ \sqrt{p_{k,\,t I_K^{-1/2}}^{(K)}} - \sqrt{p_k^{(K)}} \right\}^2 d\nu_k^{(K)}
&= \sum_{k=1}^K \frac{n_k}{N} \int \left( \tfrac{t}{2} \varphi_{P^{(K)}}^* \sqrt{p_k^{(K)}} \right)^2 d\nu_k^{(K)} + o(1) \\
&= \frac{t^2}{4} \sum_{k=1}^K \frac{n_k}{N} \operatorname{Var}_{P_k^{(K)}}(\varphi_{P^{(K)}}^*(O_{ki})) + o(1) \\
&= \frac{t^2}{4} \cdot \frac{I_K}{N} + o(1).
\end{align*}
Substituting this into the previous display and using that $\frac{N}{I_K} = O(1)$ gives:
\[
E_{P^{(K)}}\left[ \sum_{k=1}^K \sum_{i=1}^{n_k} W_{ki} \right] = -\frac{t^2}{4} + o(1).
\]
We conclude that both the mean and variance of the difference converge to zero, completing the proof of \eqref{eqn::wki-expansion}.

We now return to equation~\eqref{eqn::loglikratio}. Following equation (7.5) in \citet{van2000asymptotic}, a Taylor expansion yields $\log (1 + x) = x - \frac{1}{2} x^2 + x^2 R(2x),$
where \( R(x) \to 0 \) as \( x \to 0 \). Hence,
\begin{align}
    \log \frac{p_{t}^{(K)}(O^{(K)})}{p^{(K)}(O^{(K)})} 
    &= 2 \sum_{k=1}^K \sum_{i=1}^{n_k} \log\left(1 + \tfrac{1}{2} W_{ki}\right) \nonumber \\
    &= \sum_{k=1}^K \sum_{i=1}^{n_k} W_{ki} 
    - \frac{1}{4} \sum_{k=1}^K \sum_{i=1}^{n_k} W_{ki}^2 
    + \frac{1}{2} \sum_{k=1}^K \sum_{i=1}^{n_k} W_{ki}^2 R(W_{ki}). \label{eqn::LANalmost}
\end{align}
By Property~2 of Definition~\ref{def:leastfavorablemain}, we can write
\[
I_K W_{ki}^2 =  t^2 \varphi_{P^{(K)}}^{*2}(O_{ki}) + A_{ki},
\]
for a random variable \( A_{ki} \) such that $\sum_{k=1}^K \frac{n_k}{N} E_{P^{(K)}}[\,|A_{ki}|\,] \longrightarrow 0 \quad \text{as } K \to \infty.$
By linearity of expectation, it follows that the sum \( \frac{1}{N} \sum_{k=1}^K \sum_{i=1}^{n_k} A_{ki} \) converges to zero in expectation, and hence in probability. Thus,  as \( K \rightarrow \infty \), we have that
\begin{align*}
     \sum_{k=1}^K \sum_{i=1}^{n_k} W_{ki}^2 
     &= t^2 I_K^{-1} \sum_{k=1}^K \sum_{i=1}^{n_k} \varphi_{P^{(K)}}^{*2}(O_{ki}) 
     + \frac{N}{I_K} \cdot \frac{1}{N} \sum_{k=1}^K \sum_{i=1}^{n_k} A_{ki} \\
     &= t^2 I_K^{-1} \sum_{k=1}^K \sum_{i=1}^{n_k} \varphi_{P^{(K)}}^{*2}(O_{ki}) 
     + O_p(1) o_p(1),
\end{align*}
where in the final equality we use the assumption that \( \frac{N}{I_K} = O_p(1) \). Thus, by the law of large numbers for sums of independent random variables, it holds that
\begin{align*}
     \sum_{k=1}^K \sum_{i=1}^{n_k} W_{ki}^2  = t^2  \sum_{k=1}^K \sum_{i=1}^{n_k} \frac{\varphi_{P^{(K)}}^{*2}(O_{ki})}{I_K^{-1} }  + o_p(1) \longrightarrow_p t^2.
\end{align*}

Arguing again as in \citet{van2000asymptotic}, and applying the triangle inequality and Markov's inequality,
\begin{align*}
P\left( \max_{ki} |W_{ki}| > \varepsilon \sqrt{2} \right) 
&\leq \sum_{k=1}^K \sum_{i=1}^{n_k}  P^{(K)}\left( t^2 \varphi_{P^{(K)}}^{*2}(O_{ki}) > I_K \varepsilon^2 \right)  
+ \sum_{k=1}^K \sum_{i=1}^{n_k} P^{(K)}\left( |A_{ki}| > I_K \varepsilon^2 \right) \\
&\leq \varepsilon^{-2} \cdot \frac{t^2}{I_K} \sum_{k=1}^K \sum_{i=1}^{n_k} 
E_{P^{(K)}}\left[ \varphi_{P^{(K)}}^{*2}(O_{ki}) \cdot 
\mathbf{1}\left\{ \varphi_{P^{(K)}}^{*2}(O_{ki}) > \tfrac{I_K}{t^2} \varepsilon^2 \right\} \right] \\
&\quad + \varepsilon^{-2} \cdot \frac{1}{I_K} \sum_{k=1}^K \sum_{i=1}^{n_k} 
E_{P^{(K)}}\left[ |A_{ki}| \right] \\
&= \varepsilon^{-2} \cdot \frac{t^2}{I_K} \sum_{k=1}^K \sum_{i=1}^{n_k} 
E_{P^{(K)}}\left[ \varphi_{P^{(K)}}^{*2}(O_{ki}) \cdot 
\mathbf{1}\left\{ \varphi_{P^{(K)}}^{*2}(O_{ki}) > \tfrac{I_K}{t^2} \varepsilon^2 \right\} \right] 
+ \varepsilon^{-2} \cdot \frac{1}{I_K} \sum_{k=1}^K \sum_{i=1}^{n_k} 
E_{P^{(K)}}\left[ |A_{ki}| \right] \\
&= \varepsilon^{-2} \cdot \frac{t^2}{I_K} \sum_{k=1}^K \sum_{i=1}^{n_k} 
E_{P^{(K)}}\left[ \varphi_{P^{(K)}}^{*2}(O_{ki}) \cdot 
\mathbf{1}\left\{ \varphi_{P^{(K)}}^{*2}(O_{ki}) > \tfrac{I_K}{t^2} \varepsilon^2 \right\} \right] 
+ \varepsilon^{-2} \cdot o_p(1),
\end{align*}
where we use that \( \frac{N}{I_K} = O_p(1) \),  apply the Lindeberg condition, and use that $ \frac{1}{N} \sum_{k=1}^K \sum_{i=1}^{n_k} 
E_{P^{(K)}}\left[ |A_{ki}| \right] = o_p(1)$. 
In the proof of Theorem 7.2 of \citet{van2000asymptotic}, it is shown that \( \max_{ki} |W_{ki}| = o_p(1) \) implies \( \frac{1}{2} \sum_{k=1}^K \sum_{i=1}^{n_k} W_{ki}^2 R(W_{ki}) = o_p(1) \). Combining this with the fact that \( \sum_{k=1}^K \sum_{i=1}^{n_k} W_{ki}^2 \rightarrow_p t^2 \), and using the asymptotic expansion for \( \sum_{k=1}^K \sum_{i=1}^{n_k} W_{ki} \) in \eqref{eqn::wki-expansion}, we plug into \eqref{eqn::LANalmost} to obtain
\begin{align*}
    \log \frac{p_{t}^{(K)}(O^{(K)})}{p^{(K)}(O^{(K)})} 
    &= t\,I_K^{-1/2} \sum_{k=1}^K \sum_{i=1}^{n_k} \varphi_{P^{(K)}}^*(O_{ki}) 
    - \frac{t^2}{4} - \frac{t^2}{4} + o_p(1) \\
    &= t\,I_K^{-1/2} \sum_{k=1}^K \sum_{i=1}^{n_k} \varphi_{P^{(K)}}^*(O_{ki}) 
    - \frac{t^2}{2} + o_p(1),
\end{align*}
as desired.

\end{proof}

\subsection{Proofs for Section \ref{sec:efficiency}}

\begin{proof}[Proof of Theorem \ref{theorem::LANmain}]
    The conditions of Theorem \ref{theorem::LANmain} imply those of Theorem \ref{theorem::LAN}. The result then follows directly from Theorem \ref{theorem::LAN}.
\end{proof}

\begin{proof}[Proof of Theorem \ref{theorem::convolution}]
The theorem is a direct consequence of the general convolution theorem for locally asymptotically normal experiments, as established in \citet[Theorems 25.22--25.23]{van2000asymptotic}. We outline the key steps of the proof here.
Define
\[
T_K = \sqrt{N / \sigma_K^{*2}}\,\bigl(\widehat\psi_K - \Psi^{(K)}(P^{(K)})\bigr),
\qquad
Z_K = I_K^{-1/2} \sum_{k,i} \varphi^*_{P^{(K)}}(O_{ki}),
\]
and let \(\Lambda_K(u) = \log \tfrac{dP^{(K)}_{uI_K^{-1/2}}}{dP^{(K)}}\). By Theorem~\ref{theorem::LANmain}, we have
\[
\Lambda_K(u) = u Z_K - \tfrac{u^2}{2} + o_{P^{(K)}}(1),
\quad
Z_K \xrightarrow{d} \mathcal{N}(0,1),
\]
which implies that the sequence \(\{P^{(K)}_{uI_K^{-1/2}}\}\) is contiguous with respect to \(P^{(K)}\).

Regularity (Definition~\ref{def:regularity}) ensures that for each fixed \(u \in \mathbb{R}\),
\[
T_K - u 
= \sqrt{N / \sigma_K^{*2}}\,\bigl(\widehat\psi_K - \Psi^{(K)}(P^{(K)}_{uI_K^{-1/2}})\bigr) 
\xrightarrow{d} L 
\quad \text{under } P^{(K)}_{uI_K^{-1/2}},
\]
for some limiting distribution \(L\) that does not depend on \(u\). Hence, under \(P^{(K)}\),
\[
(T_K - u,\; \Lambda_K(u)) \xrightarrow{d} (L,\; uZ - \tfrac{u^2}{2}),
\]
and by Le Cam’s third lemma (see van der Vaart, 1998, Thm.~25.22--25.23), it follows that
\[
(T_K,\; Z_K) \xrightarrow{d} (L + Z,\; Z),
\quad
Z \sim \mathcal{N}(0,1),\; Z \perp L.
\]
Projecting onto the first coordinate yields \(T_K \xrightarrow{d} L + Z\). Since \(Z\) is independent of \(L\), we have \(\operatorname{Var}(T_K) = \operatorname{Var}(L) + 1 \ge 1\), with equality if and only if \(L \equiv 0\), i.e., \(T_K \xrightarrow{d} \mathcal{N}(0,1)\).
\end{proof}

\begin{proof}[Proof of Theorem \ref{theorem::regularitymain}]

Let \(O^{(K)}=\{O_{ki}:1\le k\le K,\;1\le i\le n_k\}\), \(N=\sum_{k=1}^K n_k\), and
\[
\bar\varphi_K
=\sum_{k=1}^K\sum_{i=1}^{n_k}\varphi^*_{P^{(K)}}(O_{ki}).
\]
By the asymptotically linear expansion, multiplying both sides by \(\sqrt{N/\sigma_K^{*2}}\) gives
\[
\sqrt{\tfrac{N}{\sigma_K^{*2}}}\bigl(\widehat\psi_K - \Psi^{(K)}(P^{(K)})\bigr)
= I_K^{-1/2}\,\bar\varphi_K + o_{P^{(K)}}(1).
\]
Next, by Theorem~\ref{theorem::LANmain}, for each fixed \(u\), the log-likelihood ratio admits the LAN expansion
\[
\log\frac{dP_u^{(K)}}{dP^{(K)}}(O^{(K)})
= u\,\frac{\bar\varphi_K}{\sqrt{I_K}} - \frac{u^2}{2} + o_{P^{(K)}}(1).
\]
Since \(\bar\varphi_K/\sqrt{I_K} \xrightarrow{d} \mathcal{N}(0,1)\) under \(P^{(K)}\), this expansion---together with the QMD property and the central limit theorem---implies that the shifted laws \(P_{uI_K^{-1/2}}^{(K)}\) are contiguous with respect to \(P^{(K)}\) \citep{van2000asymptotic}. Therefore, by Le Cam’s third lemma,
\[
\frac{\bar\varphi_K}{\sqrt{I_K}} \;\xrightarrow{d}\; \mathcal{N}(u,1)
\quad \text{under } P_{uI_K^{-1/2}}^{(K)}.
\tag{1}\label{LCT}
\]
By \ref{cond::effpathwise}, 
\(\tfrac{d}{du}\Psi^{(K)}(P_u^{(K)})\big|_{u=0}=I_K/N\), 
so a Taylor expansion gives for fixed \(t\)
\[
\Psi^{(K)}\bigl(P_{tI_K^{-1/2}}^{(K)}\bigr)
=\Psi^{(K)}(P^{(K)})+t\,\frac{\sqrt{I_K}}{N}+o\!\bigl(t \sigma_K^* N^{-1/2}\bigr).
\]
Hence under \(P_{tI_K^{-1/2}}^{(K)}\), we have
\[
\sqrt{\tfrac{N}{\sigma_K^{*2}}}\bigl(\widehat\psi_K-\Psi^{(K)}(P_{tI_K^{-1/2}}^{(K)})\bigr)
=I_K^{-1/2}\,\bar\varphi_K - t + o(1)
\;\xrightarrow{d}\;\mathcal \mathcal{N}(0,1),
\]
where the final convergence follows from \eqref{LCT}.  
\end{proof}

\end{document}